%% file: main.tex
\documentclass[11pt]{article}

\input{usepackages}

\input{commands}

\input{mytheorems}


\title{Covert Quantum Learning: \\
Privately and Verifiably Learning from Quantum Data}

\author[1]{Abhishek Anand\footnote{\href{mailto:abhi@caltech.edu}{abhi@caltech.edu}}}
\author[2]{Matthias C. Caro\footnote{\href{mailto:matthias.caro@warwick.ac.uk}{matthias.caro@warwick.ac.uk}}}
\author[3]{Ari Karchmer\footnote{\href{mailto:akarchmer0@gmail.com}{akarchmer0@gmail.com}}}
\author[4]{Saachi Mutreja\footnote{\href{mailto:sm5540@columbia.edu}{sm5540@columbia.edu}}}
\affil[1]{
California Institute of Technology, Pasadena, CA, USA}
\affil[2]{
University of Warwick, Coventry, UK}
\affil[3]{Morgan Stanley Machine Learning Research, New York, NY, USA}
\affil[4]{Columbia University, New York, NY, USA}
\date{}

\begin{document}

\pagenumbering{gobble}

\maketitle

\begin{abstract}

    Quantum learning from remotely accessed quantum compute and data must address two key challenges: 
    verifying the correctness of data and ensuring the privacy of the learner’s data-collection strategies and resulting conclusions. The \emph{covert (verifiable) learning} model of Canetti and Karchmer (TCC 2021) provides a framework for endowing classical learning algorithms with such guarantees, protecting against computationally bounded adversaries who observe and tamper with public oracle queries. However, their framework has two drawbacks: it relies on computational hardness assumptions and does not flexibly accommodate richer data-access models, such as quantum ones.
    
    In this work, we propose models of covert verifiable learning in quantum learning theory and realize them without computational hardness assumptions for remote data access scenarios motivated by established quantum data advantages.
    We consider two 
    privacy notions: (i) \emph{strategy-covertness}, where the eavesdropper does not gain information about the learner's strategy; and (ii) \emph{target-covertness}, where the eavesdropper does not gain information about the unknown object being learned. We show:
    \begin{itemize}
        \item Strategy-covert algorithms for making quantum statistical queries via classical shadows;  
        \item Target-covert algorithms for: 
        \begin{itemize}
            \item learning quadratic functions from public quantum examples and private quantum statistical queries;
            \item Pauli shadow tomography and stabilizer state learning from public multi-copy and private single-copy quantum measurements;
            \item solving Forrelation and Simon's problem from public quantum queries and private classical queries, where the adversary is a unidirectional or i.i.d.\ ancilla-free eavesdropper.
        \end{itemize} 
    \end{itemize}
    The lattermost results in particular establish that the exponential separation between classical and quantum queries for Forrelation and Simon’s problem survives under covertness constraints.
    Along the way, we design 
    covert verifiable protocols for quantum data acquisition from public quantum queries which may be of independent interest. 
    Overall, our models and corresponding algorithms demonstrate that quantum advantages are privately and verifiably achievable even with untrusted, remote data. 
\end{abstract}

\newpage
\tableofcontents
\newpage

\pagenumbering{arabic}

\section{Introduction}

As quantum technologies progress \cite{arute2019quantum, zhong2020quantum, scholl2021quantum, ebadi2021quantum, sivak2023real, google2023suppressing, bluvstein2023logical, google2025quantum, putterman2025hardware, brock2025quantum} beyond the so-called NISQ era \cite{preskill2018nisq, preskill2025megaquop}, quantum devices may advance scientific discovery in physics through enhanced experimental and simulation capabilities. 
However, with only few powerful quantum devices available, most researchers will only be able to access them remotely and indirectly by asking a quantum service provider to run quantum computations or simulations on their behalf and report back the results.
In such a setting of delegated data collection, two demands become paramount: to verify the correctness of the data and/or the conclusions drawn from it, and to ensure that painstakingly devised experimental setups and the resulting conclusions remain confidential. 
When quantum compute \emph{as well as} quantum data are accessed remotely, verified quantum computing \cite{fitzsimonsUnconditionallyVerifiableBlind2017, gheorghiu2018verification, broadbent2018howtoverify, mahadev2022classical-verification} by itself does not suffice to address these challenges.

Similar challenges arise in other scientific disciplines.
Examples include the use of proprietary datasets, where access is mediated by a third party, or instructing a field research team which specific observations to collect.

\paragraph{Covert Learning.}
To address these situations, the \textit{covert learning} model was introduced by \cite{canetti2021covert}. The covert learning model considers the task of learning an unknown concept using membership oracle queries, while said oracle queries and the oracle responses are monitored by an adversary. We refer to the 
oracle queries and responses together as the \textit{learning transcript}.
The (informal) objectives of a \textit{covert learning algorithm} are:
\begin{enumerate}
    \item \textbf{Efficient Learning:} The algorithm should efficiently use the oracle access to the unknown concept to  
    obtain good learning guarantees (e.g., as phrased in agnostic PAC-learning \cite{vapnik1971uniform, valiant1984theory}). 
    \item \textbf{Covertness:} The algorithm should prevent the adversary from 
    gaining information about the unknown concept or any prior knowledge used by the algorithm (e.g., a chosen hypothesis class).
\end{enumerate} 

Formulating a framework akin to ``real/ideal'' simulation security in secure computation \cite{goldwasser1982probabilistic-encryption}, \cite{canetti2021covert} define covertness as follows: There exists a probabilistic polynomial-time simulator that samples an ``ideal'' distribution over learning transcripts that is computationally indistinguishable from the ``real'' distribution over learning transcripts generated in executing the covert learning algorithm. 
Crucially, the simulator must function given only access to random examples from the unknown concept (and \textit{not} access to a membership oracle nor the underlying hypothesis class).\footnote{As noted by \cite{canetti2021covert}, allowing the simulator no access to the unknown concept whatsoever (no random examples, say), while tempting, is not possible, since at the very least the adversary will gain access to some oracle responses in the real distribution.}
This ensures that the learning transcript reveals no more than what can be efficiently extracted from only random examples. 
Consequently, whenever the unknown concept cannot be learned efficiently from random examples, then no efficient adversary can learn the concept from the transcript produced by a covert learning algorithm. Additionally, no information about a chosen hypothesis class can be leaked, since the existence of the simulator implies that the oracle queries made by the algorithm are indistinguishable across hypothesis classes.\footnote{Note that this guarantee is only nontrivial when operating in an agnostic model of learning (as opposed to a  realizable setting), where the choice of hypothesis class determines the optimal error of the learned hypothesis.}

\paragraph{More Broadly Applicable Models of Covert Learning Needed.}
The 
\cite{canetti2021covert} model comes with two caveats: First, the learner must make computational hardness assumptions, and the adversary has to be computationally bounded. These are \textit{necessary} assumptions: 
\textit{meaningful} covert learning algorithms in the \cite{canetti2021covert} model only exist for learning problems where learning with random examples is computationally hard.\footnote{To see this, observe that without the hardness assumption the problem becomes trivial, as the learner can just query randomly, which implies a trivial simulator.} 
Second, the \cite{canetti2021covert} model is designed specifically to leverage supposed separations between learning with classical membership queries and classical random examples. 
However, given the plethora of oracles considered in classical and quantum learning theory---such as (quantum) examples \cite{vapnik1971uniform, valiant1984theory, bshouty1995learning-DNF}, (quantum) statistical queries \cite{kearns1998efficient, arunachalam2020quantumstatisticalquerylearning}, or copies of a quantum state---more broadly applicable models for covert learning are needed.
%
%
Furthermore, while the \cite{canetti2021covert} model is theoretically appealing and has applications to \textit{undetectable} machine learning model stealing \cite{karchmer2023theoretical}\footnote{Covert learning algorithms can be used to \textit{fool} polynomial-time detection algorithms 
that mount a preventative defense against model stealing adversaries.}, 
it has not seen much progress in terms of new algorithms.
This lack of progress may be attributable to the necessity of computational assumptions and the strong requirements imposed by the definition. 

To pave the way for progress in the design of covert learning algorithms in general and in quantum learning theory in particular, we thus view the following question as central:

\begin{center}
    \textit{Is efficient quantum learning with meaningful covertness possible without computational assumptions?}
\end{center}


In a nutshell, the purpose of this work is to explore and answer this question. In particular, we can answer it in the affirmative in different quantum learning-theoretic scenarios.

\subsection{Our Contributions}

As conceptual contributions, we introduce new covert (verifiable) learning models for a variety of quantum learning scenarios and argue that these constitute operationally meaningful notions of covertness. 
Here, while the  \cite{canetti2021covert} model combines two aspects of privacy in one definition---preventing the adversary from gaining information regarding (i) the unknown concept and (ii) the query strategy used by the learning algorithm---we will treat (i) and (ii) as independently relevant aspects of covertness. We refer to (i) as \emph{target-covertness} and (ii) as \emph{strategy-covertness}. 

As technical contributions, we show that these models can be realized without computational hardness assumptions for several quantum data oracles. In particular, we design strategy-covert protocols for implementing quantum statistical queries and target-covert protocols for learning from quantum examples and measurement data, both against computationally unbounded adversaries. We also develop target-covert verifiable protocols for obtaining quantum phase states against certain physically constrained adversaries. We leverage these to show that classically hard problems can be solved covertly and verifiably by a learner with remote access to quantum data.

\subsubsection{New Covert Learning Models}\label{sbsct:new_models}

\paragraph{Covert statistical queries.} 
Our first new model focuses on \emph{strategy-covertness}, where the learner aims to hide their analysis strategy rather than the underlying data itself. Before introducing this model in the quantum setting, we begin with a classical warm-up. Consider the classical statistical query (SQ) model, where a learner has access to a public oracle that provides estimates for statistical properties of an unknown function $f$. Specifically, the learner can submit a query function $q$ and receive an approximate value for $\mathbb{E}[q(x, f(x))]$. The privacy risk is that an adversary observing the sequence of submitted query functions $\{q_i\}$ could infer the learner's goals in their analysis (i.e., their \textit{strategy}). The focus on this privacy issue is motivated by a setting where a client conducts data analysis through a cloud server that implements the SQ oracle. In this case, the client may be interested in maintaining privacy of their desired analyses, while preventing information leakage about the underlying concept is essentially unimportant since the server owns the data.

To model this scenario, we consider a compilation scheme consisting of a pair of efficient algorithms: an \emph{encoder} ($\mathcal{E}$) and a \emph{decoder} ($\mathcal{D}$). The learner's true, sensitive SQs are first passed to the encoder, which transforms them into a set of encoded public SQs. After receiving the oracle's responses, the decoder processes them to recover accurate answers to the original, private queries. Informally, a protocol in this model must satisfy three properties:
\begin{itemize}
    \item \textbf{Completeness:} The learner must still get the right answers. After the public queries are made and the responses are decoded, the resulting values must be (approximately) valid answers to the learner's original SQs, with high probability. This ensures the protocol remains useful for learning.
    \item \textbf{Privacy:} The adversary must learn nothing about the learner's strategy. The distribution of the encoded queries sent to the public oracle must be statistically indistinguishable from a distribution generated by a simulator that has no knowledge of the learner's original queries. This ensures the public transcript is independent of the learner's private interests.
    \item \textbf{Efficiency:} The protocol must not be prohibitively expensive. The number of public queries, the runtime of the encoder and decoder, and any degradation in the precision of the answers must all be bounded by a polynomial in the relevant parameters.
\end{itemize}

This definition (see \Cref{definition:covert-sq-model} for a formal version) can be understood as describing an efficient compilation of few SQs into not-too-many SQs in such a way that little information about the original SQs is leaked. As a simple proof-of-principle, fix a query function $q$ which is thought of as a multivariate polynomial. By publicly querying for the monomials in $q$ and post-processing, one can approximate the true response to $q$ while keeping the monomial coefficients hidden. Indeed, this strategy reveals which monomials are in $q$, but the strategy extends to hide the identities of the monomials themselves via a random sketching approach (see \Cref{subsec:covert-poly-sq} for details).

\emph{The Quantum Case.} The idea behind the above classical covert statistical query model is to compile few queries to an oracle of type A into not-too-many queries to an oracle of type B in such a way that the encoded queries do not hold any information about the original queries. 
This immediately generalizes to other scenarios.
Quantumly, we consider efficiently compiling few quantum statistical queries (QSQs), a quantum generalization of classical SQs introduced in \cite{arunachalam2020quantumstatisticalquerylearning}, into not-too-many randomized measurements to be performed on an unknown quantum state, all while leaking little information about the actual QSQs of interest.
To formalize this, recall that a QSQ oracle $\mathsf{O}^{\mathrm{QSQ}}$ for a (mixed) quantum state $\rho$
can be queried on a pair $(M,\tau)$ consisting of an observable $M$ 
with operator norm at most $1$ and a tolerance parameter $\tau > 0$, and outputs some $v\in\mathbb{R}$ such that $|v - \mathrm{tr}[M\rho]|\leq\tau$.\footnote{As in the case of SQs, an algorithm that queries a QSQ oracle has to succeed for any implementation of the QSQ oracle that provides answers in a way that respects the tolerance.}
In contrast to the QSQ oracle, which outputs approximate expectation values, we use $\mathsf{O}^{\mathrm{QMeasEx}}(\rho)$ for a (mixed) quantum state $\rho$ to denote a quantum measurement example oracle, which, when queried on a measurement, outputs a sample from the distribution obtained when performing said measurement on copies of $\rho$. This oracle access is natural in scenarios where one party does not have direct access to quantum capabilities, thus they send descriptions of experiments to be executed to another party (the oracle) and ask for the measurement outcomes observed in the experiment.
Now, we obtain the following quantum version of the classical covert SQ model for compiling QSQs for observables from some class $\mathcal{M}$ into quantum measurements:

\begin{definition}[Covert Quantum Statistical Query Model---Informal]\label{inf-definition:covert-qsq-model}
    A pair $(\mathcal{E},\mathcal{D})$ of (possibly randomized) algorithms, where $\mathcal{E}$ encodes the $m$ original QSQs $(\vec{M},\vec{\tau})$ from class $\mathcal{M}$ into $m^{(e)}$ positive operator-valued (POVM) measurements $\vec{E}^{(e)}$, and where $\mathcal{D}$ decodes $((\vec{M}, \vec{\tau}), \vec{E}^{(e)},\vec{v}^{(e)})$ to an $m$-dimensional real vector $\vec{v}$, where the $\vec{v}^{(e)}$ is a vector of measurement outcomes obtained when measuring $\vec{E}^{(e)}$, is a \emph{covert QSQ algorithm for $\mathcal{M}$-queries} if it satisfies the following properties:
    \begin{itemize}
        \item \textbf{Completeness:} After decoding the measurement outcomes received for the encoded POVM measurements, the vector $\vec{v}$ contains (approximately) valid QSQ responses to the original QSQs $(\vec{M},\vec{\tau})$ with high success probability.
        \item \textbf{Privacy:} There exists a simulator that, given as input only the desired number $m^{(e)}$ of encoded queries and knowledge of $\mathcal{M}$, generates a distribution over encoded POVM measurements that is statistically indistinguishable from the distribution over $\vec{E}^{(e)}$, no matter the initial QSQs $(\vec{M},\vec{\tau})$.
        \item \textbf{Efficiency:} The increase from QSQ complexity to measurement sample complexity and the runtime of $\mathcal{E}$ and $\mathcal{D}$ are polynomial in the relevant parameters. 
    \end{itemize}
\end{definition}

\paragraph{Covert learning from examples.} 
We next focus on \emph{target-covertness} 
as well as work in a different data oracle setting. Again, our main contributions live in the quantum setting, but we will begin with a classical introduction and proof-of-principle demonstration.

The \cite{canetti2021covert} model considered public membership queries and private random examples. We weaken both oracles and work with a public random example oracle and a private SQ oracle.\footnote{Clearly, the resulting covert learning model will be meaningful only for problems that are easy to solve from random examples but hard to solve from SQs.}
This setting is motivated by scenarios in which a party has access to data with noisy labels: individual examples cannot be trusted, but aggregate statistics can. In particular, this allows them to emulate an SQ oracle privately. For example, consider a genomic research group that maintains a proprietary noisy dataset but now faces a learning task that requires noiseless i.i.d.~samples. The lab can ask a field team to collect these samples. While there is no hidden data collection strategy (only i.i.d.~sampling is requested), the lab may still wish to limit leakage about the data-generating process to any eavesdropper (e.g., the field team itself). In particular, they may wish to prevent such an eavesdropper from solving the learning task using the collected samples.

In contrast to \cite{canetti2021covert}, where covertness was inherent in the query strategy and a private oracle was only needed for verifiability, here, we obtain covertness by using the private SQ oracle. Moreover, because this setting admits \emph{unconditional} separations (as we will see in the following paragraph), our results do not require computational assumptions. 

Intuitively, a learner in this model uses a relatively small number of public queries to the example oracle, thus leaking little information to the eavesdropper, to ``kick-start'' a learning procedure that relies only on further private SQs. Let us illustrate this with a problem that is easy to solve from random examples but hard to solve from SQs, namely parity learning w.r.t.~the uniform distribution: A parity function $x\mapsto s\cdot x$, where $s$ is an unknown $n$-bit string and the inner product is modulo 2, can be learned computationally efficiently and exactly from $\mathcal{O}(n)$ random examples with high probability via simple Gaussian elimination. In contrast, the seminal work \cite{kearns1998efficient} showed that $2^{\Omega(n)}$ many SQs with inverse-polynomial tolerance are needed to achieve the same. When both random examples and SQs are available, the two resources can be traded off against each other: After seeing $n-k$ linearly independent random examples, $2^k$ SQs with constant tolerance suffice to find the true parity string among the remaining $2^k$ candidates consistent with the observed examples.
However, an eavesdropper observing only the public examples can do no better than uniform random guessing over the remaining $2^k$ options. Thus, this simple procedure motivates the following covert learning definition:



\begin{definition}[Covert Exact Learning From Public Examples and Private SQs---Informal]\label{inf-definition:covert-exact-learning-public-examples-private-sq}
    An algorithm $\mathcal{L}$ that has access to a private SQ oracle $\mathsf{O}_{\mathrm{pri}}^{\mathrm{SQ}}$ and to a public random example oracle $\mathsf{O}_{\mathrm{pub}}^{\mathrm{Ex}}$ is a $(m_{\mathrm{pri}}, m_{\mathrm{pub}},\delta_c,\delta_p)$-covert exact learner for a concept class $\mathcal{F}$ from private SQs and public examples if it satisfies the following:
    \begin{itemize}
        \item $\delta_c$-\textbf{Completeness:} For any $f\in\mathcal{F}$, after making at most $m_{\mathrm{pri}}$ queries to $\mathsf{O}_{\mathrm{pri}}^{\mathrm{SQ}}(f)$ and at most $m_{\mathrm{pub}}$ queries to $\mathsf{O}_{\mathrm{pub}}^{\mathrm{Ex}}(f)$, $\mathcal{L}$ outputs $f$ with success probability $\geq 1-\delta_c$.
        \item $\delta_p$-\textbf{Privacy:} For $F\sim \mathcal{F}$ a uniformly random concept, no adversary $\mathcal{A}$ can correctly guess $F$ with success probability $\geq \delta_p$ from the at most $m_{\mathrm{pub}}$ public random examples requested by $\mathcal{L}$. 
    \end{itemize}
\end{definition}

While this definition is phrased for specific public and private oracles, it generalizes immediately to other oracle settings. For instance, we may consider a version of \Cref{inf-definition:covert-exact-learning-public-examples-private-sq} with public quantum measurements and private QSQs; apart from this change in oracles, the correctness and covertness requirements remain the same (we motivate this choice of oracles in \Cref{sbsct:new_algos}). We note that, while the QSQ and QMeasEx oracles are defined in terms of an underlying quantum state, queries sent to and answers received from these oracles are entirely classical.

\begin{definition}[Covert Exact Learning From Public QMeasExs and Private QSQs---Informal]\label{inf-definition:covert-learning-public-quantum-examples-private-qsq}
    An algorithm $\mathcal{L}$ that has access to a private QSQ oracle $\mathsf{O}_{\mathrm{pri}}^{\mathrm{QSQ}}$ and to a public quantum measurement example oracle $\mathsf{O}_{\mathrm{pub}}^{\mathrm{QMeasEx}}$ is a $(m_{\mathrm{pri}}, m_{\mathrm{pub}},\delta_c,\delta_p)$-covert exact learner for a class of states $\mathcal{S}$ from private QSQs and public examples if it satisfies the following:
    \begin{itemize}
        \item $\delta_c$-\textbf{Completeness:} For any $\rho\in\mathcal{S}$, after making at most $m_{\mathrm{pri}}$ queries to $\mathsf{O}_{\mathrm{pri}}^{\mathrm{QSQ}}(\rho)$ and at most $m_{\mathrm{pub}}$ queries to $\mathsf{O}_{\mathrm{pub}}^{\mathrm{QMeasEx}}(\rho)$, $\mathcal{L}$ outputs $\rho$ with success probability $\geq 1-\delta_c$.
        \item $\delta_p$-\textbf{Privacy:} For $\rho\sim \mathcal{S}$ a uniformly random state from the class, no adversary $\mathcal{A}$ can correctly guess $\rho$ with success probability $\geq\delta_p$ from the at most $m_{\mathrm{pub}}$ many public quantum measurement examples requested by $\mathcal{L}$.
    \end{itemize}
\end{definition}

Note that the above definition can be phrased over a class of functions $\mathcal{F}$, rather than over a class of states $\mathcal{S}$, by fixing an encoding of a Boolean function $f$ into a quantum state $\rho_f$. In that case, a QSQ or QMeasEx query to $f$ is implemented as a query to $\mathsf{O}^{\mathrm{QSQ}}(\rho_f)$ or $\mathsf{O}^{\mathrm{QMeasEx}}(\rho_f)$, respectively. Arguably the most common such encoding for learning w.r.t.\ the uniform distribution—dating back to \cite{bshouty1995learning-DNF}—takes the form $\rho_f=\ket{\psi_f^{\mathrm{Ex}}}\bra{\psi_f^{\mathrm{Ex}}}$ with the quantum example state $\ket{\psi_f^{\mathrm{Ex}}}=2^{-n/2}\sum_{x\in\{0,1\}^n}\ket{x,f(x)}$. Another common encoding, which we will use extensively, is $\rho_f=\ket{\psi_f^{\mathrm{Ph}}}\bra{\psi_f^{\mathrm{Ph}}}$ with the quantum phase state defined as $\ket{\psi_f^{\mathrm{Ph}}}=2^{-n/2}\sum_{x\in\{0,1\}^n}(-1)^{f(x)}\ket{x}$ \cite{arunachalam2022optimal}.
Without writing out the informal definition, let us note that we will also later employ an analogue of \Cref{inf-definition:covert-learning-public-quantum-examples-private-qsq} in which the private oracle access consists of measurements performed on single copies of $\rho$ at a time whereas the public access corresponds to multi-copy measurements.


\paragraph{Covert quantum learning.} 
While the covert learning models above capture certain quantum learning scenarios, the oracle queries and responses considered are purely classical, despite relating to underlying quantum systems. 
If we consider a fully quantum learning scenario with oracle queries and/or responses given as quantum states, however, a fundamental rule of quantum physics comes into play: Quantum information is fragile and cannot simply be ``observed'' by an eavesdropper; the act of observing in general perturbs a quantum system. 
Thus, while the above definitions of covert learning with passive (observing) adversaries can be extended to \emph{covert verifiable learning} (following the nomenclature of \cite{canetti2021covert}) by adding a soundness requirement against malicious, meddling adversaries, such an extension becomes inherently necessary for covertness in fully quantum learning scenarios.  
Therefore, we develop a new model of covert verifiable learning to capture learning settings with quantum data. 

We will focus on a scenario common among many quantum algorithms that achieve quantum advantages in query complexity, such as Deutsch-Jozsa \cite{deutsch1992rapid}, Bernstein-Vazirani \cite{bernstein1997complexity}, Simon \cite{simon1997power}, or Forrelation \cite{aaronson2015forrelation}, and quantum computational learning theory \cite{bshouty1995learning-DNF, arunachalam2017survey}:
The quantum (learning) algorithm has access to an unknown function $f$ via a quantum membership query oracle $\mathsf{O}^{\mathrm{QMem}}(f)$ with action $\ket{x, y} \mapsto \ket{x, y \oplus f(x)}$. By querying $\mathsf{O}^{\mathrm{QMem}}(f)$ on a uniform superposition over all inputs, the learner can create a copy of the quantum example state $\ket{\psi_f^{\mathrm{Ex}}}$ and then use multiple such copies to learn about the unknown function. Alternatively, quantum function access is also often modeled by a so-called quantum phase oracle $\mathsf{O}^{\mathrm{QPh}}(f)$, acting as $\ket{x}\mapsto (-1)^{f(x)}\ket{x}$. When queried on a uniform superposition over all computational basis states, the oracle $\mathsf{O}^{\mathrm{QPh}}$ thus outputs a copy of the phase state $\ket{\psi_f^{\mathrm{Ph}}}$. As $\mathsf{O}^{\mathrm{QMem}}$ and $\mathsf{O}^{\mathrm{QPh}}$ can be related via the phase kickback trick, and as function states are unitarily equivalent to phase states for slightly modified functions, 
these quantum access models are often (but not always) equally powerful. 
Thus, in this overview of our results, we will not be too careful in distinguishing between these two oracle models; a more detailed discussion is provided in \Cref{sec:public-quantum-oracle-private-classical-queries}.

In our first \emph{quantum} covertness definition, we isolate the subroutine of preparing a quantum phase state $\ket{\psi_f^{\mathrm{Ph}}}$ using the public oracle $\mathsf{O}^{\mathrm{QPh}}_{\mathrm{pub}}$, and formalize what it means to achieve this covertly and with verifiability (i.e., soundness) guarantees against a quantum adversary who monitors and modifies the learner-oracle interaction. 
As discussed above, in this setting we cannot meaningfully phrase covertness requirements without also requiring soundness. We thus give the learner access to a private but weak(er) oracle, classical membership queries, which they use to ensure covertness and/or soundness.  
In fact, in the following definition we consider two possible versions of the covertness requirement:

\begin{definition}[Covert Verifiable Quantum Data From Public Quantum Oracle Queries---Informal]\label{inf-definition:covert-quantum-data-public-quantum-oracle-private-classical-queries}
    A quantum algorithm $\mathcal{L}$ that has access to a private classical membership query oracle $\mathsf{O}^{\mathrm{Mem}}_{\mathrm{pri}}$ and to a public quantum phase oracle $\mathsf{O}_{\mathrm{pub}}^{\mathrm{QPh}}$ is a $(m_{\mathrm{pri}}, m_{\mathrm{pub}},\delta_c,\delta_s, \varepsilon, (\delta_{\mathrm{leak}}))$-covert verifiable procedure for producing $m$ quantum phase state copies for a concept class $\mathcal{F}$ against adversary $\mathcal{A}$ if it satisfies the following:
    \begin{itemize}
        \item $(\varepsilon,\delta_c)$-\textbf{Completeness:} For any $f\in\mathcal{F}$, after making at most $m_{\mathrm{pri}}$ queries to $\mathsf{O}_{\mathrm{pri}}^{\mathrm{Mem}}(f)$ and at most $m_{\mathrm{pub}}$ queries to $\mathsf{O}_{\mathrm{pub}}^{\mathrm{QPh}}(f)$, $\mathcal{L}$ accepts and outputs a state $\rho$ such that $\bra{\psi_f^{\mathrm{Ph}}}^{\otimes m}\rho\ket{\psi_f^{\mathrm{Ph}}}^{\otimes m}\geq 1-\varepsilon$ with success probability $\geq 1-\delta_c$.
        \item $(\varepsilon,\delta_s)$-\textbf{Soundness:} For any $f\in\mathcal{F}$, after making at most $m_{\mathrm{pri}}$ queries to $\mathsf{O}_{\mathrm{pri}}^{\mathrm{Mem}}(f)$ and at most $m_{\mathrm{pub}}$ queries to $\mathsf{O}_{\mathrm{pub}}^{\mathrm{QPh}}(f)$, the latter of which are subject to corruption by the adversary $\mathcal{A}$, $\mathcal{L}$ accepts and outputs some $\rho$ with $\bra{\psi_f^{\mathrm{Ph}}}^{\otimes m}\rho\ket{\psi_f^{\mathrm{Ph}}}^{\otimes m}< 1-\varepsilon$ with failure probability $\leq \delta_s$.
        \item \textbf{{Privacy}:}
        For $F\sim\mathcal{F}$ a randomly drawn concept, \ldots 
        \begin{itemize}
            \item \textbf{Version 1:} \ldots the adversary gains no information about $F$, in the sense that the joint state $\rho_{\mathsf{FA}}$ of the classical function register and the adversary's (in general quantum) register factorizes as $\rho_{\mathsf{FA}}=\frac{\mathds{1}_{\mathsf{F}}}{|\mathcal{F}|}\otimes \rho_{\mathsf{A}}$.
            \item \textbf{Version 2:} \ldots if the adversary has a probability of at least $\delta_{\mathrm{leak}}$\footnote{Note that we allow $m_{\mathrm{pri}}$ and  $m_{\mathrm{pub}}$ to vary based on this target bound on the adversary’s information-extraction probability, see \Cref{sec:iid-both-directions}. Additionally, Version 2 implicitly restricts the adversaries that can satisfy the definition by requiring a lower bound on the per-round information extraction probability. } of extracting information about $F$ in every interaction, then $\mathcal{L}$ accepts with failure probability $\leq \delta_s$. 
        \end{itemize}
    \end{itemize}
\end{definition}


\Cref{inf-definition:covert-quantum-data-public-quantum-oracle-private-classical-queries} can be modified to covertly acquiring quantum example states from a public quantum membership query oracle, by simply replacing $\mathsf{O}_{\mathrm{pub}}^{\mathrm{QPh}}$ by $\mathsf{O}_{\mathrm{pub}}^{\mathrm{QMem}}$ and $\ket{\psi_f^{\mathrm{Ph}}}$ by $\ket{\psi_f^{\mathrm{Ex}}}$.
Also, Version 1 of the privacy requirement in \Cref{inf-definition:covert-quantum-data-public-quantum-oracle-private-classical-queries} allows for an approximate version by weakening the equality $\rho_{\mathsf{FA}}=\frac{\mathds{1}_{\mathsf{F}}}{|\mathcal{F}|}\otimes \rho_{\mathsf{A}}$ to a closeness requirement $\rho_{\mathsf{FA}}\approx_\varepsilon \frac{\mathds{1}_{\mathsf{F}}}{|\mathcal{F}|}\otimes \rho_{\mathsf{A}}$.\footnote{If we interpret the approximation guarantee $\approx_\varepsilon$ in terms of the trace distance, then this is analogous to how \cite[Definitions 5.2.1 and 5.2.2]{vidick2023introduction} phrase ``($\varepsilon$-)ignorance'' for quantum key distribution.}

Our \Cref{inf-definition:covert-quantum-data-public-quantum-oracle-private-classical-queries} and its modifications are naturally motivated by problems in which quantum phase or example states (possibly supported by classical queries) provide a quantum advantage over classical queries.
Reiterating the examples above, such quantum advantages are known in a variety of cases, for instance for the Deutsch-Jozsa problem ($1$ quantum phase state vs $2^{n-1}+1$ deterministic classical queries), the Bernstein-Vazirani problem ($1$ quantum phase state copy vs $\Theta(n)$ classical queries), the Forrelation problem ($\Theta(1)$ quantum phase states vs $\widetilde{\Omega}(2^{n/2})$ classical queries), and Simon's problem ($\mathcal{O}(n)$ quantum examples + $2$ classical queries vs $\Omega(2^{n/2})$ classical queries). 
In any of these tasks, a client capable of quantum processing but without direct access to quantum data may want to pull such data from a public quantum data repository, while certifying that the quantum data they obtain is high-fidelity and ensuring that no eavesdropper gained information about the data through monitoring the pull request. 
Covert verifiable quantum data acquisition as in \Cref{inf-definition:covert-quantum-data-public-quantum-oracle-private-classical-queries} thus has the potential to make quantum advantages attainable in a verifiable and private manner for learners who access quantum data remotely.

When viewing quantum data acquisition as a subroutine in quantumly solving a problem with oracle access, it becomes natural to phrase task-specific versions of \Cref{inf-definition:covert-quantum-data-public-quantum-oracle-private-classical-queries}. 
We use the Forrelation problem as a concrete example. First, we recall the Forrelation problem.

\begin{definition}[Forrelation \cite{aaronson2010bqp, aaronson2015forrelation}]
    The amount of ``Forrelation'' between two Boolean functions $f,g:\{0,1\}^n\to\{0,1\}$ is defined as
    \begin{equation}
        \Phi(f,g)
        \coloneqq \frac{1}{2^{3n/2}} \sum_{x,y\in\{0,1\}^n} (-1)^{f(x)+x\cdot y + g(y)}\, .
    \end{equation}
    In the Forrelation problem, given oracle access to $f$ and $g$, the task is to decide whether (i) $\lvert \Phi(f,g)\rvert\leq \frac{1}{100}$ or (ii) $\Phi(f,g)\geq \nicefrac{3}{5}$, promised that one of these two is the case.
\end{definition}

We can now phrase a version of \Cref{inf-definition:covert-quantum-data-public-quantum-oracle-private-classical-queries} that is aimed specifically at the Forrelation problem.

\begin{definition}[Covert Verifiable Forrelation From Public Quantum Oracle Queries---Informal]\label{inf-definition:covert-forrelation}
    A quantum algorithm $\mathcal{L}$ that has access to a private classical membership query oracle $\mathsf{O}^{\mathrm{Mem}}_{\mathrm{pri}}$ and to a public quantum phase oracle $\mathsf{O}_{\mathrm{pub}}^{\mathrm{QPh}}$ is a $(m_{\mathrm{pri}}, m_{\mathrm{pub}},\delta_c,\delta_s, (\delta_{\mathrm{leak}}))$-covert verifiable procedure for solving the Forrelation problem against adversary $\mathcal{A}$ if it satisfies the following:
    \begin{itemize}
        \item $\delta_c$-\textbf{Completeness:} For any $f, g:\{0,1\}^n\to\{0,1\}$ with either (i) $\lvert \Phi(f,g)\rvert\leq \frac{1}{100}$ or (ii) $\Phi(f,g)\geq \nicefrac{3}{5}$, after making at most $m_{\mathrm{pri}}$ queries to $\mathsf{O}_{\mathrm{pri}}^{\mathrm{Mem}}(f)$ and $\mathsf{O}_{\mathrm{pri}}^{\mathrm{Mem}}(g)$ and at most $m_{\mathrm{pub}}$ queries to $\mathsf{O}_{\mathrm{pub}}^{\mathrm{QPh}}(f)$ and $\mathsf{O}_{\mathrm{pub}}^{\mathrm{QPh}}(g)$, $\mathcal{L}$ accepts and correctly decides between (i) and (ii) with success probability $\geq 1-\delta_c$.
        \item $\delta_s$-\textbf{Soundness:} For any $f, g:\{0,1\}^n\to\{0,1\}$ with either (i) $\lvert \Phi(f,g)\rvert\leq \frac{1}{100}$ or (ii) $\Phi(f,g)\geq \nicefrac{3}{5}$, after making at most $m_{\mathrm{pri}}$ queries to $\mathsf{O}_{\mathrm{pri}}^{\mathrm{Mem}}(f)$ and $\mathsf{O}_{\mathrm{pri}}^{\mathrm{Mem}}(g)$ and at most $m_{\mathrm{pub}}$ queries to $\mathsf{O}_{\mathrm{pub}}^{\mathrm{QPh}}(f)$ and $\mathsf{O}_{\mathrm{pub}}^{\mathrm{QPh}}(g)$, the latter of which are subject to corruption by the adversary $\mathcal{A}$, $\mathcal{L}$ accepts and incorrectly decides between (i) and (ii) with failure probability $\leq \delta_s$.
        \item \textbf{Privacy:} Suppose the function pair $(F,G)$ is with probability $1/2$ drawn uniformly at random from all function pairs satisfying (i) and with probability $1/2$ drawn uniformly at random from all function pairs satisfying (ii). Then, \ldots 
        \begin{itemize}
            \item \textbf{Version 1:} \ldots no adversary $\mathcal{A}$ can correctly decide between (i) and (ii) with success probability $>1/2$ from the at most $2 m_{\mathrm{pub}}$ many public quantum queries made by $\mathcal{L}$.
            \item \textbf{Version 2:} \ldots if the adversary has a probability of at least $\delta_{\mathrm{leak}}$ of extracting information about $(F, G)$ in every interaction, then $\mathcal{L}$ accepts with failure probability $\leq \delta_s$.
        \end{itemize}
    \end{itemize}
\end{definition}

More generally, once we fix a task of interest (e.g., among the ones mentioned above) to be solved using quantum data, a corresponding variant of \Cref{inf-definition:covert-forrelation} formalizes what it means to solve the task in a covert verifiable manner from public quantum oracle access and private classical query access. 

We note that the settings of quantum learning from remote data access as in \Cref{inf-definition:covert-quantum-data-public-quantum-oracle-private-classical-queries,inf-definition:covert-forrelation} are incomparable to that of verified (blind) quantum computation (V(B)QC) \cite{fitzsimonsUnconditionallyVerifiableBlind2017, gheorghiu2018verification, broadbent2018howtoverify, mahadev2022classical-verification}: 
In V(B)QC, a mostly or even entirely \emph{classical} client delegates quantum computations to be performed on \emph{trusted} input data to an untrusted quantum server. 
In contrast, we consider a \emph{quantum} client that aims to process \emph{untrusted, public} data obtained by delegating data collection to a public quantum oracle. 
Consequently, as also pointed out in recent work on interactive proofs for learning \cite{goldwasser2021interactive, gur2024power, caro2024verification, ma2024classicalverificationquantumlearning, caro2024interactiveproofsverifyingquantum}, 
V(B)QC is insufficient to achieve the kind of covert and verifiable quantum data processing that we desire.
In fact, we propose that covert verifiable quantum data acquisition as in \Cref{inf-definition:covert-quantum-data-public-quantum-oracle-private-classical-queries} may be used in tandem with and complementarily to V(B)QC tools (compare \Cref{remark:combine-with-vqc}).    

\subsubsection{New Covert Learning Results}\label{sbsct:new_algos}

\paragraph{Covert statistical queries.}

Recall that in the previous section, we obtained strategy-covertness in the covert SQ model by querying a set of random polynomials and privately post-processing the responses to recover an estimate for the SQ response to the desired query function. We obtain a quantum analogue of this setting by employing the classical shadows framework \cite{huang2020predicting}. In this recently developed randomized measurement toolkit \cite{elben2023randomized}, one applies random (typically local) unitaries to a state and then measures; the resulting measurement outcomes can later be classically post-processed to estimate many properties of the state (e.g., expectation values of chosen observables). Operationally, data acquisition is public and query-agnostic (``measure first''), whereas the choice of observable and subsequent post-processing is private (``estimate later''). An eavesdropper who observes only the public measurement transcript therefore gains no information about which observable is ultimately estimated beyond what is implied by the randomized measurement design. Accordingly, the results below can be viewed as reinterpreting existing algorithms from \cite{huang2020predicting} as \emph{covert} QSQ protocols. 

\begin{observation}[Classical Shadows as Covert QSQ Algorithms---Informal] \label{obs:classical_shadows_intro}
    Classical shadows constitute covert QSQ-algorithms in the sense of \Cref{inf-definition:covert-qsq-model}. Namely:
    \begin{itemize}
        \item Classical shadows with random Pauli measurements is an efficient covert QSQ-algorithm for bounded $k$-local $n$-qubit QSQs, where the initial QSQ complexity of $m$ increases to an encoded measurement sample complexity of $\mathcal{O}(\log(m)4^k)$.
        \item Classical shadows with random Clifford measurements is a covert QSQ-algorithm for bounded rank-$r$ $n$-qubit QSQs, where the initial QSQ complexity of $m$ increases to an encoded measurement sample complexity of $\mathcal{O}(\log(m) 2^r)$. Here, only the encoding is efficient.
    \end{itemize}
\end{observation}

\paragraph{Covert learning from examples.}
As already observed when introducing \Cref{inf-definition:covert-exact-learning-public-examples-private-sq}, parity learning serves as a proof-of-principle demonstration for covert learning from public examples and private SQs:

\begin{observation}[Covert Parity Learning from Public Examples and Private SQs---Informal]\label{inftheorem:public-example-private-sq-parity}
    There exists a computationally efficient $(m_{\mathrm{pri}}=2/\delta_p,  m_{\mathrm{pub}} =n-\lceil\log(1/\delta_p)\rceil +\mathcal{O}(\log(1/\delta_c)), \delta_c, \delta_p)$-covert exact learner from private SQs and public random examples for $n$-bit parity functions w.r.t.~the uniform distribution.
\end{observation}

This tells us: While parity learning from SQ access only is hard, one can combine private SQ access with a small number of public random examples to efficiently solve the problem. And the number of public examples vs private SQs can be balanced to keep the guessing probability of any adversary inverse-polynomially small.
This serves as a simple example of how information-theoretic privacy guarantees can be achieved in learning from public examples when a weaker private oracle is available.

Parity learning was a prime candidate to consider in the context of \Cref{inf-definition:covert-exact-learning-public-examples-private-sq} precisely because it separates SQ learning from learning from examples \cite{kearns1998efficient}.
Recently, \cite{arunachalam2023role} proved a quantum counterpart of \cite{kearns1998efficient}'s seminal result. Namely, the class of quadratic Boolean functions $x\mapsto x^\top A x~(\mathrm{mod}~2)$ for (w.l.o.g.) upper-triangular $A\in\{0,1\}^{n\times n}$ can be learned computationally efficiently and exactly from measurements performed on $\mathcal{O}(n)$ copies of the corresponding quantum example state $\ket{\psi_A}$ with high success probability. In contrast, learning this class from access to a QSQ oracle for $\ket{\psi_A}$ requires $2^{\Omega(n)}$ many queries with inverse-polynomial tolerance. This makes quadratic Boolean function learning a natural problem to study in the context of \Cref{inf-definition:covert-learning-public-quantum-examples-private-qsq}. Our next result shows how to solve this learning problem covertly.

\begin{theorem}[Covert Quadratic Function Learning from Public QMeasExs and Private QSQs---Informal]\label{inftheorem:public-qmeasexample-private-qsq-quadratic}
    There exists a computationally efficient $(m_{\mathrm{pri}}=n,m_{\mathrm{pub}}=\mathcal{O}(n + \log(1/\delta_c)), \delta_c, \delta_p = 2^{-n})$-covert exact quantum learner from private QSQs and public quantum measurement examples for degree-$2$ polynomials over $n$ bits w.r.t.~the uniform distribution.
\end{theorem}

In contrast to the toy problem considered in \Cref{inftheorem:public-example-private-sq-parity}, the covert learner here, while using only linearly-in-$n$ many private oracle queries, nevertheless ensures that the adversary cannot correctly guess the unknown function with more than exponentially-small success probability. Thus, by ``diverting'' only a small number of oracle queries from the stronger public to the weaker private oracle, the learner can achieve a strong privacy guarantee.
As we show in \Cref{subsec:covert-learning-template}, the same strategy yields covert learning advantages for Pauli shadow tomography \cite{huang2021information} and stabilizer state learning \cite{aaronson2008identifying, montanaro2017learning} under public two-copy and private single-copy measurement example oracles. In these instances, jointly querying the public and private oracles reduces the query complexity below what is achievable with the private oracle alone, while preserving covertness. This demonstrates that certain exponential advantages of learning with vs without quantum memory can be made covert. Additional state, process, and Hamiltonian learning tasks fitting our template are discussed in \Cref{subsec:covert-learning-template}.

\paragraph{Covert quantum learning.}

We next present our results for covert quantum learning with a public quantum phase oracle and a private classical membership oracle. In \Cref{sec:public-quantum-oracle-private-classical-queries}, we first prove a no-go result: Covertness is information-theoretically impossible for quantum data acquisition for certain families of states. This necessitates physical restrictions on the adversary to achieve covertness. Thus, for this model, we give two quantum algorithms for covert verifiable quantum data acquisition (as defined in \Cref{inf-definition:covert-quantum-data-public-quantum-oracle-private-classical-queries}), each secure against a distinct adversary class.

\begin{theorem}[Covert Verifiable Phase States Against Unidirectional \footnote{By ``unidirectional'' adversaries, we mean eavesdroppers that can observe and possibly tamper with the oracle$\to$learner (response) channel, but have no access to the learner$\to$oracle (query) channel. The opposite notion (adversaries restricted to the learner$\to$oracle channel only) is trivially covert in our setting, since they never interact with any object that depends on the unknown function implemented by the oracle.} Adversaries---Informal]\label{inf-theorem:covert-quantum-data-acquisition-v1}
    There exists a computationally efficient $(m_{\mathrm{pri}}= \mathcal{O}(\poly(m,n, \log(1/\delta), 1/\varepsilon)), m_{\mathrm{pub}}=\mathcal{O}(\poly(m,n, 1/\delta, 1/\varepsilon)),\delta_c=0,\delta_s=\delta, \varepsilon)$-covert verifiable procedure for producing $m$ quantum phase state copies for an arbitrary concept class $\mathcal{F}$ that satisfies \textbf{Privacy Version 1} against arbitrary unidirectional quantum adversaries that observe and modify only the quantum communication from $\mathsf{O}_{\mathrm{pub}}^{\mathrm{QPh}}$ to the learner.
\end{theorem}

Note that the above protocol does not efficiently achieve an inverse-exponential in $n$ soundness error because $m_{\mathrm{pub}}$ scales with $1/\delta$.
However, when this subroutine is used to obtain phase states to solve a specific (say, distinguishing) task, \Cref{sec:non-iid-onlybackwards} shows how to amplify the success probability to obtain inverse-exponential soundness, as reflected in \Cref{inf-corollary:covert-forrelation-v1}. Moreover, even without the privacy guarantee, applying the the above theorem to many-vs-one distinguishing tasks offers insights into limitations of quantum data: \Cref{remark:interactive-verification} shows how combining \Cref{inf-theorem:covert-quantum-data-acquisition-v1} with \cite[Corollary 2]{caro2024interactiveproofsverifyingquantum} implies that any Boolean many-vs-one task solvable from polynomially many phase state copies is also solvable with polynomially many classical queries. Hence, in this setting, phase states offer at most a polynomial information-theoretic advantage over classical queries. 

The algorithm underlying \Cref{inf-theorem:covert-quantum-data-acquisition-v1} achieves the stronger of the two privacy notions from \Cref{inf-definition:covert-quantum-data-public-quantum-oracle-private-classical-queries}, but only against adversaries with a restricted view of the learner’s public-oracle queries.\footnote{The soundness guarantee achieved by the algorithm from \Cref{inf-theorem:covert-quantum-data-acquisition-v1} in fact holds against arbitrary adversaries, even if they observe and modify the quantum communication between the learner and the public quantum query oracle in both directions.}
Our next result removes this restriction, establishing covertness even against adversaries that can observe and modify the entire quantum transcript exchanged between the learner and the public quantum oracle, however, under the restriction that the adversary’s strategy is i.i.d.\ across queries and that it is ancilla-free (i.e., has no separate quantum memory of its own).

\begin{theorem}[Covert Verifiable Phase States Against i.i.d.\ Ancilla-free Adversaries---Informal]\label{inf-theorem:covert-quantum-data-acquisition-v2}

    There exists a computationally efficient $(m_{\mathrm{pri}}=\mathcal{O}(\poly(n, m, \log(1/\delta), 1/\varepsilon, 1/\delta_{\mathrm{leak}})), m_{\mathrm{pub}}=\mathcal{O}(\poly(n,m, \log(1/\delta), 1/\varepsilon,\allowbreak  1/\delta_{\mathrm{leak}})),\delta_c=0,\delta_s=\delta, \varepsilon, \delta_{\mathrm{leak}})$-covert verifiable procedure for producing $m$ quantum phase state copies for an arbitrary concept class $\mathcal{F}$ that satisfies \textbf{Privacy Version 2} against i.i.d.~ancilla-free adversaries that gain information with probability at least $\delta_{\mathrm{leak}}$.
\end{theorem}

Together, \Cref{inf-theorem:covert-quantum-data-acquisition-v1,inf-theorem:covert-quantum-data-acquisition-v2} show that, under specified physical constraints on the adversary, efficient covert verifiable quantum data acquisition is achievable with only polynomial overhead in query complexity. We discuss the adversary models considered and possible extensions in \Cref{subsection:future-work}.

Finally, we demonstrate how the covert verifiable quantum data acquisition subroutines from \Cref{inf-theorem:covert-quantum-data-acquisition-v1,inf-theorem:covert-quantum-data-acquisition-v2} can be used to help a learner combine private classical and public quantum oracle access to solve a problem with exponential quantum advantage while maintaining privacy. Concretely, we do so for the Forrelation problem, following the phrasing in \Cref{inf-definition:covert-forrelation}.

\begin{corollary}[Covert Verifiable Forrelation against Unidirectional Adversaries---Informal]\label{inf-corollary:covert-forrelation-v1}
    There exists a computationally efficient $(m_{\mathrm{pri}}=
    \poly(n,\log(1/\delta))
    , m_{\mathrm{pub}}=
    \poly(n,\log(1/\delta))
    , \delta_c=\delta, \delta_s=\delta)$-covert verifiable procedure for solving the Forrelation problem that satisfies \textbf{Privacy Version 1} against arbitrary unidirectional 
    adversaries that observe and modify only the quantum communication from $\mathsf{O}_{\mathrm{pub}}^{\mathrm{QPh}}$ to the learner. 
\end{corollary}

\begin{corollary}[Covert Verifiable Forrelation against i.i.d.\ Ancilla-free Adversaries---Informal]\label{inf-corollary:covert-forrelation-v2}
    There exists a computationally efficient $(m_{\mathrm{pri}}= m_{\mathrm{pub}}= \mathcal{O}(\poly(n, \log(1/\delta), 1/\delta_{\mathrm{leak}})), \delta_c=\delta, \delta_s=\delta, \delta_{\mathrm{leak}})$-covert verifiable procedure for solving the Forrelation problem that satisfies \textbf{Privacy Version 2} against i.i.d.~ancilla-free quantum adversaries that gain information with probability at least $\delta_{\mathrm{leak}}$.
\end{corollary}

As seen in \Cref{inf-corollary:covert-forrelation-v2}, the scaling of 
$m_{\mathrm{pub}}$ and $m_{\mathrm{pri}}$ depends on the information-extraction 
probability of the adversary that the protocol is designed to protect against. For the precise scaling, refer to \Cref{sec:iid-both-directions}.
In particular, we obtain efficient covert quantum data acquisition and thus 
efficient covert Forrelation, against adversaries that gain information with probability at least inverse-polynomial in $n$.

\Cref{inf-corollary:covert-forrelation-v1,inf-corollary:covert-forrelation-v2} realize the promise of \Cref{inf-definition:covert-forrelation} against two different classes of adversaries, allowing the learner to make use of the public quantum oracle in conjunction with the private classical oracle to efficiently solve the Forrelation problem---thus overcoming the exponential query complexity lower bound that they would face when using only their private oracle---without leaking information to an adversarial eavesdropper.
While here we have chosen to present a task-specific version of covert quantum data processing tailored to Forrelation, our algorithmic framework extends to other problems. For instance, we also give analogues of \Cref{inf-corollary:covert-forrelation-v1,inf-corollary:covert-forrelation-v2} for efficiently solving Simon's problem in a covert verifiable manner. Thus, we have shown how to make two of the most fundamental exponential quantum advantages in the query complexity model covert and verifiable.

\subsubsection{Algorithmic Ideas}

\paragraph{Covert learning from examples.}
Both \Cref{inftheorem:public-example-private-sq-parity,inftheorem:public-qmeasexample-private-qsq-quadratic} 
rely on a simple observation about the structure of the respective procedures for learning from random/quantum measurement examples: As we collect more examples, the space of consistent concepts (parities in the case of \Cref{inftheorem:public-example-private-sq-parity}, quadratic functions in the case of \Cref{inftheorem:public-qmeasexample-private-qsq-quadratic}) shrinks. Once the space has been shrunken down sufficiently, the private SQ/QSQ oracle access becomes sufficient to query-efficiently identify the unknown function among the remaining possibilities. 


When learning a quadratic function from $\mathsf{O}^{\mathrm{QMeasEx}}$-access in \Cref{inftheorem:public-qmeasexample-private-qsq-quadratic}, the procedure from \cite{arunachalam2023role} first uses $\mathcal{O}(n+\log(1/\delta_c))$ many two-copy Bell basis measurements to learn all off-diagonal elements of the unknown matrix $A$. Knowing these off-diagonal elements, one can reduce the remaining task of learning the diagonal entries to a parity learning problem, which can be solved using only $\mathcal{O}(n)$ QSQs with constant tolerance by estimating the influence of each variable \cite{arunachalam2020quantumstatisticalquerylearning}. For our privacy considerations, it is important to note that the Bell basis measurement outcomes are completely independent of the diagonal entries of $A$. Thus, an adversary who only sees those measurement outcomes cannot do better than guess each diagonal entry uniformly at random, which implies $2^{-n}$-privacy in the language of \Cref{inf-definition:covert-learning-public-quantum-examples-private-qsq}.

\paragraph{Covert quantum learning: Verifiability.}
To achieve covert verifiable quantum data acquisition as in \Cref{inf-definition:covert-quantum-data-public-quantum-oracle-private-classical-queries}, we first focus on verifiability, i.e., on completeness and soundness. As our quantum data comes in the form of phase states, this reduces to the task of certifying whether an unknown $(mn)$-qubit quantum state $\rho$ has sufficiently large overlap with an $m$-copy phase state $\ket{\psi_f^{\mathrm{Ph}}}^{\otimes m}$, assuming that we are given classical query access to $f$.
In the i.i.d.~case, that is, for $\rho=\rho_{0}^{\otimes m}$, the shadow overlap procedure from \cite{huang2025certifying-nature} can be used to achieve this, accepting $\rho=\ket{\psi_f^{\mathrm{Ph}}}\bra{\psi_f^{\mathrm{Ph}}}^{\otimes m}$ with certainty and rejecting any $\rho$ with $\bra{\psi_f^{\mathrm{Ph}}}^{\otimes m}\rho\ket{\psi_f^{\mathrm{Ph}}}^{\otimes m}<1-\varepsilon$ with probability $\geq 1-\delta$. To do this, the procedure uses either non-adaptive single-qubit Pauli measurements on $\mathcal{O}\left(n^2 \log(1/\delta)/\varepsilon^2\right)$ many i.i.d.~copies or adaptive single-qubit measurements on $\mathcal{O}\left(n \log(1/\delta)/\varepsilon\right)$ many i.i.d.~copies.

To go beyond the i.i.d.~case, we invoke the toolkit developed in \cite{fawzi2024learning}. 
Concretely, \cite[Theorem 3]{fawzi2024learning} lifts non-adaptive quantum algorithms from the i.i.d.~setting to general inputs, and we apply it to the shadow overlap protocol with non-adaptive single-qubit Pauli measurements. 
This comes at the cost of a $\poly(m,n,1/\delta,1/\varepsilon)$ overhead in the number queries to $\mathsf{O}_{\mathrm{pub}}^{\mathrm{QPh}}(f)$.\footnote{One may also apply \cite[Theorem 1]{fawzi2024learning} to lift the adaptive shadow overlap protocol to general inputs. However, this incurs a factor exponential-in-$n$ in the number of public quantum oracle queries, which is undesirable for our purposes.}  
To achieve this combination of the results from \cite{huang2025certifying-nature, fawzi2024learning}, we observe that \cite[Algorithm 1]{fawzi2024learning} can be applied even if the underlying i.i.d.~algorithm accesses a classical membership query oracle $\mathsf{O}_{\mathrm{pri}}^{\mathrm{Mem}}(f)$, and that the non-adaptive shadow overlap adapted from \cite{huang2025certifying-nature} to our setting uses queries $\mathsf{O}_{\mathrm{pri}}^{\mathrm{Mem}}(f)$ only during the classical post-processing of the measurement outcomes.
This way, we ensure completeness and soundness from \Cref{inf-definition:covert-quantum-data-public-quantum-oracle-private-classical-queries}, which, as we discuss in \Cref{sec:public-quantum-oracle-private-classical-queries}, even without covertness, has ramifications for the interactive verification of quantum learning \cite{caro2024verification, ma2024classicalverificationquantumlearning, caro2024interactiveproofsverifyingquantum}.

\paragraph{Covert quantum learning: Covertness against Unidirectional Adversaries.}
Next, we turn our attention to the privacy requirements. First, we give two simple strategies---the first using private randomness, the second using entanglement---for the learner to achieve covertness against unidirectional adversaries that ``see'' only the quantum communication back from the oracle to the learner. 
Note that $[Z^r, \mathsf{O}_{\mathrm{pub}}^{\mathrm{QPh}}(f)]=0$ holds for all $r\in\{0,1\}^n$ and $f:\{0,1\}^n\to\{0,1\}$.\footnote{Here, we write $Z^r = \bigotimes_{i=1}^n Z^{r_i}$ for $r\in\{0,1\}^n$.}
Thus, a learner can first privately sample a uniformly random $r$, query $\mathsf{O}_{\mathrm{pub}}^{\mathrm{QPh}}(f)$ on the state $\ket{\psi^{(r)}} = Z^{r} H^{\otimes n} \ket{0^n} = \frac{1}{\sqrt{2^n}}\sum_{x\in\{0,1\}^n} (-1)^{r\cdot x} \ket{x}$, receive the state $\ket{\psi_{f}^{(r)}} = \mathsf{O}_{\mathrm{pub}}^{\mathrm{QPh}}(f)\ket{\psi^{(r)}} = \mathsf{O}_{\mathrm{pub}}^{\mathrm{QPh}}(f)Z^{r} H^{\otimes n} \ket{0^n} = Z^{r} \mathsf{O}_{\mathrm{pub}}^{\mathrm{QPh}}(f)H^{\otimes n} \ket{0^n} = Z^{r}\ket{\psi_f^{\mathrm{Ph}}}$, and undo $Z^r$ to obtain one copy of $\ket{\psi_f^{\mathrm{Ph}}}$.
The quantum state as seen by an eavesdropper who does not know $r$ and who only has access to the state returned by $\mathsf{O}_{\mathrm{pub}}^{\mathrm{QPh}}$ is the average $\mathbb{E}_{r\sim\{0,1\}^n}\left[\ket{\psi_{f}^{(r)}}\bra{\psi_{f}^{(r)}}\right]=\mathds{1}_2^{\otimes n}/2^n$. As this state is independent of $f$, the eavesdropper cannot gain any information about $f$.

This protocol admits a variant with entanglement instead of randomness: The learner begins by preparing the state $\ket{\Psi}_{\mathsf{RS}}=\frac{1}{\sqrt{2^n}}\sum_{r\in\{0,1\}^n}\ket{r}_\mathsf{R}\otimes \ket{\psi^{(r)}}_{\mathsf{S}}$, with a private $\mathsf{R}$-register.
Querying $\mathsf{O}_{\mathrm{pub}}^{\mathrm{QPh}}(f)$ on the $\mathsf{S}$-register produces $\ket{\Psi_f^{\mathrm{Ph}}}_{\mathsf{RS}}=\frac{1}{\sqrt{2^n}}\sum_{r\in\{0,1\}^n}\ket{r}_{\mathsf{R}}\otimes \ket{\psi_f^{(r)}}_{\mathsf{S}}$, which after applying suitable controlled-$Z$ gates 
becomes $H^{\otimes n}\ket{0^n}_{\mathsf{R}}\otimes \ket{\psi_f^{\mathrm{Ph}}}
_{\mathsf{S}}$. 
Thus, the learner has obtained a copy of $\ket{\psi_f^{\mathrm{Ph}}}$ in the ${\mathsf{S}}$-register. If we again consider an eavesdropper who only has access to the state returned by $\mathsf{O}_{\mathrm{pub}}^{\mathrm{QPh}}(f)$, then this reduced state (only on the $\mathsf{S}$-register) is the same as before, so we have the same privacy guarantee.

As $\ket{\psi_f^{(r)}}$ and $\ket{\Psi_f^{\mathrm{Ph}}}_{\mathsf{RS}}$ are phase states, and since $\mathsf{O}_{\mathrm{pub}}^{\mathrm{QPh}}(f)$-access (together with knowledge of $r$) suffices to simulate classical query access to the functions determining the phases in those states, the above covert query protocols can be directly integrated with our phase state certification through non-i.i.d.~shadow overlap estimation. This yields covert verifiable quantum data acquisition against unidirectional quantum adversaries as stated in \Cref{inf-theorem:covert-quantum-data-acquisition-v1}.

\paragraph{Covert quantum learning: Covertness against i.i.d.~Ancilla-free Adversaries.}

Bidirectional adversaries, who sees the state $\ket{\psi^{(r)}}$ sent from the learner to $\mathsf{O}_{\mathrm{pub}}^{\mathrm{QPh}}(f)$, immediately break the security of the first protocol (with private randomness): Measuring in the Hadamard basis $\{H^{\otimes n}\ket{r}=\ket{\psi^{(r)}}\}_{r\in\{0,1\}^n}$, the adversary learns the randomness $r$ without perturbing the state.
Performing the same measurement on the $\mathsf{S}$-subsystem of $\ket{\Psi}_{\mathsf{RS}}$ and then feeding the post-measurement state to $\mathsf{O}_{\mathrm{pub}}^{\mathrm{QPh}}(f)$ will similarly allow the adversarial eavesdropper to gain information about $f$. 
Note, though, that while this measurement was non-destructive in the randomness-based protocol, it is highly destructive in the entanglement-based protocol, completely and, since the adversarial eavesdropper does not have access to the learner's private $\mathsf{R}$-register, irrevocably breaking the entanglement between the $\mathsf{R}$- and $\mathsf{S}$-registers in the original state. Building on this intuition, we next develop a version of the entanglement-based protocol that allows the learner to detect eavesdropping by ancilla-free i.i.d.~adversaries.

Consider an \emph{ancilla-free} adversary $\mathcal{A}$ that may act on the $S$ register (the last $n$ qubits of $\ket{\Psi}_{\mathsf{RS}}$) immediately before the public phase oracle $\mathsf{O}_{\mathrm{pub}}^{\mathrm{QPh}}(f)$.
By our privacy analysis for unidirectional adversaries, the reduced state on $S$ that $\mathcal{A}$ observes is maximally mixed (i.e., $\rho_\mathsf{S}=\mathbb{I}_\mathsf{S}/2^n$).
Since $\mathcal{A}$ is ancilla-free, any pre-oracle operation it can implement without measuring is a unitary. Moreover, because $\mathsf{O}_{\mathrm{pub}}^{\mathrm{QPh}}(f)$ is also unitary and as unitaries leave the maximally mixed state invariant, if $\mathcal{A}$ does not measure then its marginal remains $f$-independent both before and after the oracle, and thus, it gains no information about $f$. Consequently, to obtain any information about $f$ from the learner's public query, an ancilla-free
adversary must perform a \emph{measurement} on at least one qubit of $S$ prior to the oracle
call \footnote{As the adversary has no quantum workspace beyond $\mathsf{S}$, its admissible operations are adaptive circuits on $\mathsf{S}$ consisting only of unitaries and projective measurements with classical feed-forward. In particular, it cannot implement general POVMs via Naimark dilation or swap any part of $\mathsf{S}$ into private quantum memory.}.

While $\ket{\Psi}_{\mathsf{RS}}$ is maximally entangled and thus has a Schmidt rank equal to the maximal value $2^n$, we show that the post-measurement state after the adversary's measurement has Schmidt number at most $2^{n-1}$. 
Any consequent quantum processing performed by the adversary, including querying $\mathsf{O}_{\mathrm{pub}}^{\mathrm{QPh}}(f)$, acts only locally on the ${\mathsf{S}}$-system, therefore it cannot increase the Schmidt number. 
Thus, the joint $\mathsf{RS}$-system state that the learner holds after receiving the $\mathsf{S}$-system back from the adversary has fidelity at most $\frac{1}{2}$ with the maximally entangled state $\ket{\Psi_f^{\mathrm{Ph}}}_{\mathsf{RS}}$ by \cite[Lemma 1]{terhal2000schmidt}.
If we consider an i.i.d.~adversary that extracts non-zero information with probability $\delta_{\mathrm{leak}}$, this adversary must measure at least one qubit with probability $\delta_{\mathrm{leak}}$ and hence, the learner receives i.i.d.~copies of a state that has fidelity at most $1 - (\delta_{\mathrm{leak}}/2)$ with the phase state $\ket{\Psi_f^{\mathrm{Ph}}}_{\mathsf{RS}}$. This low fidelity can be detected using shadow overlap estimation, now applied for the $(2n)$-qubit phase state $\ket{\Psi_f^{\mathrm{Ph}}}_{\mathsf{RS}}$. Thus, the learner can use shadow overlap estimation on top of the entanglement-based covert queries to $\mathsf{O}_{\mathrm{pub}}^{\mathrm{QPh}}(f)$ to simultaneously ensure completeness, soundness, and the weaker version of privacy from \Cref{inf-definition:covert-quantum-data-public-quantum-oracle-private-classical-queries} against i.i.d.~ancilla-free adversaries, yielding \Cref{inf-theorem:covert-quantum-data-acquisition-v2}. 

\subsection{Related Work}\label{subsct:related-work}

\paragraph{Covert Learning With Multiple Servers.} 
The core question of our work is whether there exist meaningful notions of covert (verifiable) learning that do not require any computational hardness assumptions, which are inherent to the \cite{canetti2021covert} model.
The recent work \cite{holmgren2023locally}, building on a setting proposed in \cite{ishai2019cryptographic}, defines a ``multi-oracle'' version of covert learning, which also answers this question affirmatively.
Here, one assumes that there are $k>1$ many public oracles, and the learner's goal is to prevent anyone who eavesdrops on the learner's interactions with any $k-1$ of those oracles from learning.
Such a ``multi-oracle'' setting is inspired by multi-server Private Information Retrieval protocols \cite{chor1998private}, which are known to admit information-theoretically secure protocols.
In contrast, in our single-oracle setting, it is not \textit{a priori} clear that one can define interesting notions of covert learning without vastly weakening the computational power of the adversary.

\paragraph{Verifiable Learning.} 
The interactive verification of Probably Approximately Correct (PAC) learning was introduced as PAC verification and formalized via interactive proof systems in \cite{goldwasser2021interactive}.
The follow-up works \cite{mutreja2022pac-verification, gur2024power} further developed this framework and gave new and improved PAC verification algorithms for learning and extending the framework towards verification of general statistical algorithms. 
Recently, interactive proofs for verifying different kinds of quantum learning (and testing) have been studied in \cite{caro2024verification, ma2024classicalverificationquantumlearning, caro2024interactiveproofsverifyingquantum}. 
Our completeness and soundness requirements are very much inspired by and similar in  spirit to those that appear in interactive verification of learning.
Privacy requirements, however, are not commonly considered in interactive proofs for learning, and are unique to covert notions of (verifiable) learning. Importantly, as we argued in \Cref{sbsct:new_models}, meaningful formulations of covertness for learning from quantum data must be integrated with aspects of verifiability, due to the destructive nature of measurements in quantum physics.

\paragraph{Private Information Retrieval and Differential Privacy.}
Several inequivalent notions of ``privacy'' arise in query access to data; here we contrast two standard ones with covertness.
In \emph{private information retrieval} (PIR) \cite{chor1998private}, a client retrieves an item from a database without revealing \emph{which index} was requested to the server(s). Thus, the hidden object is the query index. PIR does not aim to hide the client's algorithmic strategy beyond the index, nor to protect the contents of the items themselves.
In \emph{differential privacy} (DP) \cite{dwork2008differential}, the goal is the stability of the output distribution of a query under changes to a single item's value; informally, no single item's value can be inferred much better from the output response than without it. Here, the hidden object is any item's value and it's contribution to the query. 
By contrast, in our notion of covertness, the protected objects are the learner's \emph{strategy} (e.g., learning algorithm or hypothesis class) and/or the \emph{target} (the unknown object being learned), against an eavesdropper observing public queries. The privacy guarantee is therefore orthogonal to PIR and DP.

\subsection{Discussion and Outlook}\label{subsection:future-work}

Our work extends the scope of (verifiable) learning with privacy guarantees to quantum learning, without relying on computational hardness assumptions. We formalized new covert learning models tailored to quantum learning settings by treating \emph{strategy-} and \emph{target-covertness} separately, and we demonstrated their utility by devising protocols that meet these requirements. First, we observed that the classical shadows paradigm for estimating observables is inherently strategy-covert, since data collection is task-agnostic. Second, we showed that quantum learning algorithms following a two-step template that narrows the hypothesis class, such as quadratic function learning, stabilizer state learning, and Pauli shadow tomography, can be made target-covert by performing the second step privately. Finally, we provided procedures for verifiable, covert acquisition of quantum phase states from a public oracle against two physically restricted classes of adversaries. We employed these procedures as a subroutine to solve Forrelation and Simon's problem covertly while maintaining the exponential quantum advantage. We conclude by highlighting some open questions raised by our work.

\paragraph{New adversary models.} 
Our information-theoretic no-go results for covert verifiable acquisition of quantum phase states, together with positive results under physically restricted adversaries, mirror a familiar cryptographic pattern: information-theoretic impossibility can be overcome by imposing physical or computational constraints (e.g., in private information retrieval and quantum bit commitment). In particular, our \emph{ancilla-free} adversary model was studied in \cite{divincenzo2004security}, where a type of asymptotic security for quantum bit commitment was achieved under this restriction. Moreover, \emph{cheat-sensitive} quantum bit commitment was developed in \cite{hardy2004cheat,aharonov2000quantum,buhrman2008possibility}, providing explicit trade-offs between information gain and detectability. This is analogous in spirit to our Privacy Version~2 in \Cref{inf-definition:covert-quantum-data-public-quantum-oracle-private-classical-queries}: Any nontrivial information gain necessarily induces a fidelity drop that our certification detects. To extend beyond the adversaries considered here, two natural directions are: (i) relaxing the ancilla-free condition to bounded/noisy-storage models \cite{damgaard2007tight,wehner2008composable,damgaard2008cryptography,wehner2008cryptography}; and (ii) adopting standard computational assumptions. For the latter, a concrete path is to relate our covert learning primitives to established cryptographic tasks, analogous to \cite{canetti2021covert} connecting their model to key exchange, to determine the minimal assumptions required. 

Along with the adversary models discussed above, as \Cref{inf-theorem:covert-quantum-data-acquisition-v2} currently guarantees covertness only against ancilla-free i.i.d.~adversaries, it is interesting to ask whether this result can be extended to non-i.i.d.~adversaries that have strictly less than $n$ qubits of quantum memory. Our present information-disturbance trade-off is discrete: any nonzero information gain forces fidelity to drop to $1/2$. It would be desirable to develop smooth, fine-grained trade-offs that relate the amount of information extracted (e.g., advantage or mutual information) to the resulting fidelity loss, thereby capturing a broader range of scenarios.

\paragraph{New kinds of data.}
Our covert verifiable quantum learning algorithms have focused on quantum data in the form of (uniform) quantum example or phase states. While this form of quantum data is relevant in a variety of scenarios, there are other important forms of quantum data. 
On the one hand, non-uniform superpositions become relevant when considering learning w.r.t.~non-uniform distributions \cite{kanade2019learningDNFs, caro2020quantum} or in distribution learning and testing \cite{gilyen2020distributional}, and even noisy quantum examples can still be a powerful resource for learning \cite{cross2015quantumlearning, grilo2019LWE, caro2020quantum}.
On the other hand, more general quantum states (or channels) arise naturally as the quantum data to be processed in quantum physics experiments \cite{huang2022quantum-advantage}.
Thus, future work will have to extend covert verifiability to learning from these other forms of quantum data. For this, the state certification protocol from  \cite{gupta2025singlequbitmeasurementssufficecertify} may be relevant as it allows for efficient certification of all quantum states albeit using non-adaptive single-qubit measurements (thus, under our current approach, we can only hope for results against i.i.d.~adversaries) and using a different classical oracle than \cite{huang2024certifying-arXiv}. To develop covert protocols beyond phase states against non-i.i.d.~adversaries, it will be interesting to explore non-adaptive results from \cite{huang2024certifying-arXiv} for other classes of quantum states.

On the classical side, we have already demonstrated that some of our covert learning definitions are also meaningful for classical learning scenarios. Taking inspiration from our mainly quantum notions of covert learning, can we develop further classical notions of covert learning and, more importantly, non-trivial covert learning algorithms that satisfy those definitions?
For instance, can we identify additional classical learning algorithms that fit the two-stage template underlying \Cref{inftheorem:public-example-private-sq-parity,inftheorem:public-qmeasexample-private-qsq-quadratic} and can thus be endowed with target-covertness?

\paragraph{Algorithmic improvements.}
Our covert verifiable quantum data acquisition can in principle be used to obtain covert verifiable versions of any quantum algorithm that relies on quantum examples or phase states. However, our method incurs a polynomial complexity overhead, which leads to vacuous guarantees for problems with a (low-degree) polynomial quantum advantage such as Bernstein-Vazirani---in the sense that the learner could have achieved similar or better query complexity by only querying their private classical oracle, ignoring the public quantum oracle altogether. This raises the question: Can covert verifiable quantum data acquisition be achieved with a smaller or even no complexity overhead?

There are also niches between our impossibility results and our current algorithms. Concretely, our current no-go theorem for information-theoretic covert verifiable acquisition of quantum phase states (presented in \Cref{sec:public-quantum-oracle-private-classical-queries}) applies to function classes that are learnable with high probability from a \emph{single} phase-state copy. Can this restriction be removed? Alternatively, are there function classes that are not one-copy-learnable for which the no-go can be circumvented and information-theoretic covert learning is possible?

In the direction of strategy-covertness, while we have phrased a general compilation definition for (Q)SQs, it would be interesting to extend this definition to and instantiate it for other classical or quantum scenario. That is, one would have to find other privacy-preserving query compilations that reduce $A$-type queries to $B$-type queries which allow us to estimate $A$-query responses from $B$-query responses with small overhead, while ensuring the $B$-query transcript leaks only negligible information about the original $A$-queries.

\paragraph{Applications.}
As mentioned in the introduction, beyond theoretical interest, the design of covert learning algorithms in the \cite{canetti2021covert} model is motivated by applications to polynomial-time-undetectable model stealing algorithms (see \cite{canetti2021covert, karchmer2023theoretical} and references therein). Namely, covert learning algorithms in the \cite{canetti2021covert} model yield model stealing attacks that cannot be efficiently defended against by analyzing only the query distribution. Potential implications of our alternative formulations of covert learning to (possibly quantum) model stealing without relying on hardness assumptions remain to be explored.

\section{Preliminaries}

\subsection{Notation}

We briefly summarize the conventions we use throughout the paper. We adopt Dirac bra–ket notation for quantum states; for instance, $\ket{\psi}$ denotes a pure state, $\bra{\psi}$ its conjugate transpose, and $\ket{\psi}\bra{\psi}$ the associated density operator. We use symbols $\rho, \sigma$ to denote mixed states. For Hilbert space subsystems (registers) we use sans-serif letters such as $\mathsf{R},\mathsf{S},\mathsf{F}$. 

Script letters $\mathcal{F} and \mathcal{S}$  will denote classes of functions and states, respectively, while $\mathcal{A}$ denotes an adversary or a class of adversaries based on the context. $\mathcal{E}, \mathcal{D}, \mathcal{L}, A$, and $A'$ are used for denoting algorithms.  $H=\tfrac{1}{\sqrt{2}}\begin{pmatrix}1&1\\[2pt]1&-1\end{pmatrix}$ denotes the Hadamard gate, while $Z=\begin{pmatrix}1&0\\[2pt]0&-1\end{pmatrix}$ denotes the Pauli-$Z$ gate. We write $H^{\otimes n}$ for $n$-fold Hadamards and, for $r\in\{0,1\}^n$, $Z^r \coloneqq \bigotimes_{i=1}^n Z^{r_i}$. $f, g, h$ refer to specific Boolean functions, while $F, G, H$ refer to random variables sampled from a class of functions. We use $\mathrm{poly}(\cdot)$ for polynomially bounded quantities and standard asymptotic notation $\mathcal{O}(\cdot), \tilde{\mathcal{O}}(\cdot)$ for upper bounds. 

Oracles will be denoted by $\mathsf{O}$, with a subscript specifying whether they are \emph{public} ($\mathrm{pub}$) or \emph{private} ($\mathrm{pri}$), and a superscript specifying the type of oracle. The argument of the oracle is the underlying function or state. For example, $\mathsf{O}^{\mathrm{QMeasEx}}_{\mathrm{pri}}(\rho)$ denotes a private quantum measurement example oracle acting for a (mixed) state $\rho$, and $\mathsf{O}^{\mathrm{QPh}}_{\mathrm{pub}}(f)$ is a public quantum phase oracle for the function $f$.

We freely switch between using ``adversary'' or ``eavesdropper'' to describe the party that views communication between the learner and the oracle, and against whom we aim to provide privacy, depending on whether we are in a cryptographic or interactive verification context, respectively. 

\subsection{Oracles}\label{subsec:oracles}

Oracle access means black-box query access to a function or a state, with costs measured in query complexity. Intuitively, one can imagine delegating an experiment to a trusted but opaque service that returns outcomes associated with the function or the state, without revealing any internal structure. 
We now give an overview of the different oracles that appear in our work.

\subsubsection{The (Quantum) Statistical Query Model}

\textbf{Classical SQ model.} The statistical query (SQ) model, introduced by Kearns~\cite{kearns1998efficient}, is a restricted learning framework where an algorithm learns by requesting statistical properties of the data distribution rather than observing individual examples. 
Let $f: \{0,1\}^n \to \{0,1\}$ be the unknown target concept, and let the underlying distribution over the instance space $\{0,1\}^n$ be uniform. In the SQ model, the learner has access to an oracle that provides estimates of expectation values for user-specified query functions.

\begin{definition}[Statistical Query Oracle]
\label{def:sq_oracle}
A \emph{statistical query (SQ) oracle} for a target concept $f$ is an oracle, denoted by $\mathsf{O}^{\mathrm{SQ}}(f)$, that takes as input a query pair $(q, \tau)$. This pair consists of:
\begin{itemize}
    \item A \emph{query function} $q: \{0,1\}^n \times \{0,1\} \to \{0,1\}$.
    \item A \emph{tolerance parameter} $\tau > 0$.
\end{itemize}
The oracle returns a value $v$ that is guaranteed to be a $\tau$-accurate estimate of the true expectation of $q$ over the uniform distribution of labeled examples induced by $f$. Formally, the returned value $v$ must satisfy:
\begin{equation}
    \left| v - \mathbb{E}_{x}[q(x, f(x))] \right| \leq \tau.
\end{equation}
\end{definition}

Intuitively, the SQ oracle allows the learner to gather aggregate information about the target concept without revealing specific data points. The complexity of an SQ algorithm is measured by the number of queries it makes, the complexity of its query functions, and the required tolerances. Any SQ oracle can be simulated using a random example oracle; by drawing $\mathcal{O}(\tau^{-2})$ examples, one can use the empirical mean to answer a single query $(q, \tau)$ with high constant probability.

\textbf{Quantum SQ model.} \cite{arunachalam2020quantumstatisticalquerylearning} introduced the quantum analogue of the classical SQ model, the \emph{quantum statistical query} (QSQ) model. Let $\rho$ be an unknown (mixed) quantum state. In the QSQ model, the learner has oracle access to estimates of expectation values of user-specified observables with respect to $\rho$.

\begin{definition}[Quantum Statistical Query Oracle]
\label{def:qsq_oracle}
A \emph{quantum statistical query (QSQ) oracle} for a $n$-qubit quantum state $\rho$ is an oracle, denoted by $\mathsf{O}^{\mathrm{QSQ}}(\rho)$, that takes as input a query pair $(M, \tau)$. This pair consists of:
\begin{itemize}
    \item A bounded (Hermitian) \emph{observable} $M \in \left((\mathbb{C}^2)^{\otimes n} \times (\mathbb{C}^2)^{\otimes n}\right)$ with $\lVert M \rVert\leq 1$
     \item A \emph{tolerance parameter} $\tau > 0$.
\end{itemize}
The oracle returns a value $v$ that is guaranteed to be a $\tau$-accurate estimate of the true expectation of observable $M$ on $\rho$. Formally, the returned value $v$ must satisfy:
\begin{equation}
    \left| v - \tr{[M\rho]} \right| \leq \tau.
\end{equation}
\end{definition}

As in the SQ model, the cost of a QSQ algorithm is measured by the number of queries, the complexity of the queried observables (e.g., the circuit size needed to implement a measurement), and the required tolerances. Moreover, a single QSQ query can be efficiently simulated with high constant probability from $\mathcal{O}(\tau^{-2})$ i.i.d.\ copies of $\rho$. 

As discussed in \Cref{sbsct:new_models}, we can speak of QSQ access to Boolean functions $f:\{0,1\}^n\to\{0,1\}$ instead of states by associating each $f$ to a state $\rho_f$ and querying $\mathsf{O}^{\mathrm{QSQ}}(\rho_f)$ as in \Cref{def:qsq_oracle}. Two standard choices are:
(i) the \emph{quantum example state} on $n+1$ qubits, $\ket{\psi_f^{\mathrm{Ex}}}=2^{-n/2}\sum_{x\in\{0,1\}^n}\ket{x,f(x)}$ with $\rho_f=\ket{\psi_f^{\mathrm{Ex}}}\bra{\psi_f^{\mathrm{Ex}}}$; and
(ii) the \emph{quantum phase state} on $n$ qubits, $\ket{\psi_f^{\mathrm{Ph}}}=2^{-n/2}\sum_{x\in\{0,1\}^n}(-1)^{f(x)}\ket{x}$ with $\rho_f=\ket{\psi_f^{\mathrm{Ph}}}\bra{\psi_f^{\mathrm{Ph}}}$.

A key QSQ subroutine that we will use later in \Cref{subsec:public-QMeasEX-private-QSQ} is the following lemma.

\begin{lemma}[{\cite[Lemma 4.1]{arunachalam2020quantumstatisticalquerylearning}}]
Let $f : \{0,1\}^n \to \{-1,1\}$ and 
$\ket{\psi_f} = \frac{1}{\sqrt{2^n}} \sum_{x \in \{0,1\}^n} \ket{x, f(x)}$.
There is a procedure that, on input $T \subseteq \{0,1\}^n$, outputs a $\tau$-estimate of 
$\sum_{S \in T} \hat{f}(S)^2$
using one (efficient) QSQ query on $\ket{\psi_f}\bra{\psi_f}$ with tolerance $\tau$.
\end{lemma}

In particular, taking $T_i\coloneqq\{S\subseteq\{0, 1\}^n: s_i = 1\}$ yields a $\tau$-estimate of the influence for each variable $i$. Hence, parity can be learned with $O(n)$ QSQ queries (e.g., using $\tau<1/2$ to distinguish $\mathrm{Inf}_i(f)\in\{0,1\}$ for each $i$).

\subsubsection{Quantum Measurement Example Query Model}

In the QSQ setting, we can imagine that the $n$-qubit state $\rho$ resides on a remote server: the learner specifies an observable and a tolerance, and receives an estimate. Thus, only aggregate statistics are revealed.
We also consider a stronger model in which, instead of expectation estimates, the learner requests the \emph{raw measurement outcomes} from a specified quantum experiment (e.g., a POVM or a circuit followed by computational-basis readout), and performs all post-processing locally. Both models capture the scenario where the user lacks quantum hardware but can access a remote quantum device; they differ in the granularity of access to experimental outcomes.

\begin{definition}[Quantum Measurement Example Oracle]
\label{def:qmeasex_oracle}
A \emph{quantum measurement example oracle} for a $n$-qubit quantum state $\rho$ is an oracle, denoted by $\mathsf{O}^{\mathrm{QMeasEx}}(\rho)$, that takes as input a $m$-copy POVM $\{E_i\}_{i=1}^k$. 
The oracle outputs a sample drawn from the probability distribution $\{p_i=\tr[E_i\rho^{\otimes m}]\}_{i=1}^k$.
\end{definition}

To obtain a meaningful notion of complexity, we weight an $m$-copy POVM query by a factor of $m$ when evaluating the overall query complexity of an algorithm with access to $\mathsf{O}^{\mathrm{QMeasEx}}$. For ease of notation, if such an algorithm makes $\ell$ many queries to $\mathsf{O}^{\mathrm{QMeasEx}}(\rho)$ where the $i$th query is an $m_i$-copy query and where $\sum_{i=1}^\ell m_i\leq m_{\mathrm{pri}}$, then we will say the algorithm ``makes at most $m_{\mathrm{pri}}$ queries to $\mathsf{O}^{\mathrm{QMeasEx}}(\rho)$''.

Akin to the QSQ model above, we can apply the QMeasEx oracle to functions instead of quantum states by associating each function with a quantum state. Moreover, given a QSQ query $(M,\tau)$, we can simulate it via QMeasEx queries by measuring the projective measurement associated with $M$ and averaging the outcomes. Using $\mathcal{O}(\tau^{-2})$ 1-copy QMeasEx queries yields an estimate of the QSQ query with high constant probability.

\subsubsection{(Quantum) Example and Membership Query Model}

\textbf{Classical example and membership oracles.} In contrast to the statistical query model, other standard learning models provide access to individual data points. We briefly define the two most common oracles for a target concept $f: \{0,1\}^n \to \{0,1\}$.

\begin{definition}[Example and Membership Query Oracles]
\label{def:example_membership_oracles}
Let $f$ be the unknown target concept.
\begin{itemize}
    \item A \emph{Random Example Oracle}, denoted $\mathsf{O}^{\mathrm{Ex}}(f)$, takes no input. When queried, it draws an instance $x$ from a fixed (often uniform) distribution $\mathcal{D}$ over $\{0,1\}^n$ and returns the labeled pair $(x, f(x))$.

    \item A \emph{Membership Query Oracle}, denoted $\mathsf{O}^{\mathrm{Mem}}(f)$, takes an instance $x \in \{0,1\}^n$ as input. When queried with $x$, it returns the corresponding label $f(x)$.
\end{itemize}
\end{definition}

\textbf{Quantum example, membership and phase oracles.}
For all quantum oracles considered above, even though the underlying object is a quantum state, the queries and responses are all classical. In this section, we discuss oracles where the queries and responses are themselves quantum states. 

The simplest oracle in this setting, is the quantum example oracle. 

\begin{definition}[Quantum Example Query Oracle]
\label{def:q_example_oracle}
Let $\rho$ be the unknown quantum state. A \emph{Quantum 
Example Oracle}, denoted $\mathsf{O}^{\mathrm{QEx}}(\rho)$, takes no input and when queried, returns $\rho$.
\end{definition}

Akin to the QSQ and QMeasEx models above, we can apply the QEx oracle applied to functions instead of quantum states by associating each function with a quantum state. 

The following two quantum oracles are defined using an underlying Boolean function $f$.

\begin{definition}[Quantum Membership Query Oracle] \label{def:q_example_oracles}
 Let $f: \{0, 1\}^n \to \{0, 1\}^m$ be the unknown target concept. A \emph{Quantum Membership Oracle}, denoted $\mathsf{O}^{\mathrm{QMem}}(f)$, takes as input a quantum $(n+m)$-qubit state $\rho$. When queried, it acts as $\ket{x, y} \mapsto \ket{x, y \oplus f(x)}$ on computational basis states (and extended linearly to the whole state space). 
\end{definition}

\begin{definition}[Quantum Phase Query Oracle]
\label{def:q_phase_oracles}
    Let $f: \{0, 1\}^n \to \{0, 1\}$ be the unknown target concept. A \emph{Quantum Phase Oracle}, denoted $\mathsf{O}^{\mathrm{QPh}}(f)$, takes input a quantum $n$-qubit state $\rho$. When queried, it acts as $\ket{x} \mapsto (-1)^{f(x)}\ket{x}$ on computational basis states (and extended linearly to the whole state space). 
\end{definition}

In particular, when applied to the state $2^{-n/2}\sum_{x\in\{0, 1\}^n}\ket{x, 0^m}$, $\mathsf{O}^{\mathrm{QMem}}(f)$ returns the quantum example state associated with $f$, $\ket{\psi_f^{\mathrm{Ex}}}=2^{-n/2}\sum_{x\in\{0,1\}^n}\ket{x,f(x)}$. Similarly, when applied to the uniform superposition state $2^{-n/2}\sum_{x\in\{0, 1\}^n}\ket{x}$, $\mathsf{O}^{\mathrm{QPh}}(f)$ returns the quantum phase state associated with $f$, $\ket{\psi_f^{\mathrm{Ph}}}=2^{-n/2}\sum_{x\in\{0,1\}^n}(-1)^{f(x)}\ket{x}$. These states and oracles have been extensively studied in black-box quantum advantages \cite{deutsch1992rapid, bernstein1997complexity, simon1997power, aaronson2015forrelation} and quantum computational learning literature \cite{bshouty1995learning-DNF, arunachalam2017survey, arunachalam2022optimal}. See \Cref{remark:phase-states-vs-quantum-examples,remark:covert-quantum-examples-from-QMem} for a comparison of these two oracles and conditions under which they can be used interchangeably.

\section{Covert Quantum Statistical Queries}\label{sec:covert-(q)sqs}

\subsection{Warmup: Covert Low-Degree Polynomial Statistical Queries}\label{subsec:covert-poly-sq}

Let us begin with a warmup for the quantum setting by defining a notion of covert learning with a statistical query oracle. 

In this model, we focus on \textbf{strategy-covertness}. We consider a learner who has access to a public statistical query (SQ) oracle for an unknown concept, where the entire \emph{learning transcript}---the sequence of oracle queries and oracle responses---is observed by an adversary. In this setting, the learner's goal is not to hide the underlying concept itself (which may be known to the adversary, e.g., a cloud provider who owns the data), but rather to conceal their \textbf{learning strategy}, embodied by the specific set of queries they wish to make.

To achieve this, we introduce a framework for covertly making statistical queries. The core idea is a compilation scheme consisting of a pair of efficient algorithms, an encoder~($\mathcal{E}$) and a decoder~($\mathcal{D}$). The learner's true statistical queries of interest are first passed to the encoder, which transforms them into a set of encoded statistical queries. These encoded queries are designed to be sent to the public oracle. After receiving the oracle's responses, the learner uses the decoder to process these responses and recover accurate answers to their original, private queries.

A protocol in this model should satisfy three properties.
\begin{itemize}
    \item \textbf{Completeness}. The decoded answers are valid responses to the learner's original SQs, preserving the utility of the oracle.
    \item \textbf{Privacy}. The distribution of encoded queries is statistically indistinguishable from a distribution generated by a simulator that has no knowledge of the learner's original queries. This ensures that the public transcript reveals nothing about the learnering strategy.
    \item \textbf{Efficiency}. The protocol incurs only a polynomial overhead in complexity.
\end{itemize}

Now, formally:

\begin{definition}[Covert Statistical Query Model
]\label{definition:covert-sq-model}
    Let $Q\subseteq \mathbb{R}^{\mathcal{X}\times\mathcal{Y}}$ be a class of query functions.
    Let $\mathcal{E}: (Q \times (0,1))^{m} \to (Q\times (0,1))^{m^{(e)}}$, $(\vec{q},\vec{\tau}) \mapsto (\vec{q}^{(e)}, \vec{\tau}^{(e)})$, and $\mathcal{D}: (Q \times (0,1))^{m}\times (Q\times (0,1)\times\mathbb{R})^{m^{(e)}} \rightarrow \mathbb{R}^m$ denote (possibly randomized) encoding and decoding algorithms, respectively.
    The pair $(\mathcal{E}, \mathcal{D})$ is an $(m^{(e)}, \varepsilon, \delta_c, \gamma_p)$-statistically covert SQ algorithm for $m$ many $Q$-queries, if the following properties hold.
    \begin{itemize}
        \item \textbf{$(\varepsilon, \delta_c)$-Completeness:} For any $(\vec{q}, \vec{\tau})=((q_i, \tau_i))_{i=1}^m\in (Q\times (0,1))^m$,
        \begin{equation}
            \Pr_{\mathcal{E},\mathcal{D}}\left[ \norm{\mathcal{D}((\vec{q},\vec{\tau}),\mathcal{E}(\vec{q},\vec{\tau}),\mathsf{O}^{\mathrm{SQ}}(\mathcal{E}(\vec{q},\vec{\tau}))) - \mathsf{O}^{\mathrm{SQ}}(\vec{q}, \vec{\tau}) }_{\infty} > \varepsilon \right] \leq \delta_c \, .
        \end{equation}
        In this notation, two quantifiers are implicit; the requirement is to be read as: With probability at least $1-\delta_c$ over the internal randomness of $\mathcal{E}$ and $\mathcal{D}$, for any valid SQ oracle responses to the $m^{(e)}$ many encoded queries, the $i$th entry of the decoded vector should be $\varepsilon$-close to some valid SQ oracle response to the query $(q_i,\tau_i)$ for every $i$.
    
        \item \textbf{$\gamma_p$-Statistical Privacy:} There exists a simulator $S$ such that, for any $(\vec{q}, \vec{\tau})=((q_i, \tau_i))_{i=1}^m\in (Q\times (0,1))^m$, and for any adversary $A: (Q\times (0,1))^{m^{(e)}}\to  \{0,1\}$,
        \begin{equation}
            \left|\Pr_{\mathcal{E}} \big[ A(\mathcal{E}(\vec{q}, \vec{\tau})) = 1  \big] - \Pr_{S} \big[ A(S(1^{m^{(e)}})) = 1  \big] \right|\le \gamma_p \, .
        \end{equation}
        If $\gamma_p=0$, we speak of perfect privacy.

        \item \textbf{Query-Efficiency:} The number of encoded SQs satisfies $m^{(e)}\leq\mathcal{O}(\poly(n, m, \varepsilon^{-1}, \delta_c^{-1}, \gamma_p^{-1}))$, and the encoded tolerances satisfy $\tau_i^{(e)}\geq\Omega(\poly(\tau_i,n^{-1}, m^{-1},\varepsilon, \delta_c, \gamma_p))$.
    
    \end{itemize}
\end{definition}
One may also add a time-efficiency requirement to the above definition, requiring the runtime of both $\mathcal{E}$ and $\mathcal{D}$ to be bounded by $\poly(n, m, \varepsilon^{-1}, \delta_c^{-1}, \gamma_p^{-1})$.

Let us illustrate \Cref{definition:covert-sq-model} with a concrete example. Consider the class of $n$-variate polynomials of degree at most $d$, which is spanned by a basis of $N = \binom{n+d}{d}$ monomials $\{M_i\}_{i=1}^N$. Suppose a learner has a secret hypothesis that their analysis is best conducted using a sparse polynomial $q = \sum_{i \in S} c_i M_i$, where $S \subseteq \{1, \dots, N\}$ is a secret set of $k = |S|$ monomial indices, and $\{c_i\}_{i \in S}$ are the corresponding secret coefficients. The learner's objective is to estimate the expectation of $q$ with a final error of at most $\delta$.

\paragraph{A Direct (but Partially Insecure) Protocol.}
The most straightforward approach is to directly compute the target expectation, $\mathbb{E}[q] = c \cdot m$, where $m$ is the unknown moment vector whose components are the true expectations of the basis monomials. To do this, the learner queries the SQ oracle for only the $k$ components of $m$ corresponding to the monomials in their support set $S$, and then multiplies by the private coefficients. This protocol is partially secure: it successfully hides the secret coefficients $\{c_i\}$, but it fails to hide all of the learner's strategy, as the learning transcript explicitly reveals the support set $S$.

\paragraph{A Secure (but Inefficient) Protocol.}
To fix this privacy leak, the learner can query \emph{all} $N$ basis monomials. This protocol achieves perfect privacy, as the public query set is fixed and independent of the learner's strategy. The problem, however, is efficiency. The query complexity is now $N$, which can be prohibitive for high-dimensional problems where $k \ll N$.

\paragraph{An Efficient and Secure Protocol via Sketching.}
To achieve both privacy and efficiency simultaneously, we can use sketching via random projections. The theoretical foundation for this approach is the Johnson-Lindenstrauss (JL) Lemma, which guarantees that dot products are approximately preserved in a low-dimensional random projection.

\begin{lemma}[Johnson-Lindenstrauss Lemma, Dot Product Version]
\label{lem:jl-dot-product}
For any $\varepsilon_0>0, \delta_c < 1$, let $m^{(e)} = O(\log(1/\delta_c)/\varepsilon_0^2)$. Let $R$ be an $m^{(e)} \times N$ random matrix with entries drawn i.i.d.\ from $\mathcal{N}(0, 1/m^{(e)})$. Then for any two vectors $u, v \in \mathbb{R}^N$, with probability at least $1-\delta_c$ over the choice of $R$:
\begin{equation}
    | (Ru) \cdot (Rv) - u \cdot v | \le \varepsilon_0 ||u||_2 ||v||_2
\end{equation}
A direct consequence for a single vector $u$ is that the same event implies approximate norm preservation: $||Ru||_2^2 \le (1+\varepsilon_0)||u||_2^2$.
\end{lemma}

Leveraging this guarantee, we can directly estimate the dot product $c \cdot m$ by computing a corresponding dot product in the low-dimensional "sketched" space.

\begin{proposition}[Covert SQ via Random Projections]
\label{prop:covert-poly-sq-sketching}
Let $B_c$ and $B_m$ be public bounds on $||c||_2$ and $||m||_2$ respectively. For any query $(q, \delta)$, there exists a $(m^{(e)}, \varepsilon, \delta_c, \gamma_p)$-statistically covert SQ algorithm with perfect privacy ($\gamma_p=0$) that achieves:
\begin{itemize}
    \item \textbf{Efficiency:} The number of encoded queries is $m^{(e)} = O(\log(1/\delta_c)/\varepsilon_0^2)$, where $\varepsilon_0 = \delta/(2B_c B_m)$.
    \item \textbf{Completeness:} The final estimate is within $\delta$ of the true value (i.e., $\varepsilon = \delta$) with success probability at least $1-\delta_c$.
\end{itemize}
\end{proposition}

\begin{proof}[Proof sketch]
Let the target polynomial $q$ be defined by its coefficient vector $c \in \mathbb{R}^N$. The goal is to estimate $c \cdot m$. Since the input domain is $\{0,1\}^n$, the monomial values are bounded, and thus the moment vector $m$ has a bounded norm, $||m||_2 \le \sqrt{N}$, making a public bound $B_m$ reasonable. Recall that here, $m \in \mathbb{R}^N$ is the unknown moment vector whose components are the true expectation values of the basis monomials, such that $\mathbb{E}[q] = c \cdot m$.

\paragraph{Encoder ($\mathcal{E}$).}
The encoder generates an $m^{(e)} \times N$ random matrix $R$ as described in \Cref{lem:jl-dot-product}. The $m^{(e)}$ rows of $R$ define the coefficient vectors for the dense public queries. The tolerance for each query is set to $\tau_e = \delta/(4B_c\sqrt{m^{(e)}})$.

\paragraph{Decoder ($\mathcal{D}$).}
The decoder receives the response vector $y = Rm + \epsilon$, where $|\epsilon_j| \le \tau_e$. It performs two simple steps:
\begin{enumerate}
    \item Privately compute the "projected coefficient vector": $c_{\text{proj}} = Rc$.
    \item Compute the final estimate as the dot product in the low-dimensional space: $v_{\text{est}} = c_{\text{proj}} \cdot y$.
\end{enumerate}

\paragraph{Completeness.}
The total error is $|v_{\text{est}} - c \cdot m| = |(Rc) \cdot (Rm + \epsilon) - c \cdot m|$. We can split this into two parts using the triangle inequality:
\begin{equation}\text{Error} \le \underbrace{| (Rc) \cdot (Rm) - c \cdot m |}_{\text{Approximation Error}} + \underbrace{| (Rc) \cdot \epsilon |}_{\text{Noise Error}} \end{equation}
Let $E$ be the high-probability event (occurring with probability at least $1-\delta_c$) that the guarantee of \Cref{lem:jl-dot-product} holds for the pairs $(c,m)$ and $(c,c)$.
\begin{enumerate}
    \item \textbf{Approximation Error:} Conditioned on $E$, this term is bounded by $\varepsilon_0 ||c||_2 ||m||_2$. Using the public bounds $B_c$ and $B_m$, and our choice of $\varepsilon_0 = \delta/(2B_c B_m)$, this error is at most $\delta/2$.
    \item \textbf{Noise Error:} By the Cauchy-Schwarz inequality, $|(Rc) \cdot \epsilon| \le ||Rc||_2 ||\epsilon||_2$. Conditioned on the same event $E$, we are guaranteed that $||Rc||_2^2 \le (1+\varepsilon_0)||c||_2^2$. For any reasonable choice of $\varepsilon_0 \le 1$, this means $||Rc||_2 \le \sqrt{2}||c||_2 \le 2B_c$. The noise vector norm is bounded by $||\epsilon||_2 \le \sqrt{m^{(e)} \tau_e^2} = \tau_e \sqrt{m^{(e)}}$. By our choice of $\tau_e = \delta/(4B_c\sqrt{m^{(e)}})$, the noise error is bounded by:
    \begin{equation} ||Rc||_2 ||\epsilon||_2 \le (2B_c) \cdot (\tau_e \sqrt{m^{(e)}}) = (2B_c) \cdot \left(\frac{\delta}{4B_c\sqrt{m^{(e)}}} \sqrt{m^{(e)}}\right) = \frac{\delta}{2} \end{equation}
\end{enumerate}
Therefore, conditioned on the event $E$, the total error is at most $\delta/2 + \delta/2 = \delta$. Since $E$ occurs with probability at least $1-\delta_c$, the completeness guarantee holds.

\paragraph{Privacy.}
Privacy is perfect ($\gamma_p=0$). We define a simulator $S$ that operates without knowledge of the target query.
\begin{itemize}
    \item \textbf{Simulator $Sim$:} Given $m^{(e)}$, generate an $m^{(e)} \times N$ random matrix $R'$ from the same distribution used by the encoder. Output the rows of $R'$ as the query polynomials, paired with a default tolerance.
\end{itemize}
The distribution of the query polynomials is determined entirely by the random matrix $R$, which is independent of the the coefficients $c$ and also the relevant set $S$. Thus, the distribution of query functions produced by $\mathcal{E}$ is identical to that produced by $Sim$.
\end{proof}

\subsection{Covert Local and Low-Rank Observable Quantum Statistical Queries}

In the previous subsection, we have demonstrated how a learner can query $\mathsf{O}^{\mathrm{SQ}}_{\mathrm{pub}}$ in a way that leaks no information about the specific $m$ many statistical queries that are of interest to the learner. Notice that with a public random example oracle $\mathsf{O}^{\mathrm{Ex}}_{\mathrm{pub}}$, achieving the same kind of covertness guarantee---keeping the statistical queries of interest private---is trivial to achieve assuming bounded query functions: The learner simply collects $\mathcal{O}(\log(m/\delta)/\min_i \tau_i^2)$ many i.i.d.~random examples and then privately computes empirical averages to accurately estimate the true expectation value of the SQ with high success probability by standard Chernoff-Hoeffding bounds. 
And since the random examples collected are completely independent of the statistical queries to be evaluated, no information about the statistical query of interest (apart from what can be inferred by the number of random examples requested) is leaked to a potential eavesdropper.

The quantum version of this scenario, however, is not as immediate.
Suppose a learner $\mathcal{L}$, that does not have any quantum hardware of its own, wants to evaluate QSQs for some observable $M_1,\ldots,M_m$ by making queries to a public quantum measurement example oracle $\mathsf{O}^{\mathrm{QMeasEx}}_{\mathrm{pub}}$ in a way that leaks no (or only little) information about the $M_i$.
When querying $\mathsf{O}^{\mathrm{QMeasEx}}_{\mathrm{pub}}$, $\mathcal{L}$ has to choose which measurements to query the oracle on. This choice of measurement constitutes a public query strategy, which will in general depend on the observables $M_i$ of interest\footnote{Of course there is a strategy that does not depend on the $M_i$: Perform full state tomography to obtain an approximation of the true density matrix and then use that approximation to predict the expectation value. However, without prior assumptions on the unknown state, this approach requires a number of queries to $\mathsf{O}^{\mathrm{QMeasEx}}_{\mathrm{pub}}$ that is undesirably large \cite{haah2016sample, odonnell2016efficient}.}, thus potentially leaking information about $M$ to an eavesdropper. In this subsection, we argue that classical shadows \cite{huang2020predicting} circumvent this difficulty for relevant classes of observables.

First, we define the covert quantum statistical query model in analogy to \Cref{definition:covert-sq-model}, but now the QSQs are simulated with public quantum measurement examples (rather than public QSQs). For concreteness of exposition, we consider the case of $n$-qubit states/observables; the version for general qudits is analogous.

\begin{definition}[Covert Quantum Statistical Query Model---Formal version of \Cref{inf-definition:covert-qsq-model}]\label{definition:covert-qsq-model}

    Let $\mathcal{M}$ be a class of $n$-qubit observables.
    Let $\mathrm{POVM}_{k}$ denote the set of POVMs on $k$ copies of $n$-qubit states, let $\mathrm{POVM} = \bigcup_{k=1}^\infty \mathrm{POVM}_{k}$.
    Let $\mathcal{E}: (\mathcal{M} \times (0,1))^{m} \rightarrow \mathrm{POVM}^{m^{(e)}}$, $(\vec{M}, \vec{\tau})\mapsto \vec{M}^{(e)}$, and $\mathcal{D}: (\mathcal{M} \times (0,1))^{m}\times (\mathrm{POVM}\times \mathbb{N})^{m^{(e)}} \rightarrow \mathbb{R}^m$ denote (possibly randomized) encoding and decoding algorithms, respectively.
    The pair $(\mathcal{E}, \mathcal{D})$ is an $(m^{(e)}, \varepsilon, \delta_c, \gamma_p)$-statistically covert QSQ algorithm for $m$ many $\mathcal{M}$-queries from public quantum measurement examples if the following properties hold.
    \begin{itemize}
        \item \textbf{$(\varepsilon,\delta_c)$-Completeness:} For any $(\vec{M}, \vec{\tau})=((M_i, \tau_i))\in (\mathcal{M}\times (0,1))^m$,
        \begin{equation}
            \Pr_{\mathcal{E},\mathcal{D}, \mathrm{QMeasEx}}\left[ \norm{\mathcal{D}((\vec{M},\vec{\tau}),\mathcal{E}(\vec{M},\vec{\tau}),\mathsf{O}^{\mathrm{QMeasEx}}(\mathcal{E}(\vec{M},\vec{\tau}))) - \mathsf{O}^{\mathrm{QSQ}}(\vec{M}, \vec{\tau}) }_{\infty} > \varepsilon \right] \leq \delta_c \, .
        \end{equation}
        In this notation, a quantifier is implicit; the requirement is to be read as: With probability at least $1-\delta_c$ over the internal randomness of $\mathcal{E}$ and $\mathcal{D}$ as well as over the randomness of the sampling of measurement outcomes according to $\mathsf{O}^{\mathrm{QMeasEx}}$, the $i$th entry of the decoded vector should be $\varepsilon$-close to some valid QSQ oracle response to the query $(M_i,\tau_i)$ for every $i$.
    
        \item \textbf{$\gamma$-Statistical Covertness:} There exists a simulator $S$ such that, for any $(\vec{M}, \vec{\tau})=((M_i, \tau_i))\in (\mathcal{M}\times (0,1))^m$, and for any adversary $\mathcal{A}: \mathrm{POVM}^{m^{(e)}} \rightarrow \{0,1\}$, 
        \begin{equation}
            \left|\Pr_{\mathcal{E}} \big[ \mathcal{A}(\mathcal{E}(\vec{M}, \vec{\tau})) = 1  \big] - \Pr_{S} \big[ \mathcal{A}(S(1^{m^{(e)}})) = 1  \big] \right|\le \gamma_p \, .
        \end{equation}
        If $\gamma_p=0$, we speak of perfect privacy.
    
        \item \textbf{Query-Efficiency:} The number of encoded quantum measurement examples satisfies the bound $m^{(e)}\leq \mathcal{O}(\poly(n, m, \max_i \tau_i^{-1}, \varepsilon^{-1}, \delta_c^{-1}, \gamma_p^{-1}))$.
    \end{itemize}
\end{definition}

Additionally, if the encoding is such that $\vec{M}^{(e)}\in \mathrm{POVM}_k^{m^{(e)}}$ always holds, then we may speak of a covert QSQ algorithm from public $k$-copy measurement examples.
One may also add a time-efficiency requirement to the above definition, requiring the runtime of both $\mathcal{E}$ and $\mathcal{D}$ to be bounded by $\poly (n, m, \varepsilon^{-1}, \delta_c^{-1}, \gamma_p^{-1})$.

Our next result reinterprets classical shadows in terms of covert QSQs for bounded local observables using a public (single-copy) quantum measurement example oracle. 

\begin{proposition}[Classical shadows for $k$-local observables as a covert QSQ algorithm---Formal version of \Cref{obs:classical_shadows_intro}]\label{prop:strategy-covert-classical-shadows}
    Let $\mathcal{M}_{k}$ denote the class of $k$-local $n$-qubit observables with operator norm $\leq B$.
    Classical shadows with random Pauli measurements consitute a time-effcient $(m^{(e)}=\mathcal{O}(\log(m/\delta_p)4^k/\min_i\tau_i^2), \varepsilon = 0, \delta_p, \gamma_p = 0)$-statistically covert QSQ algorithm for $m$ many $\mathcal{M}_{k}$-queries from public random single-copy Pauli measurement examples.
\end{proposition}

Notice that $m^{(e)}$ here depends on $m,k,\delta_p$, and $\min_i\tau_i$. As the simulator in \Cref{definition:covert-qsq-model} receives $1^{m^{(e)}}$ as input, information about these parameters may leak to an adversary. However, no information about which $m$ many $k$-local observables are predicted leaks.

\begin{proof}
    This is an immediate consequence of the results from \cite{huang2020predicting}. We only have to observe that the random choice of Pauli basis queries in the classical shadows protocol is independent of the specific $k$-local observables of interest, and that classical shadows with random Paulis are computationally efficient.
\end{proof}

We can obtain a result similar to \Cref{prop:strategy-covert-classical-shadows} for low-rank (instead of $k$-local) observables. 

\begin{proposition}[Classical shadows for low rank observables as a covert QSQ algorithm---Formal version of \Cref{obs:classical_shadows_intro}]\label{prop:strategy-covert-classical-shadows-low-rank}
    Let $\mathcal{M}_{\norm{\cdot}_F\leq B}$ denote the class of $n$-qubit observables with Frobenius norm $\leq B$.
    Classical shadows with random Clifford measurements constitute a $(m^{(e)}=\mathcal{O}(\log(m/\delta_p)B^2/\min_i\tau_i^2), \varepsilon = 0, \delta_p, \gamma_p = 0)$-statistically covert QSQ algorithm for $m$ many $\mathcal{M}_{\norm{\cdot}_F\leq B}$-queries from public random single-copy Clifford measurement examples.
\end{proposition}

As in the previous proposition, while information may be leaked about $m,B,\delta_p$, and $\min_i\tau_i$, no information about the specific low-rank/low-Frobenius-norm observables of interest gets leaked.
Note that unlike in the random Pauli case above, the classical postprocessing for classical shadows with random Cliffords is not computationally efficient, so the decoder in this covert QSQ procedure is not time-efficient.

Classical shadows thus are a powerful tool for covert learning that keeps the actual question (in this case, observables) of interest private because of the ``measure first, ask questions later'' feature of the randomized measurement toolbox \cite{elben2023randomized}. Namely, in the randomized measurement framework, the data collection is the same irrespective of which specific task from a certain family of admissible/achievable tasks is considered. In fact, data can be collected before the task is even known. 
A similar task-independence of data collection appears more generally in sample-based (rather than adaptive, query-based) learning and testing: If two tasks can be solved from the same number of random samples, then those samples may be collected even before a decision is made on which of the two tasks one wants to solve. As such, the data collection process does not carry any information about the task of interest (other than on what number of samples suffices). 
For example, in standard Probably Approximately Correct (PAC) learning \cite{valiant1984theory}, if two concept classes have the same VC-dimension \cite{vapnik1971uniform}, the standard Empirical Risk Minimization (ERM) learner will collect a dataset of random examples in the same way for either class.
Similarly, if two properties have the same sample complexity in passive testing \cite{goldreich1998property, alon2016active, caro2024testingclassicalpropertiesquantum}, the specific property being tested cannot be inferred when observing only how the samples are collected.
In summary, strategy-covertness is easily achieved in scenarios of learning or testing from random samples.

\begin{remark}
    Clearly, the basis query strategy outlined in \Cref{subsec:covert-poly-sq} could also be applied to achieve a version of \Cref{definition:covert-qsq-model} in which a single QSQ is compiled into multiple QSQs, analogously to \Cref{definition:covert-sq-model}. 
    Namely, if any $M\in\mathcal{M}$ can be written as a linear combination of some basis elements $\{M_i\}_{i=1}^N$ with a coefficient vector with $1$-norm at most $C$, then we can compile any QSQ $(M, \tau)$ with $M\in\mathcal{M}$ into the QSQs $((M_i, \tau'))_{i=1}^N$ with $\tau' = \tau/C$.  
    For $k$-local or rank-$r$ observables, however, this approach requires large $N$. In the $k$-local case, $N=\sum_{i=0}^k \binom{n}{i}3^i$, which scales polynomially in $n$ for constant $k$; in the rank-$r$ case, $N$ grows exponentially with $n$ for constant $r$.
    In contrast, as shown above, via classical shadows, we can compile even many $k$-local or low-rank QSQs simultaneously into much fewer QMeasEx queries. Thus, the classical-shadows approach is, in spirit, akin to the random sketching strategy of \Cref{subsec:covert-poly-sq}: Both demonstrate that a modest number of randomized queries suffice to solve a broad class of tasks.
\end{remark}

In the learning scenarios discussed in this section, we do not address verifiability alongside covertness; designing protocols that achieve both remains an open problem.

\section{Covert Learning From Public Quantum Measurement Examples}
 
\subsection{Warmup: Public Classical Examples, Private Statistical Queries}\label{subsec:public-classical-private-sq-parity}

As a first proof-of-principle demonstration of how a learner can, with the support of a ``weak'' private oracle, learn from a ``strong'' public oracle in a covert manner, we consider the fully classical setting of covert learning from public random examples and private SQs. We begin by giving a formal version of \Cref{inf-definition:covert-exact-learning-public-examples-private-sq}.

\begin{definition}[Covert Exact Learning From Public Examples and Private SQs---Formal version of \Cref{inf-definition:covert-learning-public-quantum-examples-private-qsq}]\label{definition:covert-exact-learning-public-examples-private-sq}
    An algorithm $\mathcal{L}$ that has access to a private SQ oracle $\mathsf{O}_{\mathrm{pri}}^{\mathrm{SQ}}$ and to a public random example oracle $\mathsf{O}_{\mathrm{pub}}^{\mathrm{Ex}}$, both w.r.t.~a distribution with uniform marginal over inputs, is a $(m_{\mathrm{pri}}, m_{\mathrm{pub}},\tau_{\mathrm{pri}},\delta_c,\delta_p)$-covert exact learner for a concept class $\mathcal{F}$ from private SQs and public examples if it satisfies the following:
    \begin{itemize}
        \item $\delta_c$-\textbf{Completeness:} For any $f\in\mathcal{F}$, after making at most $m_{\mathrm{pri}}$ queries of tolerance $\tau_{\mathrm{pri}}$ to $\mathsf{O}_{\mathrm{pri}}^{\mathrm{SQ}}(f)$ and at most $m_{\mathrm{pub}}$ queries to $\mathsf{O}_{\mathrm{pub}}^{\mathrm{Ex}}(f)$, $\mathcal{L}$ outputs $f$ with success probability $\geq 1-\delta_c$.
        \item $\delta_p$-\textbf{Privacy:} For $F\sim \mathcal{F}$ a uniformly random concept, no adversary $\mathcal{A}$ can correctly guess $F$ with success probability $>\delta_p$ from the at most $m_{\mathrm{pub}}$ many public random examples requested by $\mathcal{L}$.
    \end{itemize}
\end{definition}

As mentioned in the introduction, the focus on exact learning in \Cref{definition:covert-exact-learning-public-examples-private-sq} is merely for ease of presentation. It can be naturally extended to an approximate notion of learning.\footnote{In an $\varepsilon$-approximate version of \Cref{definition:covert-exact-learning-public-examples-private-sq}, one would require the output of $\mathcal{L}$ in the completeness requirement to be an $\varepsilon$-approximation of the unknown $f$ with probability at least $1-\delta_c$, whereas in the privacy requirement the adversary should not be able to produce an $\varepsilon$-approximation of the random $F$ with probability $>\delta_p$.}
Similarly, while we have fixed the input marginal to be uniform, one may also consider versions of the definition with other (possibly unknown) underlying distributions over inputs.

For our discussion, we consider the concrete task of parity learning, which is well known to be efficiently solvable from linearly-in-$n$ many random examples but exponentially hard to solve from SQs \cite{kearns1998efficient}. Our next results gives a simple covert learning algorithm in which a learner uses public random examples to overcome the exponential bottleneck of parity learning from SQs.

\begin{observation}[Covert Parity Learning from Public Examples and Private SQs---Formal Version of \Cref{inftheorem:public-example-private-sq-parity}]\label{theorem:public-example-private-sq-parity}
    Let $n\in\mathbb{N}$.
    Let $\delta_c, \delta_p\in (0,1)$ with $k\coloneqq \lceil \log(1/\delta_p)\rceil < n$.
    There exists a computationally efficient $(m_{\mathrm{pri}} = 2/\delta_p, m_{\mathrm{pub}} = n-k +\mathcal{O}(\log(1/\delta_c)), \tau_{\mathrm{pri}} = 1/6, \delta_c, \delta_p)$-covert exact learner from private SQs and public examples for the concept class of $n$-bit parity functions.
    That is, there is a computationally efficient learner $\mathcal{L}$ that, for any unknown $s\in\{0,1\}^n$, makes at most $n-k +\mathcal{O}(\log(\delta_c))$ queries to a public parity example oracle $\mathsf{O}_{\mathrm{pub}}^{\mathrm{Ex}}(s\cdot (\cdot))$, 
    followed by at most $2/\delta_p$ many queries with tolerance $1/6$ to a private parity SQ oracle $\mathsf{O}_{\mathrm{pri}}^{\mathrm{SQ}}(s\cdot (\cdot))$, 
    in such a way that the following holds:
    \begin{itemize}
        \item \textbf{Completeness:} With probability $\geq 1-\delta_c$, $\mathcal{L}$ learns the unknown parity string $s$.
        \item \textbf{Privacy:} For $S\sim\{0,1\}^n$, any adversary $\mathcal{A}$ that only sees $\mathcal{L}$'s public example oracle queries correctly guesses the parity string $S$ with probability at most $\delta_p$.
    \end{itemize}
\end{observation}

The result of \Cref{theorem:public-example-private-sq-parity} should be contrasted with the capabilities of a learner $\mathcal{L}$ that only uses their private SQ access for queries with inverse-polynomial tolerance. Such a learner would have to make exponentially-in-$n$ many $\mathsf{O}_{\mathrm{pri}}^{\mathrm{SQ}}$-queries \cite{kearns1998efficient}. 
\Cref{theorem:public-example-private-sq-parity} demonstrates that $\mathcal{L}$ can be significantly more query- and computationally efficient when also making use of the public example oracle, while still ensuring that not too much information about the unknown parity leaks to a potential eavesdropper.

\begin{proof}
    First, we describe that protocol that $\mathcal{L}$ uses, then we prove that it has the claimed completeness and privacy properties.
    Let $s\in\{0,1\}^n$ be the unknown parity.
    
    \textbf{Algorithm:}
    The procedure is as follows: 
    \begin{enumerate}
        \setcounter{enumi}{-1}
        \item Set $k = \lceil \log(1/\delta_p)\rceil$ and $m_{\mathrm{pub}} = n-k + \mathrm{const}\cdot \log(1/\delta_c)$ for a sufficiently large constant.
        \item Initialize an empty data set $D = \{\}$ of $n$-bit inputs and $1$-bit outputs. Let $D_{\mathrm{in}}$ be the set of $n$-bit inputs that occur in $D$. Further, initialize a counter for oracle queries, $\mathrm{count}=0$. 
        While $\operatorname{dim}(\operatorname{span}(D))<n-k$, which $\mathcal{L}$ can check efficiently using Gaussian elimination, $\mathcal{L}$ does the following:
        \begin{enumerate}
            \item $\mathcal{L}$ makes a query to $\mathsf{O}_{\mathrm{pub}}^{\mathrm{Ex}}(s\cdot (\cdot))$, receiving an example $(x,s\cdot x)$. 
            \item $\mathcal{L}$ updates $D\gets D\cup\{x\}$.
            \item $\mathcal{L}$ updates $\mathrm{count}\gets\mathrm{count} + 1$. 
            \item If $\mathrm{count}\geq m_{\mathrm{pub}}$, $\mathcal{L}$ aborts and reports failure.
        \end{enumerate}
        \item $\mathcal{L}$ takes $n-k$ data points $(x_i,s\cdot x_i)$ associated to linearly independent $\{x_i\}_{i=1}^{n-k}$ from $D$. 
        Let $S$ be the space of parities consistent with the associated data points, i.e.,
        \begin{equation}
            S
            = \{t\in\{0,1\}^n~|~t\cdot x_i = s\cdot x_i~\forall 1\leq i\leq n-k\}\, .
        \end{equation}
        $\mathcal{L}$ can efficiently determine $S$ using Gaussian elimination.
        Note that $S$ is an affine subspace of dimension $n-(n-k)=k$ and therefore in particular $\lvert S\rvert = 2^k$.
        \item Set up a (binary) tournament tree with leaves labeled by strings in $S$. Play the tournament as follows: For any $t_1\neq t_2\in S$ playing a match against each other, let $\mathsf{X}_{t_1\neq t_2}$ be the set of inputs on which the parity functions for $t_1$ and $t_2$ differ, i.e., 
        \begin{equation}
            \mathsf{X}_{t_1\neq t_2}
            = \{x\in\{0,1\}^n~|~ t_1\cdot x\neq t_2\cdot x\}\, .
        \end{equation}
        Note that $\lvert \mathsf{X}_{t_1\neq t_2}\rvert = 2^{n-1}$ holds for all $t_1\neq t_2\in S$.
        Let $q_{t_1,t_2}:\{0,1\}^n\times \{0,1\}\to \{0,1\}$ be defined as
        \begin{equation}
            q_{t_1,t_2}(x,y)
            = \mathds{1}_{\mathsf{X}_{t_1\neq t_2}}(x)\cdot \delta_{t_1\cdot x, y}\, .
        \end{equation}
        Given $t_1,t_2$, $\mathcal{L}$ can computationally efficiently determine an efficient description of $q_{t_1,t_2}$.
        $\mathcal{L}$ now sends the SQ specified by $(q_{t_1,t_2}, 1/6)$ to $\mathsf{O}_{\mathrm{pri}}^{\mathrm{SQ}}(s\cdot (\cdot))$ and receives $\alpha_{t_1,t_2}\in [0,1]$ satisfying
        \begin{equation}
            \lvert \alpha_{t_1,t_2} - \mathbb{E}_{x\sim\mathcal{U}_n}[q_{t_1,t_2}(x,s\cdot x)]\rvert\leq \frac{1}{6}\, .
        \end{equation}
        If $\alpha_{t_1,t_2}\geq 1/3$, $t_1$ wins the match. Otherwise, $t_2$ wins the match. 
        In playing the whole tournament, $\mathcal{L}$ makes $m_{\mathrm{pri}}=\sum_{k=1}^{|S|-1}|S|\cdot 2^{-k} \leq |S| = 2^k\leq 2/\delta_p$ many SQs with constant tolerance parameter $1/6$.
        \item $\mathcal{L}$ outputs the string $t\in S$ that has won the tournament.
    \end{enumerate}

    \textbf{Completeness:}
    Now that we have described $\mathcal{L}$'s algorithm, we analyze its completeness and privacy. We begin by considering the completeness. 
    First, our choice of $m_{\mathrm{pub}}$ ensures that $\mathcal{L}$ sees $n-k$ linearly independent samples with probability at least $1-\delta_c$. 
    For the rest of the completeness analysis, we condition on this high probability event (and thus $\mathcal{L}$ does not not abort in Step 1).
    By definition of $S$ and $X_{t_1\neq t_2}$, we have 
    \begin{equation}
        \mathbb{E}_{x\sim\mathcal{U}_n}[q_{t_1,t_2}(x,s\cdot x)]
        = \begin{cases}
            \frac{1}{2} \quad &\textrm{if } s=t_1\\
            0 & \textrm{if } s=t_2
        \end{cases}\, ,
    \end{equation}
    and we ignore the value for $s\neq t_1,t_2$, since it is not relevant for our discussion.
    In particular, as $s\in S$, there is a unique $t_1\in S$ such that $\mathbb{E}_{x\sim\mathcal{U}_n}[q_{t_1,t_2}(x,s\cdot x)] = 1/2$ and $\mathbb{E}_{x\sim\mathcal{U}_n}[q_{t_2, t_1}(x,s\cdot x)] = 0$ for all $t_2\in S\setminus\{t_1\}$, namely $t_1 = s$.
    Therefore, choosing the SQ tolerance to be $1/6$ suffices to ensure that the true $s$ wins every match it plays in the tournament in Step 3, and thus the string that $\mathcal{L}$ outputs in Step 4 is indeed the desired $s$. This concludes the completeness proof.

    \textbf{Privacy:} Now we turn our attention to the privacy guarantee. Because of $\mathcal{L}$'s procedure, the only information that $\mathcal{A}$ sees are at most $m_{\mathrm{pub}}$ input-output pairs $(x_i,s\cdot x_i)$, where the $x_i$ span an at most $(n-k)$-dimensional space.
    Thus, $\mathcal{A}$ knows that the unknown parity has to be consistent with all these data points, i.e., lie in $S$, but the best $\mathcal{A}$ can do beyond that is to randomly guess an element of $S$. Hence, the probability that $\mathcal{A}$ correctly identifies the unknown parity string is at most $\lvert S\rvert^{-1} = 2^{-k}\leq \delta_p$ by our choice of $k$. This is the claimed privacy guarantee.
\end{proof}

The covert learner $\mathcal{L}$ in the proof of \Cref{theorem:public-example-private-sq-parity} proceeds in two phases. In the first phase, $\mathcal{L}$ queries the public oracle to obtain partial information about the unknown function, thereby restricting the space of possibilities. In the second phase, $\mathcal{L}$ queries its private oracle to identify the correct one among the remaining possibilities. Note that this procedure is adaptive, the private queries in the second phase depend on what $\mathcal{L}$ has learned from the public queries in the first phase. By design, any adversary $\mathcal{A}$ can only gain the partial information collected in the first phase of the procedure, which is how we obtain a covertness guarantee. The same two-phase structure will reappear in the next subsection.

We note that \Cref{theorem:public-example-private-sq-parity} is concerned with a realizable (in fact, exact) learning scenario. Therefore, this can trivially be made sound against arbitrarily powerful, malicious adversaries (compare the discussion in \cite{goldwasser2021interactive}). Namely, after running the protocol and obtaining some candidate string $\hat{s}$, the learner $\mathcal{L}$ can use a single private SQ (in fact, a single \emph{correlational} SQ \cite{feldman2008evolvability} suffices) to estimate the correlation between $f_{\hat{s}}$ and the unknown parity. If that estimate is sufficiently close to $1$, the learner accepts the interaction. Otherwise, the learner rejects the interaction.
Thus, the covert learner from \Cref{theorem:public-example-private-sq-parity} can trivially be made verifiable.

\subsection{Public Quantum Measurement Examples, Private Quantum Statistical Queries}\label{subsec:public-QMeasEX-private-QSQ}

As mentioned in the introduction, the definition of covert learning from public examples and private SQs readily generalizes to the case of public quantum measurement examples and private QSQs. The public oracle we consider $\mathsf{O}_{\mathrm{pri}}^{\mathrm{QMeasEx}}$, formally defined in \Cref{subsec:oracles}, is not a fully quantum oracle; when queried on a measurement, it outputs a sample drawn from the measurement statistics.
We now give the formal definition of the associated covert learning model:

\begin{definition}[Covert Exact Learning From Public Quantum Measurement Examples and Private QSQs--Formal version of \Cref{inf-definition:covert-learning-public-quantum-examples-private-qsq}]\label{definition:covert-learning-public-quantum-examples-private-qsq}
    An algorithm $\mathcal{L}$ that has access to a private QSQ oracle $\mathsf{O}_{\mathrm{pri}}^{\mathrm{QSQ}}$ and to a public quantum measurement example oracle $\mathsf{O}_{\mathrm{pri}}^{\mathrm{QMeasEx}}$ is an $(m_{\mathrm{pri}}, m_{\mathrm{pub}},\tau_{\mathrm{pri}},\delta_c,\delta_p)$-covert exact learner for a class of states $\mathcal{S}$ from private QSQs and public examples if it satisfies the following:
    \begin{itemize}
        \item $\delta_c$-\textbf{Completeness:} For any $\rho\in\mathcal{S}$, after making at most $m_{\mathrm{pri}}$ queries of tolerance at least $\tau_{\mathrm{pri}}$ to $\mathsf{O}_{\mathrm{pri}}^{\mathrm{QSQ}}(\rho)$ and at most $m_{\mathrm{pub}}$ queries to $\mathsf{O}_{\mathrm{pub}}^{\mathrm{QMeasEx}}(\rho)$, $\mathcal{L}$ outputs $\rho$ with success probability $\geq 1-\delta_c$.
        \item $\delta_p$-\textbf{Privacy:} For $\rho\sim \mathcal{S}$ a uniformly random state from the class, no adversary $\mathcal{A}$ can correctly guess $\rho$ with success probability $>\delta_p$ from the at most $m_{\mathrm{pub}}$ many public measurement examples requested by $\mathcal{L}$.
    \end{itemize}
\end{definition}

Note that by choosing $\mathsf{O}_{\mathrm{pri}}^{\mathrm{QMeasEx}}$ as the public oracle, all public oracle queries and responses are classical messages. The same applies to the private QSQ oracle $\mathsf{O}_{\mathrm{pri}}^{\mathrm{QSQ}}$. Thus, while these oracles arise from underlying quantum states, the corresponding learning problem defined in \Cref{definition:covert-learning-public-quantum-examples-private-qsq} is classical.

Mirroring the exponential separation between classical random examples and classical SQs due to \cite{kearns1998efficient}, \cite{arunachalam2023role} (see \cite{arunachalam2023role-arXiv} for the arXiv version) proved an exponential separation between quantum measurement examples and QSQs.\footnote{\cite{arunachalam2023role, nietner2023unifyingquantumstatisticalparametrized} also contain additional separations between quantum examples and QSQs for different learning and testing tasks.} While the former is for parity learning, the latter is for learning the class of $(n+1)$-qubit states
\begin{equation}
    \mathcal{S}_{\mathrm{quad}}
    = \left\{\rho_A=\ket{\psi_A}\bra{\psi_A} \, \Bigg\vert\, \ket{\psi_A}=\frac{1}{\sqrt{2^n}}\sum_{x\in\{0,1\}^n} \ket{x, x^\top Ax} \textrm{ for upper-triangular } A\in\{0,1\}^{n\times n}\right\}\, .
\end{equation}
Here, $x^\top Ax$ is to be understood modulo 2. 
Note that we can w.l.o.g. only consider upper-triangular $A$ in this definition because $x^\top A x = \sum_i A_{ii} x_i + \sum_{i\neq j}x_i A_{ij} x_j = \sum_i A_{ii} x_i + \sum_{i<j}x_i (A_{ij}+A_{ji}) x_j$.
The elements of $\mathcal{S}_{\mathrm{quad}}$ can be viewed as quantum example states for purely quadratic functions. With the following theorem, we show that the idea behind \Cref{theorem:public-example-private-sq-parity} extends to the setting of \Cref{definition:covert-learning-public-quantum-examples-private-qsq}, allowing for a covert learning algorithm that overcomes the limitations of QSQs for learning $\mathcal{S}_{\mathrm{quad}}$ with the help of a public quantum measurement example oracle.

\begin{theorem}[Covert Quadratic Function Learning from Public QMeasExs and Private QSQs---Formal version of \Cref{inftheorem:public-qmeasexample-private-qsq-quadratic}]\label{theorem:public-qmeasexample-private-qsq-quadratic}
    Let $n\in\mathbb{N}$.
    Let $\delta_c\in (0,1)$.
    There exists a computationally efficient $(m_{\mathrm{pri}} = n,m_{\mathrm{pub}} = \mathcal{O}(n + \log(1/\delta_c)), \tau_{\mathrm{pri}} = 1/3, \delta_c, \delta_p = 2^{-n})$-covert exact learner from private QSQs and public quantum measurement examples for the concept class $\mathcal{S}_{\mathrm{quad}}$.
    That is, there is a computationally efficient learner $\mathcal{L}$ that, for any unknown upper-triangular $A\in\{0,1\}^{n\times n}$, makes at most $\mathcal{O}(n + \log(1/\delta_c))$ many queries to a public quantum measurement example oracle $\mathsf{O}_{\mathrm{pub}}^{\mathrm{QMeasEx}}(\rho_A)$, followed by at most $n$ many queries with tolerance $1/3$ to a private parity QSQ oracle $\mathsf{O}_{\mathrm{pri}}^{\mathrm{QSQ}}(\rho_A)$ in such a way that the following holds:
    \begin{itemize}
        \item \textbf{Completeness:} With probability $\geq 1-\delta_c$, $\mathcal{L}$ learns the unknown matrix $A$.
        \item \textbf{Privacy:} For a uniformly random (upper-triangular) $A\in\{0,1\}^{n\times n}$, any adversary $\mathcal{A}$ that only sees $\mathcal{L}$'s public quantum measurement example oracle queries correctly guesses $A$ with probability at most $2^{-n}$.
    \end{itemize}
\end{theorem}

\begin{proof}
    First, we describe that protocol that $\mathcal{L}$ uses, then we prove that it has the claimed completeness and privacy properties.
    Let $A\in\{0,1\}^{n\times n}$ be the upper-triangular Boolean matrix specifying the unknown state $\rho_A=\ket{\psi_A}\bra{\psi_A}$.
    
    \textbf{Algorithm:} The procedure is as follows; it closely follows the algorithm from \cite[Fact 3]{arunachalam2023role-arXiv}:
    \begin{enumerate}
        \item Set $m_{\mathrm{pub}}=\lceil \frac{1}{\log_2(8/7)}\cdot (n + \log_2(1/\delta_c))\rceil$.
        \item Query $\mathsf{O}_{\mathrm{pub}}^{\mathrm{QMeasEx}}(\rho_A)$ for $m_{\mathrm{pub}}$ many times on the $2$-copy Bell sampling POVM $\{E_{y,z,b}\}_{y,z\in\{0,1\}^n,b\in\{0,1\}^2}$\footnote{The POVM can be written down explicitly as $E_{y,z,b} = F_{y,z,b}^\dagger F_{y,z,b}$ for $$F_{y,z,b}=B_{y,z,b}(\ket{b_1}\bra{b_1}_{n+1}\otimes \ket{b_2}\bra{b_2}_{2n+2})(H_{n+1}\otimes H_{2n+2})\, ,$$ 
        where we defined $B_{y,z,11} = \ket{z}\bra{z}_{1,\ldots,n}H^{\otimes n}_{1,\ldots,n} \ket{y}\bra{y}_{n+1,\ldots,2n+1}(\bigotimes_{i=1}^n \mathrm{CNOT_{i,n+1+i}})$ and $B_{y,z,b}=\ket{z}\bra{z}_{1,\ldots,n}\otimes \ket{y}\bra{y}_{n+1,\ldots,2n+1}$ for $b\neq 11$. 
        Here we used indices to denote the subsystems acted on, and we did not write out identity tensor factors. This description of the POVM arises as follows: $H_{n+1}\otimes H_{2n+2}$ first transforms from two copies of $\ket{\psi_A}$ to two copies of $\frac{1}{\sqrt{2}}(\ket{+}^{\otimes n}\ket{0} + \sum_{x\in\{0,1\}^{n}} (-1)^{x^\top Ax}\ket{x,1})$.
        We then measure the last qubits of both copies in the computational basis; if we observe outcome $b=11$, we proceed, otherwise we effectively abort. This effective post-selection on $b=11$ moves the function to the phase. The remaining operations now implement the standard Bell sampling procedure on two copies of the phase state, compare, e.g., the proof of \cite[Fact 3]{arunachalam2023role-arXiv}.}, receiving outcomes $(y_j,z_j,b_j)$ for i.i.d. uniformly random $y_j$ and $b_j$. 
        \item Consider the subsequence $(j_\ell)_\ell$ of rounds for which $b_{j_\ell}=11$. For these rounds, we have $z_{j_\ell} = y_{j_\ell}^\top (A+A^\top)  y_{j_\ell}$. If the $(y_{j_\ell})_{j_\ell}$ do not contain a set of $n$ linearly independent vectors (which $\mathcal{L}$ can check efficiently using Gaussian elimination), abort. Otherwise, continue.
        \item Use Gaussian elimination for the data points $(y_{j_\ell}, z_{j_\ell})=(y_{j_\ell}, y_{j_\ell}^\top (A+A^\top)  y_{j_\ell})$ to determine all off-diagonal entries of $A$, $A_{ij}$ for $i\neq j$ .
        \item For each $1\leq i\leq n$, query $\mathsf{O}_{\mathrm{pri}}^{\mathrm{QSQ}}(\rho_A)$ on the observable $UM_iU^\dagger$ with tolerance $\tau_{\mathrm{pri}} = 1/3$, obtaining outcome $\hat{I}_i$. Here, $M_i$ is the observable for estimating the influence of the $i$th variable (compare \cite[Lemma 4.2]{arunachalam2020quantumstatisticalquerylearning}), and $U=\prod_{1\leq i<j\leq n} U(A_{ij})$ is defined via $U(A_{ij})\ket{x,b}=\ket{x, b\oplus (x_iA_{ij}x_j)}$\footnote{Note that $\mathcal{L}$ can computationally efficiently produce a description of $UM_iU^\dagger$ and then use this description for the QSQs.}. 
        If $\hat{I}_i\leq 1/2$, conclude $A_{ii}=0$. If $\hat{I}_i> 1/2$, conclude $A_{ii}=1$.
        \item $\mathcal{L}$ outputs $\rho_A$ using the learned $n\times n$ matrix $A$.
    \end{enumerate}

    \textbf{Completeness:}
    We now prove the completeness guarantee for the above procedure. As the $y_j$ 
    are i.i.d. uniformly random strings in $\{0,1\}^n$ 
    and by our choice of ${m}_{\mathrm{pub}}$, the subsequence $(y_{j_\ell})_{j_\ell}$ contains a set of $n$ linearly independent vectors with probability $\geq 1-\delta_c$.
    The remainder of the completeness analysis for our procedure consists of combining the proofs from \cite[Fact 3]{arunachalam2023role-arXiv} and \cite[Lemma 4.2]{arunachalam2020quantumstatisticalquerylearning}), using that $U$ from Step 5 uncomputes the contributions from the off-diagonal terms in $A$ (learned in Step 4), leaving only the parity function $x\mapsto \sum_{i}A_{ii}x_i$, which we can learn by estimating the influence of every variable.

    \textbf{Privacy:} 
    Now we turn our attention to the privacy guarantee. 
    By $\mathcal{L}$'s procedure, the only information that the adversary $\mathcal{A}$ sees are the $(y_j,z_j,b_j)$ for $1\leq j\leq n$, where the $y_j$ and $b_j$ are i.i.d. uniformly random.
    If $b_j=11$, then $z_j=y_j^\top (A+A^\top)y_j$, so the data point contains no information about the diagonal entries of $A$.
    If $b_j\neq 11$, then $z_j$ is itself a uniformly random $n$-bit string that contains no information about $A$. 
    Thus, $\mathcal{A}$ has received no information about the diagonal entries of $A$, and the best they can do is to randomly guess those entries. The probability that $\mathcal{A}$ guesses all $n$ diagonal entries of $A$ correctly is $2^{-n}$. This is the claimed privacy guarantee.
\end{proof}

The above covert learning algorithm utilizes the learning procedure from \cite[Fact 3]{arunachalam2023role}, which is based on two-copy Bell sampling. 
We leave open the question whether the single-copy procedure for learning $\mathcal{S}_{\mathrm{quad}}$ via derivative learning from \cite[Theorem 3]{arunachalam2022optimal} can be made covert with a similarly strong covertness guarantee, achieving an exponentially small guessing probability for any adversary.
We also note that, for the same reason as discussed at the end of \Cref{subsec:public-classical-private-sq-parity}, the above covert learner can trivially be made verifiable against malicious adversaries.

\begin{remark}
    Our proof above used only one feature of the private oracle $\mathsf{O}_{\mathrm{pri}}^{\mathrm{QSQ}}$: one can efficiently learn parity functions from access to this oracle. Hence, the covert learning procedure above can, after a small modification, be obtained with other sufficiently strong private oracles.
    For instance, we could juxtapose the public quantum measurement example oracle $\mathsf{O}_{\mathrm{pub}}^{\mathrm{QMeasEx}}$ with a private classical example oracle $\mathsf{O}_{\mathrm{pri}}^{\mathrm{Ex}}$.
    In this case, the covert learner with access to both oracles would use $\mathcal{O}(n+\log(1/\delta_c))$ queries to $\mathsf{O}_{\mathrm{pub}}^{\mathrm{QMeasEx}}$ followed by $\mathcal{O}(n)$ queries to $\mathsf{O}_{\mathrm{pri}}^{\mathrm{Ex}}$. 
    As learning degree-$2$ polynomials in $n$ binary variables classically requires $\Omega(n^2)$ uniformly random examples, this would constitute a (mild) polynomial improvement in both sample and computational complexity while maintaining covertness.
\end{remark}

\begin{remark}In this subsection, we have focused on the standard notion of QSQs for quantum states introduced in \cite{arunachalam2020quantumstatisticalquerylearning}. We note that models of QSQ access to quantum channels have recently been proposed \cite{angrisani2025learningunitariesquantumstatistical, wadhwa2025process-qsq}, and our \Cref{definition:covert-learning-public-quantum-examples-private-qsq} can naturally be extended to QSQ learning of channels rather than states.
Additionally, as the exponential separation between QSQs and quantum measurement examples for quadratic function learning carries over from states to unitaries \cite[Section 7]{angrisani2025learningunitariesquantumstatistical}, one can immediately obtain an analogue of \Cref{theorem:public-qmeasexample-private-qsq-quadratic} for covertly learning certain quadratic function unitaries from public quantum measurement examples and private (Choi state) QSQs.
\end{remark}

\subsection{A Covert Learning Template and Further Quantum Applications}\label{subsec:covert-learning-template}

We conclude this section with a discussion of how the underlying principles of the previous two subsections can be applied more broadly.
The covert learning procedures from \Cref{theorem:public-example-private-sq-parity,theorem:public-qmeasexample-private-qsq-quadratic} suggest the following template for covert learning:
Suppose a class of objects $\mathcal{C}$ can be learned using $m_{1}$ queries to $\mathsf{O}_{1}$ and $m_{2}$ queries to $\mathsf{O}_{2}$, whereas learning the same task from $\mathsf{O}_{2}$ alone requires $m'_{2}\gg m_{1}+m_{2}$ queries. 
Further suppose that learning $\mathcal{C}$ from $\mathsf{O}_{1},\mathsf{O}_{2}$-access implicitly works with an underlying partition $\mathcal{C}=\bigsqcup_{i}\mathcal{C}_i$ by first using $m_{1}$ queries to $\mathsf{O}_{1}$-access to identify which $\mathcal{C}_i$ the unknown concept lies in and then solving the learning problem for that $\mathcal{C}_i$ using the at most $m_{2}$ queries to $\mathsf{O}_{2}$.
In this setting, there exists a covert learner for $\mathcal{C}$ that makes $m_{1}$ public queries to $\mathsf{O}_{\mathrm{pub}} = \mathsf{O}_{1}$ to determine $\mathcal{C}_i$, followed by at most $m_{2}$ queries to $\mathsf{O}_{\mathrm{pri}} = \mathsf{O}_{2}$ to recover the unknown concept. 
This comes with the covertness guarantee that any adversary $\mathcal{A}$ who sees only the public queries and answers can guess the unknown concept with success probability at most $(\min_i |\mathcal{C}|_i)^{-1}$.

This simple template can be instantiated in other scenarios in addition to those of \Cref{theorem:public-example-private-sq-parity,theorem:public-qmeasexample-private-qsq-quadratic}. For example, consider a setting in which $\mathsf{O}_{\mathrm{pub}}$ provides access to (outcomes of measurements performed on) multiple copies of an unknown quantum state, whereas $\mathsf{O}_{\mathrm{pri}}$ only provides access to (outcomes of measurements performed on) single copies. This difference in oracle access---which can also be viewed as a difference in terms of the quantum memory available to the learner---is by now known to lead to large separations in query complexity for a variety of tasks; see \cite{huang2021information, aharonov2022quantum, chen2022tight, chen2022exponential, huang2022quantum-advantage, chen2023unitarity, chen2024optimal, caro2022learning-PTM} for a non-exhaustive list. For some of these scenarios, a learner can use the above template to overcome the limitation of single-copy access in a covert manner by accessing a public multi-copy oracle. In the following, we highlight two such scenarios as representative examples.

First, we consider the task of Pauli shadow tomography \cite{huang2021information,king2025triply, chen2024optimal}, a special case of general shadow tomography \cite{aaronson2018shadow, badescu2024improved, abbas2023quantum, chen2024optimalhighprecisionshadowestimation, sinha2025dimension}. 
Formally, this task is defined as follows:

\begin{problem}[Pauli Shadow Tomography]\label{problem:pauli_shadow_tomography}
    Let $\mathcal{P}_n =\{\mathds{1}_2, X,Y,Z\}^{\otimes n}\}$ be the set of $n$-qubit Pauli operators. Pauli shadow tomography for a subset $\mathcal{S}\subseteq\mathcal{P}_n$ is the following problem: Given oracle access to an unknown $n$-qubit state $\rho$, output, with probability $\geq 1-\delta$, estimate $\hat{\alpha}_P$ such that $|\hat{\alpha}_P - \tr{P\rho}|\leq \varepsilon$, for all $P\in\mathcal{S}$.
\end{problem}

For $\mathcal{S}=\mathcal{P}_n$, this task requires exponentially-in-$n$ many single-copy oracle queries \cite{huang2021information, chen2022exponential}, but it can be solved from linearly-in-$n$ many two-copy oracle queries \cite{chen2024optimal, king2025triply}. The most recent such two-copy algorithms uses in the first phase, two-copy Bell sampling to identify Pauli operators that have non-negligible expectation values w.r.t.~the unknown state, and then in a second phase, employs single-copy access to determine the signs of those expectation values \cite{chen2024optimal, king2025triply}. Thus, following our above template, a learner can covertly make use of a public multi-copy oracle and a private single-copy oracle to solve \Cref{problem:pauli_shadow_tomography} from linearly-in-$n$ many queries, thus overcoming the exponential single-copy lower bound, while any eavesdropper is left guessing about the signs of the Pauli expectation values.
In more detail, following our template and the Pauli shadow tomography procedures from \cite{chen2024optimal, king2025triply}, we obtain the following result:

\begin{proposition}[Covert Pauli Tomography using Public Two-Copy and Private Singe-Copy Quantum Measurement Example Oracle]\label{proposition:covert-pauli-shadow-tomography}
    Let $n\in\mathbb{N}$. Let $\delta\in (0,1)$ and $\varepsilon\in (0,1/2)$.
    Let $\mathcal{S}\subseteq\mathcal{P}_n$.
    There exists a covert Pauli shadow tomography algorithm $\mathcal{L}$ from private single-copy quantum measurement examples and public two-copy quantum measurement examples that achieves the following guarantees:
    For any unknown $n$-qubit state $\rho$, $\mathcal{L}$ makes at most $\mathcal{O}((\log|\mathcal{S}| + \log(1/\delta))/\varepsilon^4)$ many queries to a public two-copy quantum measurement example oracle for $\rho$, followed by at most $\mathcal{O}((\log|\mathcal{S}| + \log(1/\delta))/\varepsilon^4)$ many queries to a private single-copy quantum measurement example oracle for $\rho$ in such a way that the following holds:
    \begin{itemize}
        \item \textbf{Completeness:} With probability $\geq 1-\delta$, $\mathcal{L}$ outputs estimates $\hat{\alpha}_P$ such that $|\hat{\alpha}_P - \tr{P\rho}|\leq \varepsilon ~\forall P\in\mathcal{S}$.
        \item \textbf{Privacy:} For a uniformly randomly chosen $\rho_{b,P} = \frac{\mathds{1}_2^{\otimes n} +2\varepsilon (-1)^{b}   P}{2^n}$, $b\sim\{0,1\}$ and $P\sim\mathcal{S}$, any adversary $\mathcal{A}$ that only sees $\mathcal{L}$'s public two-copy quantum measurement example oracle correctly guesses $b$ with probability at most $1/2$.
        Moreover, if $\varepsilon\leq (2\max_{b\in\{0,1\}^{|\mathcal{S}|}}\norm{\sum_{P\in\mathcal{S}} (-1)^{b_P} P}_\infty)^{-1}$, then, for a uniformly random $\rho_b = \frac{\mathds{1}_2^{\otimes n} + 2\varepsilon\sum_{P\in\mathcal{S}} (-1)^{b_P} P}{2^n}$, $b\sim\{0,1\}^{|\mathcal{S}|}$, any adversary $\mathcal{A}$ that only sees $\mathcal{L}$'s public two-copy quantum measurement example oracle queries correctly guesses $b$ with probability at most $1/2^{|\mathcal{S}|}$.\footnote{Notice that $\norm{\sum_{P\in\mathcal{S}} (-1)^{b_P} P}_\infty\leq |\mathcal{S}|$ by triangle inequality and, using $\norm{\cdot}_\infty\geq \norm{\cdot }_F/\sqrt{2^n}$, $\norm{\sum_{P\in\mathcal{S}} (-1)^{b_P} P}_\infty\geq \sqrt{|\mathcal{S}|}$. So, for the assumption $\varepsilon\leq (2\max_{b\in\{0,1\}^{|\mathcal{S}|}}\norm{\sum_{P\in\mathcal{S}} (-1)^{b_P} P}_\infty)^{-1}$ to hold, it is necessary that $2\varepsilon\leq 1/\sqrt{|\mathcal{S}|}$ and sufficient that $2\varepsilon\leq 1/|\mathcal{S}|$. Given the bound of $1/2^{|\mathcal{S}|}$, we see that we can ensure an adversary guessing probability of at most $\delta_p$ for $2\varepsilon\leq 1/\log(1/\delta_p)$.}
    \end{itemize}
\end{proposition}

Here, a two-copy (single-copy) quantum measurement example oracle is the natural variant of $\mathsf{O}^{\mathrm{QMeasEx}}$ in which the only admissible queries are two-copy (single-copy) POVMs. We note that \Cref{proposition:covert-pauli-shadow-tomography} also implicitly shows how \Cref{definition:covert-exact-learning-public-examples-private-sq,definition:covert-learning-public-quantum-examples-private-qsq} can easily be modified to formulate covertness in learning scenarios other than (exact) function/state learning.

\begin{remark}
    While the focus of this section is target-covertness, we note that the two-phase Pauli shadow tomography procedure outlined above also comes with a certain strategy-covertness. 
    Namely, the first phase (Bell sampling) is independent of the specific set $\mathcal{S}$ of Paulis of interest; it only depends on $|\mathcal{S}|$, which determines the number of Bell sampling measurements to be performed.
    Thus, any adversary $\mathcal{A}$ that sees only $\mathcal{L}$'s public two-copy quantum measurement example oracle queries does not gain any information about the set $\mathcal{S}$ other than its size. This means that the above procedure is also a sample-efficient (albeit in general computationally inefficiently) statistically covert QSQ algorithm for $\mathcal{S}$-queries from public (two-copy) quantum measurement examples (compare \Cref{definition:covert-qsq-model})---if we assume that the size of $\mathcal{S}$ is known in advance and given to the simulator as an input.
\end{remark}

As a second scenario in which our template applies, consider the problem of (exactly) learning an unknown $n$-qubit stabilizer state. This can be achieved with $\mathcal{O}(n)$ queries to a two-copy oracle \cite{montanaro2017learning} but requires $\Omega(n^2)$ queries when only a single-copy oracle is available \cite{arunachalam2022optimal}.
Here, learning stabilizer states with two-copy access proceeds by first using two-copy Bell difference sampling to determine a set of $n$ candidate stabilizer generators up to signs, and then uses further single-copy measurements to determine those signs \cite[Algorithm 1]{montanaro2017learning}, thus obtaining an efficient description that uniquely identifies the stabilizer state. 
Again, our template can be applied to obtain a covert variant of this procedure, allowing a learner with private single-copy access to benefit from public multi-copy access while keeping the information about the correct choice of signs for the stabilizer generator hidden from any eavesdropper. 
Concretely, this yields the following:

\begin{proposition}[Covert Stabilizer State Learning using Public Two-Copy and Private Singe-Copy Quantum Measurement Example Oracle]\label{proposition:covert-stabilizer-state-learning}
    Let $n\in\mathbb{N}$. 
    There exists a covert stabilizer state learning algorithm $\mathcal{L}$ from private single-copy quantum measurement examples and public two-copy quantum measurement examples that achieves the following guarantees:
    For any unknown $n$-qubit stabilizer state $\ket{\psi}$, $\mathcal{L}$ makes at most $\mathcal{O}(n)$ many queries to a public two-copy quantum measurement example oracle for $\ket{\psi}$, followed by at most $\mathcal{O}(n)$ many queries to a private single-copy quantum measurement example oracle for $\ket{\psi}$ in such a way that the following holds:
    \begin{itemize}
        \item \textbf{Completeness:} With probability $\geq 1-2^{-n}$, $\mathcal{L}$ outputs a valid set of $n$ stabilizer generators for $\ket{\psi}$.
        \item \textbf{Privacy:} For a uniformly random stabilizer state $\ket{\psi}$, any adversary $\mathcal{A}$ that only sees $\mathcal{L}$'s public two-copy quantum measurement example oracle correctly guesses $\ket{\psi}$ with probability at most $1/2^n$.
    \end{itemize}
\end{proposition}

Similarly, some of the recent algorithms for learning $t$-doped stabilizer states \cite{grewal2023efficient, leone2023learning, hangleiter2023bell} and for agnostic stabilizer state learning \cite{grewal2023improved}\footnote{See also \cite{chen2025stabilizer,bakshi2025learningclosestproductstate,grewal2025agnostictomographystabilizerproduct} for related work on agnostic learning of quantum states, and \cite{wadhwa2024agnosticprocesstomography} for extensions to agnostic channel learning.} can be viewed as first using two-copy Bell difference sampling to narrow down the set of candidate states and then identifying a suitable candidate using single-copy classical Clifford shadows \cite{huang2020predicting}. 
This again makes them amenable to our covert learning template.

While in this subsection we have considered problems of quantum state learning, our template naturally carries over to some quantum process learning tasks, such as Pauli transfer matrix learning \cite{caro2022learning-PTM} or learning Clifford unitaries 
from Choi state access. 
One may also consider our template in the context of Hamiltonian learning, where Bell sampling, which uses entangled input states, is a powerful subroutine for Hamiltonian structure learning, while coefficients can be learned from simple quantum experiments without auxiliary systems once the structure is known; see for instance \cite{bakshi2024structure, hu2025ansatzfree}. 
It remains interesting to identify and formalize further scenarios in which our template applies.

\section{Covert Verifiable Learning From Public Quantum Oracle Queries}\label{sec:public-quantum-oracle-private-classical-queries}

A \emph{quantum phase oracle $\mathsf{O}^{\mathrm{QPh}}(f)$} for a Boolean function $f:\{0,1\}^{n}\to\{0,1\}$ acts on a computational basis state as
\begin{equation}
  \mathsf{O}^{\mathrm{QPh}}(f): \ket{x}\mapsto (-1)^{f(x)}\ket{x}\, .
\end{equation}
Applying this oracle to the uniform superposition state on $n$ qubits yields the \emph{phase state} of the function $f$,
\begin{equation}
  \ket{\psi_f^{\mathrm{Ph}}} 
  = \mathsf{O}^{\mathrm{QPh}}(f) H^{\otimes n} \ket{0^n}
  = \frac{1}{\sqrt{2^{n}}}\sum_{x\in\{0,1\}^{n}}(-1)^{f(x)}\ket{x}\, .
\end{equation}
Such phase states serve as a key resource in several quantum learning routines (compare \cite{bshouty1995learning-DNF, arunachalam2017survey}, see also \cite{cross2015quantumlearning, grilo2019LWE,kanade2019learningDNFs,caro2020quantum, caro2024verification} for non-uniform or noisy variants) and more generally in quantum algorithms such as \cite{deutsch1992rapid, bernstein1997complexity, simon1997power, aaronson2015forrelation}.

Our objective in this section is for the learner to obtain copies of this phase state from a \emph{public} quantum phase oracle in a \emph{target-covert} way, that is, while preventing an adversary from learning about function $f$. Unlike the settings considered earlier in the paper, because the channel between the learner and the oracle is \emph{quantum}, an eavesdropper cannot passively observe the transcript of interaction; it must \emph{actively} modify the state to extract information. Therefore, target-covert protocols in this setting must ensure not only that they leak no information about $f$ but also actively check that the correct state is received by the learner. That is, we must ensure both privacy and verifiability. To achieve this, we equip the learner with a weaker classical private oracle and allow the learner to abort the interaction, which together will enable state verification and the detection of covertness violations. These requirements are similar to the covert \emph{verifiable} learning framework of \cite{canetti2021covert}, which allows the adversary to modify the responses to classical membership queries sent by the learner. We formalize our setting with the following definition:

\begin{definition}[Covert Verifiable Quantum Data From Public Quantum Oracle Queries---Restating \Cref{inf-definition:covert-quantum-data-public-quantum-oracle-private-classical-queries}]\label{definition:covert-quantum-data-public-quantum-oracle-private-classical-queries}
    A quantum algorithm $\mathcal{L}$ that has access to a private classical membership query oracle $\mathsf{O}^{\mathrm{Mem}}_{\mathrm{pri}}$ and to a public quantum phase oracle $\mathsf{O}_{\mathrm{pub}}^{\mathrm{QPh}}$ is a $(m_{\mathrm{pri}}, m_{\mathrm{pub}},\delta_c,\delta_s, \varepsilon, (\delta_{\mathrm{leak}}))$-covert verifiable procedure for producing $m$ quantum phase state copies for a concept class $\mathcal{F}$ against adversary $\mathcal{A}$ if it satisfies the following:
    \begin{itemize}
        \item $(\varepsilon,\delta_c)$-\textbf{Completeness:} For any $f\in\mathcal{F}$, after making at most $m_{\mathrm{pri}}$ queries to $\mathsf{O}_{\mathrm{pri}}^{\mathrm{Mem}}(f)$ and at most $m_{\mathrm{pub}}$ queries to $\mathsf{O}_{\mathrm{pub}}^{\mathrm{QPh}}(f)$, $\mathcal{L}$ accepts and outputs a state $\rho$ such that $\bra{\psi_f^{\mathrm{Ph}}}^{\otimes m}\rho\ket{\psi_f^{\mathrm{Ph}}}^{\otimes m}\geq 1-\varepsilon$ with success probability $\geq 1-\delta_c$.
        \item $(\varepsilon,\delta_s)$-\textbf{Soundness:} For any $f\in\mathcal{F}$, after making at most $m_{\mathrm{pri}}$ queries to $\mathsf{O}_{\mathrm{pri}}^{\mathrm{Mem}}(f)$ and at most $m_{\mathrm{pub}}$ queries to $\mathsf{O}_{\mathrm{pub}}^{\mathrm{QPh}}(f)$, the latter of which are subject to corruption by the adversary $\mathcal{A}$, $\mathcal{L}$ accepts and outputs some $\rho$ with $\bra{\psi_f^{\mathrm{Ph}}}^{\otimes m}\rho\ket{\psi_f^{\mathrm{Ph}}}^{\otimes m}< 1-\varepsilon$ with failure probability $\leq \delta_s$.
        \item \textbf{{Privacy}:}
        For $F\sim\mathcal{F}$ a randomly drawn concept, \ldots 
        \begin{itemize}
            \item \textbf{Version 1:} \ldots the adversary gains no information about $F$, in the sense that the joint state $\rho_{\mathsf{FA}}$ of the classical function register and the adversary's (in general quantum) register factorizes as $\rho_{\mathsf{FA}}=\frac{\mathds{1}_{\mathsf{F}}}{|\mathcal{F}|}\otimes \rho_{\mathsf{A}}$.
            \item \textbf{Version 2:} \ldots if the adversary has a probability of at least $\delta_{\mathrm{leak}}$ of extracting information about $F$ in every interaction, then $\mathcal{L}$ accepts with failure probability $\leq \delta_s$. 
        \end{itemize}
    \end{itemize}
\end{definition}

\textbf{Version~1} is a strictly stronger guarantee: regardless of the adversary’s strategy and even if the protocol aborts, the adversary learns nothing about $F$ (the joint state factorizes as stated). By contrast, \textbf{Version~2} is a \emph{cheat-sensitive} guarantee: if the adversary attempts to extract information about $F$ with probability at least $\delta_{\mathrm{leak}}$, then the learner detects the violation and aborts with high probability. As stated, for $\delta_{\mathrm{leak}}$ to be well defined, Version~2 implicitly assumes that we can lower bound the per-round information-extraction probability\footnote{This poses a restriction on the adversary that does not exist in Version 1: it must attempt to extract information in each round for $\delta_{\mathrm{leak}}$ to be non-zero.}. 

The choice of $\mathcal{A}$ in the above definition defines the power and capabilities of the adversary against which we obtain covertness. Ideally, we want to provide Privacy Version 1 against powerful adversaries with unbounded computational power, quantum memory of their own, and the ability to ``see'' and modify both directions of the communication with the oracle: the query as well as the response. However, in the worst-case, the next observation demonstrates that such an adversary has a simple strategy that prevents information-theoretic or even computational covertness. 

\begin{observation}[Impossibility of Covertness Against General Quantum Adversaries]
\label{obs:no-covert-qmem}
    Let $\mathcal{F}$ be any class of Boolean functions with $|\mathcal{F}|>1$ that can be efficiently exactly learned from a \emph{single} phase state copy. 
    Then, no learner for $\mathcal{F}$ that makes at least a single query to a public quantum phase oracle $\mathsf{O}^{\mathrm{QPh}}_{
    \mathrm{pub}
    }$ can simultaneously satisfy the completeness and privacy requirement from \Cref{definition:covert-quantum-data-public-quantum-oracle-private-classical-queries} (either version) against an adversary with $n$ qubits of quantum memory that can modify both directions of the learner-oracle quantum communication. 
\end{observation}

\begin{proof}[Proof sketch]
    As Version 2 of the privacy requirement is the less restrictive one, it suffices to show the impossibility result for this version.
    Suppose the learner satisfies the completeness requirement from \Cref{definition:covert-quantum-data-public-quantum-oracle-private-classical-queries} with $\delta_c>0$.
    Suppose further that the learner $\mathcal{L}$ queries $\mathsf{O}^{\mathrm{QPh}}_{\mathrm{pub}}$ on the $\mathsf{L}_1$-register of a state $\ket{\psi}_{\mathsf{L}_1,\mathsf{L}_2}$.
    An adversary $\mathcal{A}$ with quantum memory that can modify both directions of the learner-oracle quantum communication can now proceed as follows in every round:
    
    \begin{enumerate}
      \item $\mathcal{A}$ prepares $H^{\otimes n}\ket{0^n} = \frac{1}{\sqrt{2^{n}}}\sum_{x\in\{0,1\}^{n}}\ket{x}$ in its quantum memory register $\mathsf{A}$.
      \item Upon receiving the $\mathsf{L}_1$-register of $\ket{\psi}_{\mathsf{L}_1,\mathsf{L}_2}$ intended for the public oracle, $\mathcal{A}$ swaps it with the uniform superposition from its own quantum memory by applying $\mathrm{SWAP}_{\mathsf{L}_1, \mathsf{A}}$. 
      \item $\mathcal{A}$ now sends the $\mathsf{L}_1$-register to $\mathsf{O}^{\mathrm{QPh}}_{\mathrm{pub}}(f)$ and receives the corresponding phase state $\ket{\psi_f^{\mathrm{Ph}}}$.
      \item By assumption, $\mathcal{A}$ can now apply a quantum learning algorithm to the single copy of $\ket{\psi_f^{\mathrm{Ph}}}$ to exactly learn $f$.
      \item Having exactly learned $f$, $\mathcal{A}$ can now perfectly simulate the action of $\mathsf{O}^{\mathrm{QPh}}_{\mathrm{pub}}(f)$ on the state currently stored in the $\mathsf{A}$-register and on every subsequent learner query.
    \end{enumerate}

    From the learner’s viewpoint, the transcript and reduced states are indistinguishable from an execution without an adversary; hence, by completeness, $\mathcal{L}$ accepts with probability at least $1-\delta_c$. However, the adversary extracts information about $F$ with probability $1$ in every round (i.e., $\delta_{\mathrm{leak}}=1$), so Version~2 would require that acceptance occur with probability at most $\delta_s$. Thus, the learner cannot satisfy Version 2 of the privacy requirement from \Cref{definition:covert-quantum-data-public-quantum-oracle-private-classical-queries}, against such an adversary $\mathcal{A}$. 
\end{proof}

We can derive variants of \Cref{obs:no-covert-qmem} under looser learning requirements; for example, when the target class is only learnable with high (rather than unit) success probability, or when the adversary has sufficient prior information that a \emph{single} phase-state copy suffices. However, even the stated \Cref{obs:no-covert-qmem} has consequences regarding the feasibility of covert learning. For linear (parity) functions, the Bernstein–Vazirani algorithm \cite{bernstein1997complexity} succeeds with no error from one phase state copy. Consequently, learning linear functions cannot be made covert against general adversaries in the sense of \Cref{definition:covert-quantum-data-public-quantum-oracle-private-classical-queries}. An analogous impossibility holds for the Deutsch–Jozsa promise problem: a single phase query decides with certainty whether a function $f$ is constant or balanced \cite{deutsch1992rapid,cleve1998quantum}, so no protocol can simultaneously satisfy completeness and privacy against general adversaries for this distinguishing task as well.

We conclude from \Cref{obs:no-covert-qmem}: To covertly obtain a phase state, we must in general either restrict the adversary or impose additional conditions on the learning problem. In the simple strategy described above, the adversary needed both its own quantum memory (to perform the swap step) and the ability to tamper with the quantum communication between learner and oracle in both directions. We will achieve covertness without limiting the learning task itself by considering adversaries that are weakened along one of these dimensions:

\begin{enumerate}
    \item General adversaries that can tamper only with the quantum communication from the oracle to the learner (see Section~\ref{sec:non-iid-onlybackwards}),
    \item i.i.d.\ ancilla-free adversaries that can tamper with both directions of quantum communication (see Section~\ref{sec:iid-both-directions}).
\end{enumerate}

Our two covert verifiable protocols differ in the nature of privacy guarantee and how it is achieved, but they employ the same primitive to verify the state received by the learner. Section~\ref{sec:verification-shadows} develops 
versions of the shadow overlap certification protocol \cite{huang2025certifying-nature} that will form the core subroutine of our procedures, and the following sections combine these primitives with additional privacy analyses to obtain the two covert verifiable protocols for acquiring quantum data. In Section \ref{sec:phase-state-apps}, we then discuss concrete quantum algorithms that can be made covert and verifiable using these protocols while maintaining the large quantum advantage. 

\subsection{Phase State Certification} \label{sec:verification-shadows}

For state certification, we equip the learner with a private classical membership query oracle $\mathsf{O}^{\mathrm{Mem}}_{\mathrm{pri}}$. With this oracle, the shadow overlap estimation protocol of \cite{huang2025certifying-nature} (which was first published as \cite{huang2024certifying}; we will mostly refer to the corresponding preprint \cite{huang2024certifying-arXiv}) can be used as a phase state certification procedure from i.i.d.\ copies using single-qubit measurements. To certify a phase state in the non‑i.i.d.\ setting, we invoke a result from \cite{fawzi2024learning}, which gives a general approach to transform any non-adaptive incoherent quantum algorithm designed to work on i.i.d.~copies into a non-i.i.d.~version thereof. The following two subsections review these results from prior work and develop them specifically for the certification of phase states. 

\begin{remark}\label{remark:phase-states-vs-quantum-examples}
    In this subsection, we focus on certifying phase states obtained by querying a public quantum phase oracle $\mathsf{O}_{\mathrm{pub}}^{\mathrm{QPh}}$. However, our certification results easily carry over to quantum example states, that is, states of the form $\ket{\psi_f^{\mathrm{Ex}}}\coloneqq \frac{1}{\sqrt{2^n}}\sum_{x\in\{0,1\}^n} \ket{x,f(x)}$, where $f:\{0,1\}^n\to\{0,1\}^m$. These states naturally arise when querying a quantum membership query oracle $\mathsf{O}_{\mathrm{pub}}^{\mathrm{QMem}}$ on a uniform superposition.
    That our certification results indeed extend to quantum example states is most easily seen via the unitary equivalence between phase states and example states (compare, e.g., the discussion in \cite[Section 5.1]{caro2024testingclassicalpropertiesquantum}).
    Namely,
    \begin{equation}\label{eq:equivalence-phase-qex}
        (\mathds{1}_2^{\otimes n}\otimes H^{\otimes m})\ket{\psi_f^{\mathrm{Ex}}}
        = \frac{1}{\sqrt{2^{n+m}}}\sum_{x\in\{0,1\}^n} \sum_{y\in\{0,1\}^m}(-1)^{y\cdot f(x)}\ket{x,y}\, ,
    \end{equation}
    which is clearly a phase state for the function $\tilde{f}:\{0,1\}^{n+m}\to\{0,1\}$, $\tilde{f}(x,y)=y\cdot f(x)$.
    Thus, when the public quantum oracle is a membership query oracle, the learner can simply first transform the states by $\mathds{1}_2^{\otimes n}\otimes H^{\otimes m}$, certify the corresponding phase state with the procedures developed below (note that classical query access to $f$ enables simulation of classical query access to $\tilde{f}$), and, if certification is passed, transform back with another application of $\mathds{1}_2^{\otimes n}\otimes H^{\otimes m}$ to recover a certified example state. 
\end{remark}

\subsubsection{Certification With i.i.d. Copies}

\cite{huang2024certifying-arXiv} introduced a surrogate measure for the fidelity between quantum states that they termed ``shadow overlap'' and proved a key relation between this surrogate and the true fidelity. 
We state this relationship here specifically for phase states.

\begin{theorem}[{\cite[Theorem 4 +  Lemma 27]{huang2024certifying-arXiv}}]\label{theorem:shadow-overlap-completeness-and-soundness}
    Let $f:\{0,1\}^n\to\{0,1\}$ be an unknown Boolean function and let $\ket{\psi_f^{\mathrm{Ph}}}$ denote the associated $n$-qubit phase state.
    The empirical shadow overlap $\omega$ estimated from i.i.d.~copies of an unknown $n$-qubit state $\rho$ and from classical query access to $f$ via \cite[Protocol 1]{huang2024certifying-arXiv} is a random variable whose expectation value $\mathbb{E}[\omega]$, the so-called shadow overlap, satisfies
    \begin{equation}
        \begin{cases}
            \textrm{if }~ \mathbb{E}[\omega]\geq 1-\varepsilon, \quad &\textrm{then } \bra{\psi_f^{\mathrm{Ph}}}\rho\ket{\psi_f^{\mathrm{Ph}}}\geq 1-n\varepsilon\\
            \textrm{if } \bra{\psi_f^{\mathrm{Ph}}}\rho\ket{\psi_f^{\mathrm{Ph}}}\geq 1-\varepsilon, &\textrm{then } \mathbb{E}[\omega]\geq 1-\varepsilon
        \end{cases}\, .
    \end{equation}
    Moreover, $\mathbb{E}[\omega]$ can be written as $\mathbb{E}[\omega]=\tr[L\rho]$ for an ($f$-dependent) observable $L$ satisfying $0\leq L\leq \mathds{1}$.
\end{theorem}

To obtain algorithmic guarantees in terms of an empirical estimate of the shadow overlap, \cite{huang2024certifying-arXiv} combine the connection between the shadow overlap and the true fidelity with concentration bounds. Again phrased for the special case of phase states, this yields:

\begin{theorem}[{\cite[Theorems 5 and 6 +  Lemma 27]{huang2024certifying-arXiv}}]\label{theorem:shadow-overlap-copy-complexity}
    Let $f:\{0,1\}^n\to\{0,1\}$ be an unknown Boolean function and let $\ket{\psi_f^{\mathrm{Ph}}}$ denote the associated $n$-qubit phase state.
    There is a procedure that, given classical query access to $f$, uses 
    \begin{enumerate}
        \item either non-adaptive single-qubit Pauli measurements on $\mathcal{O}\left(\frac{n^2 \log(1/\delta)}{\varepsilon^2}\right)$ many i.i.d.~copies, 
        \item or adaptive single-qubit measurements on $\mathcal{O}\left(\frac{n \log(1/\delta)}{\varepsilon}\right)$ many i.i.d.~copies 
    \end{enumerate} 
    of an unknown $n$-qubit state $\rho$ to produce, with success probability $\geq 1-\delta$, an $(\varepsilon/4n)$-accurate estimate $\hat{\omega}$ of the shadow overlap $\mathbb{E}[\omega]$.
    This suffices to produce, again with success probability $\geq 1-\delta$, the output \textsf{FAILED} if $\bra{\psi_f^{\mathrm{Ph}}}\rho\ket{\psi_f^{\mathrm{Ph}}}<1-\varepsilon$ and the output \textsf{CERTIFIED} if $\bra{\psi_f^{\mathrm{Ph}}}\rho\ket{\psi_f^{\mathrm{Ph}}}\geq 1 - \frac{\varepsilon}{2n}$.
    Moreover, if $\bra{\psi_f^{\mathrm{Ph}}}\rho\ket{\psi_f^{\mathrm{Ph}}}=1$, we produce output \textsf{CERTIFIED} with success probability $1$.
\end{theorem}

Therefore, under an i.i.d.~assumption, we can, with inverse-exponential in $n$ failure probability, produce a certified phase state copy employing \Cref{theorem:shadow-overlap-copy-complexity} by using $\poly(n,1/\varepsilon)$ copies of the unknown state $\rho$. 
Moreover, as $\ket{\psi_f^{\mathrm{Ph}}}^{\otimes m}$ is itself an $(mn)$-qubit phase state, we can produce a certified state $\ket{\psi_f^{\mathrm{Ph}}}^{\otimes m}$ using $\poly(m, n,1/\varepsilon)$ copies. 
In the language of \Cref{definition:covert-quantum-data-public-quantum-oracle-private-classical-queries}, this already achieves completeness and soundness against i.i.d.~adversaries.
We will use this procedure directly in \Cref{sec:iid-both-directions} when studying covert verifiable quantum learning against ancilla-free i.i.d.~adversaries. 

\subsubsection{Certification With Non-i.i.d. Copies}

As per \Cref{theorem:shadow-overlap-copy-complexity}, under the i.i.d.~assumption, there is a way to certify phase states efficiently using only non-adaptive, incoherent measurements. This allows us to lift the certification protocol to handle non-i.i.d.~quantum data, by combining it with the following template. 

\begin{theorem}[{\cite[Algorithm 1, Theorem 4.8 and Remark 4.10]{fawzi2024learningpreprint}}, specified to our scenario of interest]\label{theorem:extending-iid-to-noniid}
    Let $\mathcal{L}$ be a quantum algorithm that performs non-adaptive, incoherent measurements on $k_\mathcal{L}$ many copies of an unknown $n$-qubit state $\rho$, followed by classical post-processing, to produce, with success probability $\geq 1-\frac{\delta}{6}$, an $(\varepsilon/2)$-accurate estimate $\hat{p}$ of an expectation value $\tr[L\rho]$, where $0\leq L\leq \mathds{1}$ is an observable.
    Then, there is a quantum algorithm $\mathcal{L}'$ that, given a permutation-invariant state $\rho^{A_1\ldots A_N}$ with $N=\tilde{\mathcal{O}}\left(\frac{n k_\mathcal{L}^2 \log^2(1 / \delta)}{\delta^2 \varepsilon^2}\right)$ many $n$-qubit subsystems, performs non-adaptive, incoherent measurements of the same type as those used by $\mathcal{L}$ on the first $N-1$ subsystems $\rho^{A_1\ldots A_{N-1}}$ and, with success probability $\geq 1-\delta$, produces an estimate $\hat{p}$ and calibration information $c$ such that $\left\lvert\hat{p} - \tr[L\rho_{c,p}^{A_N}]\right\rvert\leq \varepsilon$, where $\rho_{c,\hat{p}}^{A_N}$ is the post-measurement state on subsystem $A_N$.
    The estimate $\hat{p}$ is produced by classically post-processing $k_\mathcal{L}$ of the obtained measurement outcomes with the same classical post-processing as in $\mathcal{L}$.
\end{theorem}

Our next result illustrates how to combine the two ingredients above and obtain an efficient procedure (\Cref{alg:certify-state}) that produces one copy of a phase state even from non-i.i.d.~quantum data, in the sense that it achieves completeness and soundness.
As \Cref{theorem:shadow-overlap-copy-complexity} uses a classical membership query oracle, so does our non-i.i.d.~version. And as the translation from i.i.d.~to non-i.i.d~ achieved by \Cref{theorem:extending-iid-to-noniid} comes with an overhead in the number of copies, our procedure inherits this overhead. 

\begin{algorithm}
    \caption{Certified public quantum phase oracle query---$\ComplexityFont{CertifyState}^{\mathsf{O}^{\mathrm{Mem}}(f)}_{ \varepsilon, \delta}$}
    \begin{algorithmic}[1]
        \Require Number of qubits $n\in\mathbb{N}$; access to oracle $\mathsf{O}^{\mathrm{Mem}}(f)$; confidence parameter $\delta\in (0,1)$; accuracy parameter $\varepsilon\in (0,1)$; $(Nn)$-qubit state $\rho$.
        \Ensure ``reject'' or copy of an $n$-qubit state
        \State $\rho\gets \frac{1}{N!}\sum_{\pi\in S_N} U_\pi \rho U_\pi^\dagger$ \Comment{Randomly permute the $N$ many $n$-qubit subsystems of the input state.}
        \State Run $\mathcal{L}'^{\mathsf{O}^{\mathrm{Mem}}(f)}$ using  from \Cref{theorem:extending-iid-to-noniid}, where $\mathcal{L}$ is the non-adaptive procedure from \Cref{theorem:shadow-overlap-copy-complexity} for confidence parameter $\delta/6$ and accuracy parameter $\varepsilon/2$. Obtain an estimate $\hat{\omega}$, calibration information $c$, and an $n$-qubit post-measurement state $\rho_{c,\hat{\omega}}^{A_N}$.
        \If{$\hat{\omega}<1 - \frac{3\varepsilon}{4n}$}
            \State Output ``reject''.
        \Else
            \State Output $\rho_{c,\hat{\omega}}^{A_N}$.
        \EndIf
    \end{algorithmic}\label{alg:certify-state}
\end{algorithm}

\begin{theorem}[{Non-i.i.d.~Phase State Certification Subroutine}] \label{theorem:verified-shadow-overlap-algorithm}
    Let $n\in\mathbb{N}$ and $\delta,\varepsilon\in (0,1)$.
    $\ComplexityFont{CertifyState}^{\mathsf{O}^{\mathrm{Mem}}(f)}_{ \varepsilon, \delta}$ (\Cref{alg:certify-state}), when given as input an $(Nn)$-qubit state $\rho$ with $N=\Tilde{\mathcal{O}}\left(\frac{n^5}{\delta^2 \varepsilon^6}\right)$ and given access to a membership oracle $\mathsf{O}^{\mathrm{Mem}}(f)$ for a Boolean function $f:\{0,1\}^n\to\{0,1\}$, has the following guarantees:
    \begin{itemize}
        \item \textbf{Completeness:} If $\rho = \ket{\psi_f^{\mathrm{Ph}}}\bra{\psi_f^{\mathrm{Ph}}}^{\otimes N}$ then, 
        \begin{equation}
            \Pr\left[\ComplexityFont{CertifyState}^{\mathsf{O}^{\mathrm{Mem}}(f)}_{ \varepsilon, \delta}(\rho)\text{ accepts and } \ComplexityFont{CertifyState}^{\mathsf{O}^{\mathrm{Mem}}(f)}_{ \varepsilon, \delta}(\rho) = \ket{\psi_f^{\mathrm{Ph}}}\bra{\psi_f^{\mathrm{Ph}}}\right]
            =1\, .
        \end{equation}
        \item \textbf{Soundness:} For any $(Nn)$-qubit state $\rho$, 
        \begin{equation}
            \Pr\left[\ComplexityFont{CertifyState}^{\mathsf{O}^{\mathrm{Mem}}(f)}_{ \varepsilon, \delta}(\rho)\text{ accepts and } \bra{\psi_f^{\mathrm{Ph}}}\ComplexityFont{CertifyState}^{\mathsf{O}^{\mathrm{Mem}}(f)}_{ \varepsilon, \delta}(\rho)\ket{\psi_f^{\mathrm{Ph}}} < 1 - \varepsilon  \right]
            \leq \delta .
        \end{equation}
        \item \textbf{Efficiency: } $\ComplexityFont{CertifyState}^{\mathsf{O}^{\mathrm{Mem}}(f)}_{ \varepsilon, \delta}(\rho)$ makes $\mathcal{O}(\poly(n, 1/\delta, 1/\varepsilon))$ many calls to $\mathsf{O}^{\mathrm{Mem}}(f)$ and runs in time $\mathcal{O}(\poly(n, 1/\delta, 1/\varepsilon))$. Moreover, all measurements performed by the algorithm are non-adaptive single-qubit Pauli measurements. 
    \end{itemize}
\end{theorem}
\begin{proof}
    We take $N=\tilde{\mathcal{O}}\left(\frac{n (n^2 \log(1/\delta) / \varepsilon^2)^2 \log^2( 1 / \delta)}{\delta^2 \varepsilon^2}\right) =\Tilde{\mathcal{O}}\left(\frac{n^5}{\delta^2 \varepsilon^6}\right)$ to be the number of copies that we get from combining \Cref{theorem:shadow-overlap-copy-complexity,theorem:extending-iid-to-noniid}. With this copy complexity, we now first prove completeness and soundness, then we discuss the efficiency of the procedure.

    \textbf{Completeness:} Assume $\rho = \ket{\psi_f^{\mathrm{Ph}}}\bra{\psi_f^{\mathrm{Ph}}}^{\otimes N}$ consists of $N$ i.i.d.~copies of the ideal state $\ket{\psi_f^{\mathrm{Ph}}}$. In this case, the post-measurement state on the last $n$-qubit subsystem is simply $\rho_{c,\hat{\omega}}^{A_N}=\ket{\psi_f^{\mathrm{Ph}}}\bra{\psi_f^{\mathrm{Ph}}}$ as the copies are independent and unentangled. 
    Moreover, unless $\mathcal{L}'^{\mathsf{O}^{\mathrm{Mem}}(f)}$ aborts before measuring, all measurements performed by it are performed on i.i.d.~copies of $\ket{\psi_f^{\mathrm{Ph}}}$ and hence will lead to an empirical shadow overlap $\hat{\omega}=1$. The algorithm will accept and output $\ket{\psi_f^{\mathrm{Ph}}}$ contained in the last subsystem. 
    Thus, the only mode of failure in \cref{alg:certify-state} is a potential reject coming from the additional pre-processing in \cite[Algorithm 1]{fawzi2024learningpreprint}. To analyze this possibility, let us examine the algorithm $\mathcal{L}'^{\mathsf{O}^{\mathrm{Mem}}(f)}$ in more detail.
    
    According to \cite[Algorithm 1]{fawzi2024learningpreprint}, before making any measurements, $\mathcal{L}'^{\mathsf{O}^{\mathrm{Mem}}(f)}$ first samples random (indices of) measurements $r_1,\ldots,r_{k_\mathcal{L}\log(6k_\mathcal{L}/\delta)}$ from all measurements of the original (i.i.d.~version) of the algorithm indexed by $\{1,\ldots,k_\mathcal{L}\}$. If $r_1,\ldots,r_{k_\mathcal{L}\log(6k_\mathcal{L}/\delta)}$ does not cover the entire set $\{1,\ldots,k_\mathcal{L}\}$, $\mathcal{L}'^{\mathsf{O}^{\mathrm{Mem}}(f)}$ fails. Note that, by the coupon collector problem, the probability of this event is small, and the failure probability bound provided by \cref{theorem:extending-iid-to-noniid} already takes this into account. 
    Thus, to achieve completeness for our purposes, we consider the following minor adjustment to \cite[Algorithm 1]{fawzi2024learningpreprint}: If $r_1,\ldots,r_{k_\mathcal{L}\log(6k_\mathcal{L}/\delta)}$ does \emph{not} cover the entire set $\{1,\ldots,k_\mathcal{L}\}$, we run the i.i.d.~certification procedure on any $k_{\mathcal{L}}$ many subsystems selected uniformly at random from the first $k_\mathcal{L}\log(6k_\mathcal{L}/\delta)$ subsystems. 
    In the completeness case, each subsystem contains one copy of $\ket{\psi_f^{\mathrm{Ph}}}$, so the estimator still returns $\hat{\omega}=1$, and we accept. 
    This modification leaves soundness untouched because the original error probability bound already accounts for the (low-probability) event of failing to cover all of $\{1,\ldots,k_\mathcal{L}\}$; therefore, the behavior when the low‑probability event occurs is irrelevant.

    \textbf{Soundness:} Let $\rho=\rho^{A_1\ldots A_N}$ be an arbitrary $(Nn)$-qubit state. Combining the guarantees from \Cref{theorem:shadow-overlap-copy-complexity,theorem:extending-iid-to-noniid}, we see that the chosen number of copies, $N$, suffices to ensure: With probability $\geq 1-\delta$, Step 2 of  $\ComplexityFont{CertifyState}^{\mathsf{O}^{\mathrm{Mem}}(f)}_{ \varepsilon, \delta}(\rho)$ produces an estimate $\hat{\omega}$, calibration information $c$, and an $n$-qubit post-measurement state $\rho_{c,\hat{\omega}}^{A_N}$ on the $N$th subsystem such that $\left\lvert\hat{\omega} - \tr[L\rho_{c,\hat{\omega}}^{A_N}]\right\rvert\leq \varepsilon/4n$, where $L$ is the observable from \Cref{theorem:shadow-overlap-completeness-and-soundness}.
    As \Cref{theorem:shadow-overlap-completeness-and-soundness} also tells us that $\tr[L \rho_{c,\hat{\omega}}^{A_N}]=\mathbb{E}[\omega]$ exactly equals the shadow overlap estimated from copies of $\rho_{c,\hat{\omega}}^{A_N}$, we can now rely on the connections between fidelity and expected overlap established in that same theorem. These imply:
    If $\bra{\psi_f^{\mathrm{Ph}}}\rho_{c,\hat{\omega}}^{A_N}\ket{\psi_f^{\mathrm{Ph}}}<1-\varepsilon$, then $\tr[L \rho_{c,\hat{\omega}}^{A_N}]=\mathbb{E}[\omega]\leq 1-\frac{\varepsilon}{n}$, so $\hat{\omega}<1-\frac{3\varepsilon}{4n}$. Therefore, if $\bra{\psi_f^{\mathrm{Ph}}}\rho_{c,\hat{\omega}}^{A_N}\ket{\psi_f^{\mathrm{Ph}}}<1-\varepsilon$, we reject with probability $\geq 1- \delta$. This yields soundness.

    \textbf{Efficiency:} The shadow overlap estimation procedure behind \Cref{theorem:shadow-overlap-copy-complexity} ($\mathcal{L}$ in Step 2 of \Cref{alg:certify-state}) queries $\mathsf{O}^{\mathrm{Mem}}(f)$ only during the classical post-processing of the obtained measurement outcomes, and it does so at most $k_\mathcal{L} = \mathcal{O}\left(\frac{n^2 \log(1/\delta)}{\varepsilon^2}\right)$ many times. 
    According to \Cref{theorem:extending-iid-to-noniid}, $\mathcal{L}'^{\mathsf{O}^{\mathrm{Mem}}(f)}$ only runs the classical post-processing part of $\mathcal{L}$ once and hence, also queries $\mathsf{O}^{\mathrm{Mem}}(f)$ at most $\mathcal{O}\left(\frac{n^2 \log(1/\delta)}{\varepsilon^2}\right)$ many times. The runtime bound as well as the claim on the kind of measurements performed follow immediately from the efficiency of shadow overlap estimation and the fact that, according to \Cref{theorem:extending-iid-to-noniid}, the non-i.i.d.~extension $\mathcal{L}'^{\mathsf{O}^{\mathrm{Mem}}(f)}$ uses the same kind of measurements and post-processing as the underlying i.i.d.~algorithm $\mathcal{L}$.
\end{proof}

By \Cref{theorem:verified-shadow-overlap-algorithm}, when given access to a public quantum phase oracle for $f$ and a private classical membership query oracle for $f$, one can use \Cref{alg:certify-state} to obtain a single certified (approximate) copy of the phase state $\ket{\psi_f^{\mathrm{Ph}}}$ with completeness and soundness requirements as in \Cref{definition:covert-quantum-data-public-quantum-oracle-private-classical-queries}.
As discussed at the end of the previous subsection, this immediately extends from a single to $m$ certified phase state copies when we set $N=\tilde{\mathcal{O}}(m^5n^5/\delta^2\varepsilon^6)$.\footnote{We point out that this $N$ refers to the number of $(mn)$-qubit subsystems. Thus, one incurs a further $m$-factor when rephrasing in terms of the number of $n$-qubit subsystems.}
Note that this upper bound falls short of the usually desired (poly-)logarithmic scaling in $1/\delta$. 
While we leave open the question of how to achieve a $\mathrm{poly}(\log(1/\delta))$-scaling for (covert) verifiable quantum data acquisition as in \Cref{definition:covert-quantum-data-public-quantum-oracle-private-classical-queries}, at the end of \cref{sec:non-iid-onlybackwards}, we discuss how to ensure a $\mathrm{poly}(\log(1/\delta))$-scaling when using this verifiable quantum data acquisition subroutine for solving any specific decision or learning task of interest.

\begin{remark}\label{remark:combine-with-vqc}     
    \cite{goldwasser2021interactive, caro2024interactiveproofsverifyingquantum} both highlight that verification of learning, classical as well as quantum, in general goes beyond verification of computation because of the presence of untrusted data on the prover side.
    In the quantum context, \Cref{theorem:verified-shadow-overlap-algorithm} provides a complete and sound procedure for certifying a specific type of quantum data (even with privacy guarantees, as we argue below). 
    Thus, a verifier may first use our procedure to acquire certified quantum data, and they can then (compare \cite[Observation 1]{caro2024interactiveproofsverifyingquantum}) delegate quantum computations on that data using verified quantum computing \cite{gheorghiu2018verification, fitzsimonsUnconditionallyVerifiableBlind2017, broadbent2018howtoverify}.
    This may serve as a more general recipe for combining data certification with interactive proofs for computing to obtain interactive proofs for learning.
\end{remark}

\subsection{Covert Verifiable Quantum Learning Against Unidirectional Adversaries} \label{sec:non-iid-onlybackwards}

We now give two simple procedures for querying a public quantum phase oracle such that, if an eavesdropper only has access to the state sent back from the public quantum phase oracle to the learner (in Step 3 of \Cref{alg:covert-public-phase-query-randomness,alg:covert-public-phase-query-entanglement} below), then the eavesdropper does not learn anything about the unknown function $f$. 

\begin{algorithm}
    \caption{Target-covert public quantum phase oracle query from classical randomness}
    \label{alg:covert-public-phase-query-randomness}
    \begin{algorithmic}[1]
        \Require Access to a public oracle $\mathsf{O}_{\mathrm{pub}}^{\mathrm{QPh}}(f)$
        \Ensure One copy of $\ket{\psi_f^{\mathrm{Ph}}}$
        \State Privately choose $r\sim\{0,1\}^n$ uniformly at random.
        \State Prepare the state $\ket{\psi^{(r)}}=Z^{r}(H\ket{0})^{\otimes n}=\frac{1}{\sqrt{2^n}}\sum_{x\in\{0,1\}^n} (-1)^{r\cdot x} \ket{x}$ and send it to $\mathsf{O}_{\mathrm{pub}}^{\mathrm{QPh}}(f)$. 
        \State Receive the state $\ket{\psi_{f}^{(r)}}=\frac{1}{\sqrt{2^n}}\sum_{x\in\{0,1\}^n} (-1)^{r\cdot x + f(x)} \ket{x}$ in response to the public oracle query.
        \State Apply $Z^{r}$ to $\ket{\psi_{f}^{(r)}}$ to obtain $\ket{\psi_f^{\mathrm{Ph}}}$.
    \end{algorithmic}
\end{algorithm}

\begin{algorithm} 
    \caption{Target-covert public quantum phase oracle query from entanglement}
    \label{alg:covert-public-phase-query-entanglement}
    \begin{algorithmic}[1]
        \Require Access to a public oracle $\mathsf{O}_{\mathrm{pub}}^{\mathrm{QPh}}(f)$
        \Ensure One copy of $\ket{\psi_f^{\mathrm{Ph}}}$
        \State Prepare the state $(H\ket{0})^{\otimes n}\otimes (H\ket{0})^{\otimes n}=\left(\frac{1}{\sqrt{2^n}}\sum_{r\in\{0,1\}^n}\ket{r}\right)\otimes \left(\frac{1}{\sqrt{2^n}}\sum_{x\in\{0,1\}^n}\ket{x}\right)$.
        \State For $1\leq i\leq n$, apply controlled-$Z$ gates between the $i$th register (as control) and the $(n+i)$th register (as target), thus preparing the state $\frac{1}{\sqrt{2^n}}\sum_{r\in\{0,1\}^n}\ket{r}\otimes \ket{\psi^{(r)}}$. Send the last $n$ qubits to $\mathsf{O}_{\mathrm{pub}}^{\mathrm{QPh}}(f)$. 
        \State Receive the response to the public oracle query, leading to the overall state $\frac{1}{\sqrt{2^n}}\sum_{r\in\{0,1\}^n}\ket{r}\otimes \ket{\psi_{f}^{(r)}}$.
        \State Measure the first $n$ qubits, observe outcome $r$, then apply $Z^{r}$ to the remaining $n$ qubits to obtain $\ket{\psi_f^{\mathrm{Ph}}}$. (Alternatively: Undo the controlled-$Z$ gates to obtain $(H\ket{0})^{\otimes n} \otimes \ket{\psi_f^{\mathrm{Ph}}}$ and discard the first $n$ qubits.)
    \end{algorithmic}
\end{algorithm}

It is easy to see that \Cref{alg:covert-public-phase-query-randomness,alg:covert-public-phase-query-entanglement} do not leak information to a unidirectional adversary. We include formal observation and a short proof for ease of reference and for completeness.

\begin{observation} \label{obs:target-covert-query}
    \Cref{alg:covert-public-phase-query-randomness,alg:covert-public-phase-query-entanglement} satisfy the following guarantees:
    \begin{itemize}
        \item[(i)] \textbf{Completeness:} If there is no eavesdropper, then, when given public oracle access to $\mathsf{O}_{\mathrm{pub}}^{\mathrm{QPh}}(f)$, \Cref{alg:covert-public-phase-query-randomness,alg:covert-public-phase-query-entanglement} each output one copy of $\ket{\psi_f^{\mathrm{Ph}}}$.
        \item[(ii)] \textbf{Privacy:} For $F\sim\mathcal{F}$ drawn from some class $\mathcal{F}$ with $|\mathcal{F}|>1$ according to some probability distribution $\mu$ over $\mathcal{F}$, a unidirectional adversary that only eavesdrops on the quantum communication from $\mathsf{O}_{\mathrm{pub}}^{\mathrm{QPh}}(f)$ to the learner in \Cref{alg:covert-public-phase-query-randomness} or \Cref{alg:covert-public-phase-query-entanglement} gains no information about $F$. That is, the joint state $\rho_{\mathsf{FA}}$ of the classical function register and the adversary's register factorizes as $\rho_{\mathsf{FA}} = \mathbb{E}_{F\sim\mu}[\ket{F}\bra{F}_{\mathsf{F}}]\otimes \rho_{\mathsf{A}}$. 
    \end{itemize}
\end{observation} 
\begin{proof}
    Completeness is immediate from the description of the algorithms. Here, we only rely on the fact that $[Z^r, \mathsf{O}_{\mathrm{pub}}^{\mathrm{QPh}}(f)]=0$ holds for all $r$ and $f$. 

    For privacy, we present the argument for \Cref{alg:covert-public-phase-query-randomness}. The reasoning for \Cref{alg:covert-public-phase-query-entanglement} is identical because the reduced state on the subsystem that the eavesdropper has access to is the same in both procedures.
    We can describe the overall system after the public quantum phase oracle has been queried---consisting of the classical random function, the classical random bits, and the quantum state $\ket{\psi_{f}^{(r)}}$ prepared in Step 3 of \Cref{alg:covert-public-phase-query-randomness}---by the classical-quantum state 
    \begin{equation}
        \rho_{\mathsf{FRQ}}
        = \sum_{f:\{0,1\}^n\to\{0,1\}} \mu(f)\ket{f}\bra{f}_\mathsf{F} \otimes \frac{1}{2^n}\sum_{r\in\{0,1\}^n} \ket{r}\bra{r}_\mathsf{R} \otimes\ket{\psi_{f}^{(r)}}\bra{\psi_{f}^{(r)}}_\mathsf{Q}\, .
    \end{equation}
    A unidirectional adversary that only eavesdrops on the quantum communication from $\mathsf{O}_{\mathrm{pub}}^{\mathrm{QPh}}(f)$ to the learner does not have access to the private randomness system $\mathsf{R}$. Thus, the relevant classical-quantum state from the adversary's perspective is
    \begin{equation}
        \rho_{\mathsf{FQ}}
        = \tr_{\mathsf{R}}[\rho_{\mathsf{FRQ}}]
        = \sum_{f:\{0,1\}^n\to\{0,1\}} \mu(f)\ket{f}\bra{f}_\mathsf{F} \otimes \frac{1}{2^n}\sum_{r\in\{0,1\}^n} \ket{\psi_{f}^{(r)}}\bra{\psi_{f}^{(r)}}_\mathsf{Q}
        \, .
    \end{equation}
    A direct calculation shows that $\frac{1}{2^n}\sum_{r\in\{0,1\}^n} \ket{\psi_{f}^{(r)}}\bra{\psi_{f}^{(r)}}_\mathsf{Q} = \frac{1}{2^n}\sum_{x\in\{0,1\}^n} \ket{x}\bra{x}_\mathsf{Q} = \mathds{1}_\mathsf{Q}/2^n$ holds for every $f$. 
    Hence, we have shown that the state factorizes as $\rho_{\mathsf{FQ}}=\mathbb{E}_{F\sim\mu}[\ket{F}\bra{F}_{\mathsf{F}}]\otimes \mathds{1}_\mathsf{Q}/2^n$.
    As the adversary can now only apply local processing in the form of a quantum channel mapping the $\mathsf{Q}$-system to some (potentially larger) $\mathsf{A}$-system, the final state $\rho_{\mathsf{FA}}$ also factorizes as $\rho_{\mathsf{FA}} = \mathbb{E}_{F\sim\mu}[\ket{F}\bra{F}_{\mathsf{F}}]\otimes \rho_{\mathsf{A}}$.
\end{proof}

The argument given above immediately extends to non-uniform superpositions. Concretely, when replacing $\frac{1}{\sqrt{2^n}}\sum_x \ket{x}$ by $\sum_x \alpha_x \ket{x}$, the intermediate factorization becomes $\rho_{\mathsf{FQ}} = \mathbb{E}_{F\sim\mu}[\ket{F}\bra{F}_{\mathsf{F}}]\otimes \sum_x |\alpha_x|^2 \ket{x}\bra{x}_\mathsf{Q}$. However, as quantum testing and learning procedures often focus on the state $\ket{\psi_f^{\mathrm{Ph}}}$, so do we. We also note that, at this point, \Cref{alg:covert-public-phase-query-entanglement} provides the same guarantees as \Cref{alg:covert-public-phase-query-randomness}, despite requiring more involved (quantum) processing from the learner. The potential benefits of \Cref{alg:covert-public-phase-query-entanglement} over \Cref{alg:covert-public-phase-query-randomness} will become clearer in \Cref{sec:iid-both-directions}.

Next, we show how to combine \Cref{alg:covert-public-phase-query-randomness} or \Cref{alg:covert-public-phase-query-entanglement} with the non-i.i.d.~phase state certification developed in \Cref{theorem:verified-shadow-overlap-algorithm} to obtain a covert verifiable procedure for acquiring phase state copies from a public quantum phase query oracle. We first give the algorithm and then prove the corresponding guarantees.

\begin{algorithm}
    \caption{Covert Verifiable Phase States from Public Oracle against Unidirectional Adversaries}
    \begin{algorithmic}[1]
        \Require Number of qubits $n \in \mathbb{N}$; number of desired copies $m \in \mathbb{N}$; parameters $\delta, \varepsilon \in (0,1)$; access to oracles $\mathsf{O}_{\mathrm{pub}}^{\mathrm{QPh}}(f)$ and $\mathsf{O}_{\mathrm{pri}}^{\mathrm{Mem}}(f)$
        \Ensure ``reject'' or a an $(mn)$-qubit state
        \State Set $N = N(n, m, \varepsilon, \delta) = \tilde{\mathcal{O}}\left( \frac{nm \left((nm)^2 \log(1/\delta) / \varepsilon^2\right)^2 \log^2(1/\delta)}{\delta^2 \varepsilon^2} \right) = \tilde{\mathcal{O}}\left( \frac{(nm)^5}{\delta^2 \varepsilon^6} \right)$.
        \State Query $\mathsf{O}_{\mathrm{pub}}^{\mathrm{QPh}}(f)$ using either \Cref{alg:covert-public-phase-query-randomness} or \Cref{alg:covert-public-phase-query-entanglement} for $Nm$ times, with fresh private randomness or a fresh randomness register, respectively, in each query. Denote the resulting $(Nmn)$-qubit state by $\rho$.
        \State Using $\mathsf{O}_{\mathrm{pri}}^{\mathrm{Mem}}(f)$, apply $\ComplexityFont{CertifyState}^{\mathsf{O}^{\mathrm{Mem}}(f)}_{ \varepsilon, \delta}$ (\Cref{alg:certify-state}) to certify $\rho$ against the phase state $\ket{\psi_f^{\mathrm{Ph}}}^{\otimes m}$. 
        \If{$\ComplexityFont{CertifyState}^{\mathsf{O}^{\mathrm{Mem}}(f)}_{ \varepsilon, \delta}(\rho)$ rejects}
            \State Output ``reject''.
        \Else
            \State Return the resulting $(mn)$-qubit state $\sigma$.
        \EndIf
    \end{algorithmic}
    \label{alg:covert-verifiable-public-quantum-private-classical}
\end{algorithm}

\begin{theorem}[Covert Verifiable Phase States against Unidirectional Adversaries---Formal version of \Cref{inf-theorem:covert-quantum-data-acquisition-v1}] \label{theorem:covert-shadow-overlap} 
For $\delta, \varepsilon \in (0, 1)$, \Cref{alg:covert-verifiable-public-quantum-private-classical} is a $(m_{\mathrm{pri}} = \mathcal{O}\left(\frac{n^2m^3 \log(1/\delta)}{\varepsilon^2}\right), m_{\mathrm{pub}} = \tilde{\mathcal{O}}\left( \frac{n^5m^6}{\delta^2 \varepsilon^6} \right),\delta_c = 0,\delta_s = \delta, \varepsilon)$-covert verifiable quantum algorithm for producing $m$ quantum phase states for an arbitrary concept class $\mathcal{F}$ against unidirectional adversaries. That is: 
    \begin{itemize}
        \item \textbf{Completeness:} With no adversary, for any $f\in\mathcal{F}$, after making at most $m_{\mathrm{pri}}$ queries to $\mathsf{O}_{\mathrm{pri}}^{\mathrm{Mem}}(f)$ and at most $m_{\mathrm{pub}}$ queries to $\mathsf{O}_{\mathrm{pub}}^{\mathrm{QPh}}(f)$, \Cref{alg:covert-verifiable-public-quantum-private-classical} accepts and outputs $\ket{\psi_f^{\mathrm{Ph}}}\bra{\psi_f^{\mathrm{Ph}}}^{\otimes m}$  with success probability $1$.
        \item \textbf{Soundness:} For any $f\in\mathcal{F}$, after making at most $m_{\mathrm{pri}}$ queries to $\mathsf{O}_{\mathrm{pri}}^{\mathrm{Mem}}(f)$ and at most $m_{\mathrm{pub}}$ queries to $\mathsf{O}_{\mathrm{pub}}^{\mathrm{QPh}}(f)$, the latter of which are subject to corruption by an arbitrary adversary, \Cref{alg:covert-verifiable-public-quantum-private-classical} accepts and outputs some $\sigma$ with $\bra{\psi_f^{\mathrm{Ph}}}^{\otimes m}\sigma\ket{\psi_f^{\mathrm{Ph}}}^{\otimes m}< 1-\varepsilon$ with failure probability $\leq \delta$.
        \item \textbf{{Privacy}:} For $F\sim\mathcal{F}$ a randomly drawn concept, a unidirectional adversary adversary gains no information about $F$ from interacting with the at most $m_{\mathrm{pub}}$ queries to $\mathsf{O}_{\mathrm{pub}}^{\mathrm{QPh}}(f)$ made by \Cref{alg:covert-verifiable-public-quantum-private-classical}. That is, the joint state $\rho_{\mathsf{FA}}$ of the classical function register and the adversary's register factorizes as $\rho_{\mathsf{FA}} = \frac{\mathds{1}_{\mathsf{F}}}{|\mathcal{F}|}\otimes \rho_\mathsf{A}$. 
\end{itemize}
\end{theorem}

As will be evident from the proof, while the privacy is against unidirectional adversaries, the soundness holds even against arbitrary adversaries.

\begin{proof}
    \textbf{Completeness:} Due to \Cref{obs:target-covert-query}(i), we know that $\rho$ in Line 2 of \Cref{alg:covert-verifiable-public-quantum-private-classical} is $\ket{\psi_f^{\mathrm{Ph}}}\bra{\psi_f^{\mathrm{Ph}}}^{\otimes mN}$. Moreover, completeness of $\ComplexityFont{CertifyState}^{\mathsf{O}^{\mathrm{Mem}}(f)}_{ \varepsilon, \delta}$ from \Cref{theorem:verified-shadow-overlap-algorithm} guarantees that with probability $1$, we accept and output $\ket{\psi_f^{\mathrm{Ph}}}\bra{\psi_f^{\mathrm{Ph}}}^{\otimes m}$. 
    
    \textbf{Soundness:} Number of copies used ($N$) suffices to use the soundness of $\ComplexityFont{CertifyState}^{\mathsf{O}^{\mathrm{Mem}}(f)}_{ \varepsilon, \delta}$ from \Cref{theorem:verified-shadow-overlap-algorithm} applied to the $m$-copy phase state of $f$. As $\ComplexityFont{CertifyState}^{\mathsf{O}^{\mathrm{Mem}}(f)}_{ \varepsilon, \delta}$ works with non-i.i.d~copies, we can apply to an arbitrary $(Nmn)$-qubit state ($\rho$ in this case). Therefore, no matter the action of adversary, 
    \begin{equation}
        \Pr\left[\text{\Cref{alg:covert-verifiable-public-quantum-private-classical} accepts and } \bra{\psi_f^{\mathrm{Ph}}}^{\otimes m}\sigma\ket{\psi_f^{\mathrm{Ph}}}^{\otimes m} < 1 - \varepsilon  \right]
            \leq \delta .
    \end{equation}
    \textbf{Privacy:} As we use fresh randomness or a fresh entanglement register, respectively, in each of the $Nm$ queries, \Cref{obs:target-covert-query}(ii) tells us that, in each round, the unidirectional adversary gains no information about the unknown $F$ and the joint function-adversary state factorizes.
    Consequently, even for the whole interaction, the adversary gains no information about $F$ and the joint function-adversary state factorizes. 
    
    \textbf{Efficiency:} To use the non-i.i.d.~certification protocol from \Cref{theorem:verified-shadow-overlap-algorithm}, we need sufficiently many public oracle queries to prepare $N = N(n, m, \varepsilon, \delta) = 
    \tilde{\mathcal{O}}\left( \frac{(nm)^5}{\delta^2 \varepsilon^6} \right)$ many copies of $\ket{\psi_f^{\mathrm{Ph}}}^{\otimes m}$.
    As one copy of $\ket{\psi_f^{\mathrm{Ph}}}^{\otimes m}$ can be prepared from $m$ many queries to $ \mathsf{O}_{\mathrm{pub}}^{\mathrm{QPh}}(f)$ queries, we arrive at an overall number of $Nm = \tilde{\mathcal{O}}\left( \frac{n^5m^6}{\delta^2 \varepsilon^6} \right)$ queries to $ \mathsf{O}_{\mathrm{pub}}^{\mathrm{QPh}}(f)$. According to \Cref{theorem:verified-shadow-overlap-algorithm}, the classical post-processing of the obtained measurement outcomes uses $\mathcal{O}\left(\frac{(nm)^2 \log(1/\delta)}{\varepsilon^2}\right)$ classical membership queries to the function $f^{\otimes m}(x_1,\ldots,x_m)=f(x_1)\oplus \ldots \oplus f(x_m)$ associated with the phase state $\ket{\psi_f^{\mathrm{Ph}}}^{\otimes m}$. One query to $f^{\otimes m}$ can be simulated using $m$ queries to $\mathsf{O}_{\mathrm{pri}}^{\mathrm{Mem}}(f)$, leading us to a total number of $\mathcal{O}\left(\frac{n^2m^3 \log(1/\delta)}{\varepsilon^2}\right)$ queries to $\mathsf{O}_{\mathrm{pri}}^{\mathrm{Mem}}(f)$. 
\end{proof}

\Cref{alg:covert-verifiable-public-quantum-private-classical} effectively augments the non-i.i.d.~state certification of phase states (\Cref{alg:certify-state}) with the property of covertness. However, as discussed as the end of \Cref{sec:verification-shadows}, we inherit the $\poly(1/\delta)$ scaling from \Cref{alg:certify-state}. Lacking a general quantum counterpart of majority voting\footnote{While \cite{buhrman2022quantummajorityvote} provides a certain kind of quantum majority vote, it is unfortunately insufficient for our purposes because it assumes that there are only two possible candidate states and that those two states are orthogonal.} between copies produced from single rounds in the above protocol, it is not clear how to apply the standard approach to amplify the success probability and thus improve the $\delta$-dependence. However, when considering algorithms that use such copies to solve a (distinguishing) task, we now show that one can indeed amplify the success probability. This succeeds for essentially the same reason that weak error reduction for \textsc{QMA} (see, for instance, \cite[Lemma 14.1]{kitaev2002classicalandquantum}) succeeds. 
We leave as an open question whether one can also imitate \emph{strong} error reduction for \textsc{QMA} to obtain improved upper bounds on the number of copies. 

For ease of presentation, here we consider binary distinguishing problems (with a promise) for Boolean functions. 
That is, we consider a problem $P$ defined by subsets $S_{\mathrm{YES}}, ~S_{\mathrm{NO}}\subset\{0,1\}^{\{0,1\}^n}$. If $f\in S_{\mathrm{YES}}$ ($f\in S_{\mathrm{NO}}$), the only valid solution to $P$ for $f$ is YES (NO). If $f\not\in S_{\mathrm{YES}}\cup S_{\mathrm{NO}}$, both YES and NO are valid solutions to $P$ for $f$. 
With this nomenclature established, we now present a version of \Cref{theorem:covert-shadow-overlap} for solving a distinguishing problem with amplified success probability.

\begin{algorithm}
    \caption{\Cref{alg:covert-verifiable-public-quantum-private-classical} applied to distinguishing task}
    \begin{algorithmic}[1]
        \Require Number of qubits $n\in\mathbb{N}$; distinguishing algorithm $A$ with parameters $\varepsilon_A\in (0,1)$ and $\delta_A\in (0,1/4)$; access to oracles $\mathsf{O}_{\mathrm{pub}}^{\mathrm{QPh}}(f)$, $\mathsf{O}_{\mathrm{pri}}^{\mathrm{Mem}}(f)$; target confidence parameter $\delta\in (0,1)$.
        \Ensure ``reject'' or candidate solution
        \State Set $\ell=\frac{2\ln(1/\delta)}{(1-4\delta_A)^2}$ and $N=N(n,m,\varepsilon_A,\delta_A,\delta)=\Tilde{\mathcal{O}}\left(\frac{n^5 m^5}{\delta_A^2 \varepsilon_A^6}\right)$.
        \State Query $\mathsf{O}_{\mathrm{pub}}^{\mathrm{QPh}}(f)$ using either \Cref{alg:covert-public-phase-query-randomness} or \Cref{alg:covert-public-phase-query-entanglement}, $N\ell m$ times, with fresh private randomness or a fresh randomness register, respectively, in each round. Denote the resulting $(N\ell mn)$-qubit state by $\rho$.
        \For{$1\leq j\leq \ell$}
            \State Using $\mathsf{O}_{\mathrm{pri}}^{\mathrm{Mem}}(f)$, apply $\ComplexityFont{CertifyState}^{\mathsf{O}^{\mathrm{Mem}}(f)}_{ \varepsilon_A, \delta_A}$ to the first $Nnm$ qubits of $\rho$ against the phase state $\ket{\psi_f^{\mathrm{Ph}}}^{\otimes m}$.
            \If{$\ComplexityFont{CertifyState}^{\mathsf{O}^{\mathrm{Mem}}(f)}_{ \varepsilon_A, \delta_A}(\rho)$ rejects} 
                \State Output ``reject''.
            \Else 
                \State Denote the resulting $((N(\ell-j)+1)mn)$-qubit state by $\tilde{\rho}^{(j)}$.
                \State $X_j \gets A(\tr_{>mn}[\tilde{\rho}^{(j)}])$.
            \EndIf 
            \State Set $\rho\gets \tr_{<mn}[\tilde{\rho}^{(j)}]$.
        \EndFor
        \State Return $\text{Majority}(X_1,\ldots, X_\ell)$.
    \end{algorithmic}\label{alg:amp-verifiable-public-quantum-private-classical}
\end{algorithm}

\begin{corollary}\label{corollary:verifiable-binary-distinguishing-from-phase-states}
    Consider a binary distinguishing problem $P$ for Boolean functions. Let $\ket{\psi_f^{\mathrm{Ph}}}$ be the phase state of a Boolean function $f: \{0, 1\}^n \to \{0, 1\}$. Let $A$ be a quantum algorithm for solving $P$ that, upon input of any $(mn)$-qubit state $\rho$, satisfies the following guarantee for some $\varepsilon_A \in (0,1), \delta_A\in (0,\frac{1}{4}),~m\in \mathbb{N}$: 
    \begin{align}\label{equation:assumption-algorithm-solves-problem-robustly}
        \text{ if } \bra{\psi_f^{\mathrm{Ph}}}^{\otimes m} \rho \ket{\psi_f^{\mathrm{Ph}}}^{\otimes m } \geq 1 - \varepsilon_A, \text{ then } \Pr[ A(\rho) \text{ is a valid solution to }P\text{ for }f] \geq 1 - \delta_A\, .
    \end{align}
    Let $\delta\in (0,1)$ be a confidence parameter.
    Set $\ell=\frac{2\ln(1/\delta)}{(1-4\delta_A)^2}$.
    \Cref{alg:amp-verifiable-public-quantum-private-classical} gives a quantum algorithm $A'$ that, using at most $m_{\mathrm{pub}}=\Tilde{\mathcal{O}}\left(\frac{n^5m^6 \ell}{\delta_A^2 \varepsilon_A^6}\right)$ queries to a public quantum phase oracle $\mathsf{O}_{\mathrm{pub}}^{\mathrm{QPh}}(f)$ and at most $m_{\mathrm{pri}}=\Tilde{\mathcal{O}}\left(\frac{n^2 m^3 \ell \log\left(1/\delta_A\right)}{\varepsilon_A^2}\right)$ queries to a private classical membership query oracle $\mathsf{O}_{\mathrm{pri}}^{\mathrm{Mem}}(f)$, achieves the following guarantees:
    \begin{itemize}
        \item \textbf{Completeness:} With no adversary, for any $f\in\mathcal{F}$, after making at most $m_{\mathrm{pri}}$ queries to $\mathsf{O}_{\mathrm{pri}}^{\mathrm{Mem}}(f)$ and at most $m_{\mathrm{pub}}$ queries to $\mathsf{O}_{\mathrm{pub}}^{\mathrm{QPh}}(f)$, 
        \begin{equation}
            \Pr\left[A'^{\mathsf{O}_{\mathrm{pri}}^{\mathrm{Mem}}(f), \mathsf{O}_{\mathrm{pub}}^{\mathrm{QPh}}(f)}\text{ accepts and outputs a valid solution to } P \text{ for }f\right]
            \geq 1-\delta\, .
        \end{equation}
        \item \textbf{Soundness:}  For any $f\in\mathcal{F}$, after making at most $m_{\mathrm{pri}}$ queries to $\mathsf{O}_{\mathrm{pri}}^{\mathrm{Mem}}(f)$ and at most $m_{\mathrm{pub}}$ queries to $\mathsf{O}_{\mathrm{pub}}^{\mathrm{QPh}}(f)$, the latter of which are subject to corruption by an arbitrary adversary,  
        \begin{equation}
            \Pr\left[A'^{\mathsf{O}_{\mathrm{pri}}^{\mathrm{Mem}}(f), \mathsf{O}_{\mathrm{pub}}^{\mathrm{QPh}}(f)}\text{ accepts and outputs an invalid solution to } P \text{ for }f\right]
            \leq \delta\, .
        \end{equation}
        \item \textbf{{Privacy}:} For $F\sim\mathcal{F}$ drawn uniformly at random from any concept class $\mathcal{F}$, a unidirectional adversary adversary gains no information about $F$ from interacting with the at most $m_{\mathrm{pub}}$ queries to $\mathsf{O}_{\mathrm{pub}}^{\mathrm{QPh}}(f)$ made by \Cref{alg:amp-verifiable-public-quantum-private-classical}. That is, the joint state $\rho_{\mathsf{FA}}$ of the classical function register and the adversary's register factorizes as $\rho_{\mathsf{FA}} = \frac{\mathds{1}_{\mathsf{F}}}{|\mathcal{F}|}\otimes \rho_\mathsf{A}$. 
        \item \textbf{Efficiency:} $A'$ runs in time $\mathcal{O}(\poly(n, m, 1/\delta_A, 1/(1-4\delta_A)^2, 1/\varepsilon_A, \log(1/\delta), \ComplexityFont{Complexity}(A))$ where $\ComplexityFont{Complexity}(A)$ is the run time of the original algorithm $A$. In addition to the operations used in $A$, $A'$ only employs uniform superposition state preparation and single-qubit Pauli gates if \Cref{alg:covert-public-phase-query-randomness} is used (single-qubit Pauli and CZ gates if \Cref{alg:covert-public-phase-query-entanglement} is used). 
    \end{itemize}
\end{corollary}

As in \Cref{theorem:covert-shadow-overlap}, the soundness in \Cref{corollary:verifiable-binary-distinguishing-from-phase-states} is against general adversaries. Only the privacy in \Cref{corollary:verifiable-binary-distinguishing-from-phase-states} requires restricting the adversary to be unidirectional.

\begin{proof}
    \textbf{Completeness: } Due to \Cref{obs:target-covert-query}(i), the state $\rho$ in Line 2 of \Cref{alg:amp-verifiable-public-quantum-private-classical} equals \(\ket{\psi_f^{\mathrm{Ph}}}\bra{\psi_f^{\mathrm{Ph}}}^{\otimes N\ell}\).  The completeness of $\ComplexityFont{CertifyState}^{\mathsf{O}^{\mathrm{Mem}}(f)}_{ \varepsilon_A, \delta_A}$ from \Cref{theorem:verified-shadow-overlap-algorithm} therefore implies that, with probability 1, the protocol accepts and outputs \(\tr_{>mn}[\tilde{\rho}^{(j)}] = \ket{\psi_f^{\mathrm{Ph}}}\bra{\psi_f^{\mathrm{Ph}}}^{\otimes m}\) for every round \(1 \le j \le \ell\). For each such round we have $\bra{\psi_f^{\mathrm{Ph}}}^{\otimes m} \tr_{>mn}[\tilde{\rho}^{(j)}] \ket{\psi_f^{\mathrm{Ph}}}^{\otimes m} = 1 \ge 1-\varepsilon_A$, so by the robustness guarantee of $A$, we have 
    \begin{equation}
        \Pr[ A(\tr_{>mn}[\tilde{\rho}^{(j)}]) \text{ is a valid solution to }P\text{ for }f] \ge 1 - \delta_A \, .
    \end{equation}
    That is, Step 9 outputs the correct solution with probability $\ge 1-\delta_A$ in every round, and the different rounds are i.i.d.~because $\rho=\ket{\psi_f^{\mathrm{Ph}}}\bra{\psi_f^{\mathrm{Ph}}}^{\otimes N\ell}$. So, by a Chernoff bound we conclude that 
    \begin{align}
        \Pr [A' \text{ outputs a valid solution to } P(f)]
        &= \Pr[\text{Majority}(X_1,\ldots, X_\ell) \text{ is a valid solution to } P(f)]\\
        &= \Pr[\mathrm{Binom}(l, 1-\delta_A)\ge \ell/2]\\
        &\ge 1 - \exp\left(-2 \ell \left(\frac{1}{2} - \delta_A\right)^2\right)\\
        &\ge 1 - \delta \, ,
        \end{align}
    where in the last step we used $\ell=\frac{2\ln(1/\delta)}{(1-4\delta_A)^2}$.
    
    \textbf{Soundness: } Consider now an arbitrary adversary that can interfere with the public oracle queries in Step 2. Here, we use the soundness guarantee of \Cref{theorem:verified-shadow-overlap-algorithm} and the robustness guarantee of algorithm $A$ to argue that we either ``reject'' or output the valid solution with high probability. We only err if we accept in Line 4 in all rounds and if in a majority of those rounds the invalid solution is output in Line 9. Hence, we want to upper bound the following probability probability: 
    \begin{align}
        &\Pr [A'\text{ accepts the interaction and outputs an invalid solution}]\\
        &= \Pr \left[\text{no round rejects and }\geq \frac{1}{2} \text{ of all } \ell \text{ rounds produce an invalid solution}\right]\\
        &= \sum_{b\in\{0,1\}^\ell: |b|\geq \ell/2} \Pr \left[\text{no round rejects and exactly the rounds $j$ with $b_j=1$ produce an invalid solution}\right]\\
        &\leq \sum_{b\in\{0,1\}^\ell: |b|\geq \ell/2} \Pr \left[\text{all rounds $j$  with $b_j=1$ accept and produce an invalid solution}\right]\\
        &= \sum_{b\in\{0,1\}^\ell: |b|\geq \ell/2} \Pr \left[\text{for all $j$  with $b_j=1$: Step 4 accepts and Step 9 outputs an invalid solution}\right]\\
        &\leq \sum_{b\in\{0,1\}^\ell: |b|\geq \ell/2} (2 \delta_A)^{|b|}\\
        &= \Pr[\mathrm{Binom}(m, 2\delta_A)\geq \ell/2]\\
        &\leq \exp\left(-2 \ell \left(\frac{1}{2} - 2\delta_A\right)^2\right) \\
        &\leq \delta\, ,
    \end{align}
    where in the last step we used $\ell=\frac{2\ln(1/\delta)}{(1-4\delta_A)^2}$. Moreover, in the fifth step, we have combined the soundness guarantee of \Cref{theorem:verified-shadow-overlap-algorithm} with the robustness guarantee of algorithm $A$ as follows.  Fix a $b\in\{0,1\}^\ell$, and denote by $j_1,\ldots,j_{|b|}$ the indices $j$ such that $b_j=1$. Let $\ComplexityFont{ACCEPT}_j$ be the event that Step 4 accepts in round $j$ and let $\ComplexityFont{INVALID}_j$ be the event that Step 9 outputs an invalid solution in round $j$. 
    Then
    \begin{align}
        &\Pr\left[\text{for all $j$  with $b_j=1$: Step 4 accepts and Step 9 outputs an invalid solution}\right]\\
        &= \Pr\left[\bigcap_{k=1}^{|b|} \ComplexityFont{ACCEPT}_{j_k} \cap \ComplexityFont{INVALID}_{j_k}\right]\\
        &= \Pr[\ComplexityFont{ACCEPT}_{j_1} \cap \ComplexityFont{INVALID}_{j_1}]\cdot\Pr\left[\bigcap_{k=2}^{|b|} \ComplexityFont{ACCEPT}_{j_k} \cap \ComplexityFont{INVALID}_{j_k} ~\bigg\vert~ \ComplexityFont{ACCEPT}_{j_1} \cap \ComplexityFont{INVALID}_{j_1}\right]\\
        &= \Pr[\ComplexityFont{ACCEPT}_{j_1} \cap \ComplexityFont{INVALID}_{j_1}]\cdot\ldots\cdot \Pr\left[\ComplexityFont{ACCEPT}_{j_{|b|}} \cap \ComplexityFont{INVALID}_{j_{|b|}} ~\bigg\vert~ \bigcap_{k=1}^{|b|-1} \ComplexityFont{ACCEPT}_{j_k} \cap \ComplexityFont{INVALID}_{j_k}\right]\\
        &\leq (2\delta_A)^{|b|}\, .
    \end{align}        
    \normalsize
    Here, the last step holds because for any $1\leq i \leq |b|$, we can further divide the corresponding factor based on whether the resulting state $\tr_{>mn}[\tilde{\rho}^{(j_i)}]$ has high overlap with $\ket{\psi_f^{\mathrm{Ph}}}^{\otimes m}$ or not. Define $\ComplexityFont{GOOD}_j$ to be the event $\bra{\psi_f^{\mathrm{Ph}}}^{\otimes m} \tr_{>mn}[\tilde{\rho}^{(j)}] \ket{\psi_f^{\mathrm{Ph}}}^{\otimes m} \ge 1-\varepsilon_A$ for each round $j$. Then, for any $i$, 
    \small
    \begin{align}
        &\Pr\left[\ComplexityFont{ACCEPT}_{j_{i}} \cap \ComplexityFont{INVALID}_{j_{i}} ~\bigg\vert~ \bigcap_{k=1}^{i-1} \ComplexityFont{ACCEPT}_{j_k} \cap \ComplexityFont{INVALID}_{j_k}\right]\\
        &= \Pr\left[\ComplexityFont{ACCEPT}_{j_{i}} \cap \ComplexityFont{INVALID}_{j_{i}}  \cap \ComplexityFont{GOOD}_{j_i}~\bigg\vert~ \bigcap_{k=1}^{i-1} \ComplexityFont{ACCEPT}_{j_k} \cap \ComplexityFont{INVALID}_{j_k}\right] \\
        &+ \Pr\left[\ComplexityFont{ACCEPT}_{j_{i}} \cap \ComplexityFont{INVALID}_{j_{i}}  \cap \overline{\ComplexityFont{GOOD}_{j_i}}~\bigg\vert~ \bigcap_{k=1}^{i-1} \ComplexityFont{ACCEPT}_{j_k} \cap \ComplexityFont{INVALID}_{j_k}\right]\\
        &\leq \Pr\left[\ComplexityFont{INVALID}_{j_{i}}  \cap \ComplexityFont{GOOD}_{j_i}~\bigg\vert~ \bigcap_{k=1}^{i-1} \ComplexityFont{ACCEPT}_{j_k} \cap \ComplexityFont{INVALID}_{j_k}\right] + \Pr\left[\ComplexityFont{ACCEPT}_{j_{i}} \cap \overline{\ComplexityFont{GOOD}_{j_i}}~\bigg\vert~ \bigcap_{k=1}^{i-1} \ComplexityFont{ACCEPT}_{j_k} \cap \ComplexityFont{INVALID}_{j_k}\right] \\
        &\leq \Pr\left[\ComplexityFont{INVALID}_{j_{i}}  ~\bigg\vert~ \ComplexityFont{GOOD}_{j_i} \cap \bigcap_{k=1}^{i-1} \ComplexityFont{ACCEPT}_{j_k} \cap \ComplexityFont{INVALID}_{j_k}\right] + \Pr\left[\ComplexityFont{ACCEPT}_{j_{i}} \cap \overline{\ComplexityFont{GOOD}_{j_i}}~\bigg\vert~ \bigcap_{k=1}^{i-1} \ComplexityFont{ACCEPT}_{j_k} \cap \ComplexityFont{INVALID}_{j_k}\right] \\
        &\leq \delta_A + \delta_A = 2\delta_A.
    \end{align}
    \normalsize
    Here, for the first term, we have invoked the robustness guarantee of algorithm~$A$ from \Cref{equation:assumption-algorithm-solves-problem-robustly}. This guarantee is stated over all states having high fidelity with the target state. Therefore, conditioning on events from previous rounds only changes which state is fed to~$A$; as long as the conditioned input remains ``good,'' the bound coming from the robustness guarantee is unaffected by earlier events.
    For the second term, we have used the soundness guarantee of \Cref{theorem:verified-shadow-overlap-algorithm}. Crucially, this guarantee also holds over any input state, including those obtained after adversarial modifications. In particular, even after conditioning on $\bigcap_{k=1}^{i-1}\,\ComplexityFont{ACCEPT}_{j_k}\cap \ComplexityFont{INVALID}_{j_k}$, the state entering Step~4 in iteration~$j_i$ of \Cref{alg:amp-verifiable-public-quantum-private-classical}, $\sigma_{o_{1},\dots,o_{j_{i}-1}}$, where $o_{j}$ is the measurement outcome from round $j$, is admissible for the theorem. 
    Hence, $\Pr\left[\ComplexityFont{ACCEPT}_{j_{i}} \cap \overline{\ComplexityFont{GOOD}_{j_i}}~\bigg\vert~ \bigcap_{k=1}^{i-1} \ComplexityFont{ACCEPT}_{j_k} \cap \ComplexityFont{INVALID}_{j_k}\right]$ is bounded exactly as prescribed by \Cref{theorem:verified-shadow-overlap-algorithm} evaluated on $\sigma_{o_{1},\dots,o_{j_{i-1}}}$. This justifies treating the per-round bound as unaffected by the history of previous rounds in our analysis.
    
    \textbf{Privacy: }  Similar to \Cref{theorem:covert-shadow-overlap}, as we use fresh randomness in each of the $N\ell$ rounds, in each of the rounds, the state from the eavesdropper's perspective contains no information about $F$ as shown in \Cref{obs:target-covert-query}(ii), irrespective of whether \Cref{alg:covert-public-phase-query-randomness} or \Cref{alg:covert-public-phase-query-entanglement} is used. Therefore, throughout the interaction, the adversary learns nothing about $F$ and the joint function-adversary state factorizes. 
        
    \textbf{Efficiency: } In line 2, we prepare $N\ell m$ copies of uniform states augmented by either classical randomness or entanglement. Hence, in total we query the public oracle $\mathsf{O}_{\mathrm{pub}}^{\mathrm{QPh}}(f)$ for $N\ell m = \Tilde{\mathcal{O}}\left(\frac{n^5m^6 \ell}{\delta_A^2 \varepsilon_A^6}\right)$ times. 
    
    Each iteration $1 \leq j \leq \ell$, uses $Nnm$ qubits of the state received and runs $\ComplexityFont{CertifyState}^{\mathsf{O}^{\mathrm{Mem}}(f)}_{ \varepsilon_A, \delta_A}$ with $\tilde{\mathcal{O}}\left(\frac{(nm)^2\log(1/\delta_A)}{\varepsilon_A^2}\right)$ classical membership queries to $f^{\otimes m}$. The entire algorithm then uses $\mathcal{O}\left(\frac{n^2m^3 \ell \log(1/\delta_A)}{\varepsilon_A^2}\right)$ queries to $\mathsf{O}_{\mathrm{pri}}^{\mathrm{Mem}}(f)$ as we need $m$ queries to $\mathsf{O}_{\mathrm{pri}}^{\mathrm{Mem}}(f)$ to simulate one query to $\mathsf{O}_{\mathrm{pri}}^{\mathrm{Mem}}(f^{\otimes m})$. We run algorithm $A$ on the resulting output state of $\ComplexityFont{CertifyState}^{\mathsf{O}^{\mathrm{Mem}}(f)}_{ \varepsilon_A, \delta_A}$ in each round. Hence, $A$ is called $\ell$ times. 
\end{proof}

\begin{remark}\label{remark:robustness-assumption}
    \Cref{equation:assumption-algorithm-solves-problem-robustly} can be viewed as requiring that the quantum algorithm $A$ solves the problem $P$ ``robustly'' from $m$ copies with reasonably high success probability.
    This robustness assumption, however, is not really restrictive, which can be seen as follows: Suppose $A$ is a quantum algorithm that solves the problem $P$ from $m$ copies of $\ket{\psi_f^{\mathrm{Ph}}}$ with success probability $1-\tilde{\delta}_A$. Now, suppose $\bra{\psi_f^{\mathrm{Ph}}}^{\otimes m} \rho \ket{\psi_f^{\mathrm{Ph}}}^{\otimes m} \geq 1 - \varepsilon_A$. Then, by Fuchs-van de Graaf \cite{fuchs1999cryptographic}, we know that $\frac{1}{2}\norm{\rho - \ket{\psi_f^{\mathrm{Ph}}}\bra{\psi_f^{\mathrm{Ph}}}^{\otimes m}}_1\leq \sqrt{\varepsilon_A}$. So, given $\rho$, $A$ produces a valid solution with probability $\geq 1-\tilde{\delta}_A-\sqrt{\varepsilon
    _A}$. Thus, if we for example set $\delta_A  =  2 \tilde{\delta}_A$ and $\varepsilon_A=\tilde{\delta}_A^2$, \Cref{equation:assumption-algorithm-solves-problem-robustly} is satisfied.
\end{remark}

\begin{remark} \label{rmk:extension_to_estimation}
    In \Cref{corollary:verifiable-binary-distinguishing-from-phase-states} we have focused on binary distinguishing problems for simplicity of presentation. 
    The algorithm and proof can straightforwardly be extended to (promise) problems of distinguishing between more than two options. 
    We can also modify the procedure and proof to work for estimation problems instead.
    Namely, if we consider a problem of estimating a single parameter $\mu = \mu(f)$, and if we assume that $\bra{\psi_f^{\mathrm{Ph}}}^{\otimes m} \rho \ket{\psi_f^{\mathrm{Ph}}}^{\otimes m } \geq 1 - \varepsilon_A$ implies $\Pr[ |A(\rho) - \mu|\leq 2\varepsilon_A/5] \geq 1 - \delta_A$, we can replace the majority vote in Step 13 of \Cref{alg:amp-verifiable-public-quantum-private-classical} by the following rule: Identify a subset $S\subseteq [\ell]$ with $|S|\geq 2\ell/3$ such that $|X_i-X_j|\leq 4\varepsilon_A / 5$ holds for all $i,j\in S$. (If no such $S$ exists, abort.) Then, output the empirical average $\hat{\mu}\coloneqq \frac{1}{|S|}\sum_{i\in S} X_i$. If we pick $\ell$ large enough such that $\Pr[\mathrm{Binom}(\ell, 1-\delta_A)\ge 2\ell/3]\geq 1-\delta$, then also $\Pr[|\hat{\mu}-\mu|\leq\varepsilon_A]\geq 1-\delta$.\footnote{This can be seen as follows: Suppose $|X_i-\mu|\leq 2\varepsilon_A/5$ holds for at least a $(2/3)$-fraction of all $i\in [\ell]$. Clearly, those $i$ then form a valid set $S$ as specified in the rule, thus such an $S$ exists 
    Moreover, for any set $S\subseteq [\ell]$ of size $|S|\geq 2\ell/3$, we have $|X_i-\mu|\leq 2\varepsilon_A/5$ for at least a $(1/2)$-fraction of $i\in S$. If we additionally require that $|X_i-X_j|\leq 4\varepsilon_A/5$ for all $i,j\in S$, we conclude that the remaining at most $(1/2)$-fraction of $j\in S$ satisfies $|X_j-\mu|\leq 6\varepsilon_A / 5$. Thus, $|\hat{\mu} - \mu|\leq \frac{2\varepsilon}{5} + \frac{6\varepsilon}{2\cdot 5} = \varepsilon$.}
    In the same spirit, an extension to multi-parameter estimation is possible.
\end{remark} 

\begin{remark} \label{remark:interactive-verification}
    When focusing only on the completeness and soundness guarantees established in \Cref{corollary:verifiable-binary-distinguishing-from-phase-states}, the result constitutes a novel insight into interactive proofs for quantum learning. Namely, it shows: Any (binary) distinguishing task that can be solved from a polynomial number of phase state copies copies (possibly together with polynomially many classical queries) can also be solved by an interactive proof system $(V,P)$ in which the quantum verifier $V$ makes polynomially many queries to a classical membership query oracle and in which the honest 
    quantum prover $P$ makes polynomially many queries to a quantum phase state oracle. 
    Combining this with \cite[Corollary 2]{caro2024interactiveproofsverifyingquantum}---which states that a verifier cannot improve their query complexity for solving a many-vs-one distinguishing task by interacting with an untrusted prover who has access to a stronger oracle---we conclude: Any many-vs-one distinguishing task for Boolean functions that can be solved from polynomially many phase state copies (possibly together with polynomially many classical queries) can also be quantumly solved from polynomially many classical queries. In other words: For Boolean function many-vs-one distinguishing tasks, quantum phase states can have at most a polynomial (information-theoretic) advantage over classical queries.
\end{remark}

\paragraph{Loss of covertness against bidirectional adversaries.}
The algorithms in this subsection are covert against unidirectional adversaries that can only interact with the \emph{responses} to the oracle queries. Against a bidirectional adversary that can tamper with both the learner$\to$oracle and oracle$\to$learner channels, masking via randomness alone is no longer sufficient for privacy. In particular, such an adversary can either replace the learner’s query state with its own (as in \Cref{obs:no-covert-qmem}) or measure the learner’s query in some basis and prepare its own state. Concretely, against \Cref{alg:covert-public-phase-query-randomness}, the adversary could efficiently measure in the basis $\{\,\ket{\psi^{(r)}} = H^{\otimes n}\ket{r} \,=\, H^{\otimes n}X^{r}\ket{0^n} \,=\, Z^{r}(H\ket{0})^{\otimes n}\,\}_{r\in\{0,1\}^n}$ to learn the private string $r^*\in\{0,1\}^n$ used by the learner and then process the oracle’s output to obtain $\ket{\psi_f^{\mathrm{Ph}}}$, completely violating Version~1 of our privacy guarantee. Even if we weaken Version~1 to require privacy only upon acceptance, as in \Cref{obs:no-covert-qmem}, the attack persists whenever the function is learnable from one phase-state copy. As our current verification step only checks fidelity with $\ket{\psi_f^{\mathrm{Ph}}}^{\otimes m}$, the adversary can acquire $\ket{\psi_f^{\mathrm{Ph}}}$, learn $f$, and re-apply the mask using the observed $r^*$; the fidelity test still passes and we accept. This motivates augmenting the protocol with checks \emph{beyond} target-state fidelity that let us abort upon detection of information leakage.

We do so in the next subsection by observing that, in \Cref{alg:covert-public-phase-query-entanglement}, masking is implemented via entanglement between a quantum randomness register and the register on which the oracle is queried. In this setting, information leakage can be tied to breaking the entanglement between these registers. For an ancilla-free, i.i.d.\ adversary model, we will show that such entanglement breaking can be efficiently detected, yielding a mechanism to abort on privacy violations.

\begin{remark}
    In \Cref{alg:covert-public-phase-query-randomness}, the masking amounts to querying on the phase state corresponding to a \emph{linear} (parity) function, which is vulnerable to the measurement attack described above. One can strengthen the mask by using higher-degree functions (e.g., quadratic) or, under computational assumptions, a quantum-secure pseudorandom function; this hinders efficient recovery of the mask from a single query. Exploring stronger masking families may yield additional protection in regimes where the learning problem cannot be solved from a single phase-state query, but it does not overcome the worst-case impossibility of \Cref{obs:no-covert-qmem}. We therefore present linear masking to emphasize the minimal masking requirements under the unidirectional adversary model and leave guarantees with richer masking families to future work.
\end{remark}

\subsection{Covert Verifiable Quantum Learning Against Ancilla-Free i.i.d. Adversaries} \label{sec:iid-both-directions}

Given the discussion at the end of the previous subsection and in the introduction of \Cref{sec:public-quantum-oracle-private-classical-queries}, to obtain covertness even against an adversary that can modify both directions of the quantum channel, we must restrict either the adversary or the learning problem. We consider an \emph{ancilla‐free} adversary, i.e., one with no additional quantum memory of its own. Note that this in particular rules out the swap attack of \Cref{obs:no-covert-qmem}. 
In this setting, we will show that querying the public oracle via the entanglement‐based strategy of \Cref{alg:covert-public-phase-query-entanglement}, followed by a state‐verification step, achieves covertness. 

As the first step, we make the following simple but key observation. 

\begin{observation}[Ancilla-free and measurement-free adversary gains no information] \label{obs:ancilla-free-learns-nothing}
    Let $F\sim \mathcal{F}$ be drawn from a distribution $\mu$ over some function class $\mathcal{F}$.
    \begin{enumerate}
        \item\textbf{Completeness:} With no eavesdropper, given public oracle access to $\mathsf{O}_{\mathrm{pub}}^{\mathrm{QPh}}(f)$, \Cref{alg:covert-public-phase-query-entanglement} produces one copy of $\ket{\psi_f^{\mathrm{Ph}}}$.
        \item\textbf{Privacy:} Any ancilla‐free eavesdropper that makes no measurements before the public‐oracle call to $\mathsf{O}_{\mathrm{pub}}^{\mathrm{QPh}}(f)$ obtains no information about $F$.
    \end{enumerate}
\end{observation}

\begin{proof}
Item (i) follows directly from \Cref{obs:target-covert-query}(i).
For (ii), we describe the system immediately before the public oracle call in \Cref{alg:covert-public-phase-query-entanglement} (after Step 2). Let $\mathsf{F}$ denote the classical register holding $F$, $\mathsf{R}$ the private randomness register, and $\mathsf{Q}$ the quantum register sent to the oracle. The joint classical–quantum state is:
\begin{equation}
    \rho_{\mathsf{FRQ}}^{\mathrm{pre}}
    = \sum_{f:\{0,1\}^n\to\{0,1\}} \mu(f)\ket{f}\bra{f}_\mathsf{F} \otimes \frac{1}{2^n}\sum_{r, r'\in\{0,1\}^n} \ket{r}\bra{r'}_\mathsf{R} \otimes \ket{\psi^{(r)}}\bra{\psi^{(r')}}_\mathsf{Q} \, .
\end{equation}

An ancilla‐free and measurement‐free eavesdropper acts on the system $Q$ with a mixed unitary channel $\mathcal{E}(\cdot)=\sum_i p_i U_i (\cdot) U_i^\dagger$ on $Q$. 
This holds because, without ancillas or measurements, the only admissible adversary actions are unitaries that depend on the adversary's classical randomness.
After applying $\mathcal{E}$ and the public phase oracle unitary $\mathsf{O}_{\mathrm{pub}}^{\mathrm{QPh}}(f)$ on the $\mathsf{Q}$-system, the state becomes: 
\begin{equation}
    \rho_{\mathsf{FRQ}}^{\mathrm{post}}
    = \sum_{f:\{0,1\}^n\to\{0,1\}} \mu(f)\ket{f}\bra{f}_\mathsf{F} \otimes \frac{1}{2^n}\sum_{r, r'\in\{0,1\}^n} \ket{r}\bra{r'}_\mathsf{R} \otimes \mathsf{O}_{\mathrm{pub}}^{\mathrm{QPh}}(f)\mathcal{E}\left(\ket{\psi^{(r)}}\bra{\psi^{(r')}}_\mathsf{Q}\right) \mathsf{O}_{\mathrm{pub}}^{\mathrm{QPh}}(f) \, .
\end{equation}

The eavesdropper does not have access to the private randomness $\mathsf{R}$-system. Thus, the relevant classical-quantum state from the eavesdropper's perspective after the oracle is the reduced state:
\begin{equation}
\rho_{\mathsf{FQ}}^{\mathrm{post}} = \tr_{\mathsf{R}} \left(\rho_{\mathsf{FRQ}}^{\mathrm{post}}\right) = 
\sum_{f:\{0,1\}^n\to\{0,1\}} \mu(f)\ket{f}\bra{f}_\mathsf{F} \otimes \frac{1}{2^n}\sum_{r\in\{0,1\}^n} \mathsf{O}_{\mathrm{pub}}^{\mathrm{QPh}}(f)\mathcal{E}\left(\ket{\psi^{(r)}}\bra{\psi^{(r)}}_\mathsf{Q}\right) \mathsf{O}_{\mathrm{pub}}^{\mathrm{QPh}}(f) \, .
\end{equation}
Because $\{\ket{\psi^{(r)}}\}_{r}$ is a complete orthonormal basis of the $\mathsf{Q}$-register,
$\sum_{r} 
\ket{\psi^{(r)}}\bra{\psi^{(r)}}_\mathsf{Q} = \mathds{1}_\mathsf{Q}$. Moreover, as unitary conjugations leave the maximally mixed state invariant, we obtain:
\begin{equation}
\rho_{\mathsf{FQ}}^{\mathrm{post}}
= \sum_{f:\{0,1\}^n\to\{0,1\}} \mu(f)\ket{f}\bra{f}_\mathsf{F} \otimes \frac{\mathds{1}_\mathsf{Q}}{2^n}\, .
\end{equation}
This factorization shows that $\mathsf{F}$ and $\mathsf{Q}$ are independent, so the system the eavesdropper can access contains no information about $F$.
\end{proof}

This implies that any ancilla-free adversary must measure at least one of the $\mathsf{Q}$-register (this is the only way to perform a non-unitary action without ancillas) before the oracle action in order to obtain any information about $F$, once the state is returned. Next, we prove that any such measurement on $\mathsf{Q}$ reduces the fidelity between the state returned to the learner and the ideal phase state (on both $\mathsf{R}$ and $\mathsf{Q}$ registers) to at most $\frac{1}{2}$.

\begin{observation}[Pre-oracle measurement drops fidelity to ideal phase state] \label{thm:measurement-fidelity-drop}
    Fix any function $f$.
    Consider any ancilla-free eavesdropper that, in \Cref{alg:covert-public-phase-query-entanglement}, acts on the $\mathsf{Q}$-register before and after the public oracle call (Step~3). Suppose that, with probability $1-\delta_{\mathrm{leak}} \in [0, 1]$ over the eavesdropper’s internal randomness, the pre-oracle action on $\mathsf{Q}$ is measurement-free, and with probability $\delta_{\mathrm{leak}}$, the eavesdropper measures at least one qubit of $\mathsf{Q}$ before $\mathsf{O}_{\mathrm{pub}}^{\mathrm{QPh}}(f)$ is queried. 
    If $\sigma_{\mathsf{RQ}}$ is the state returned to the learner averaged over the eavesdropper’s internal randomness and the randomness of the outcomes of any measurements performed by the eavesdropper, then
    \begin{equation}
        \bra{\Psi_f^{\mathrm{Ph}}}_{\mathsf{RQ}}\,\sigma_{\mathsf{RQ}}\,\ket{\Psi_f^{\mathrm{Ph}}}_{\mathsf{RQ}} \le 1-\frac{\delta_{\mathrm{leak}}}{2} \, ,
    \end{equation}
    where the ideal target phase state is
    \begin{equation}
        \ket{\Psi_f^{\mathrm{Ph}}}_{\mathsf{RQ}} = \frac{1}{\sqrt{2^n}}\sum_{r\in\{0,1\}^n}\ket{r}\otimes \ket{\psi_{f}^{(r)}} \, .
    \end{equation}
\end{observation}

\begin{proof}
We decompose the eavesdropper’s strategy into two disjoint classes of branches, according to its private randomness: \emph{measurement-free pre-oracle} branches, and \emph{measuring pre-oracle} branches, in which at least one qubit of $\mathsf{Q}$ is measured before the public oracle.

For such a branch, the eavesdropper’s strategy and the public oracle query can be decomposed as:
\small
\begin{equation}
    \mathsf{Q} \xrightarrow{U} \mathsf{Q} \xrightarrow{\{\Pi_b\}_{b\in\{0,1\}}} (\mathsf{Q}, b) \xrightarrow[\text{(depends on $b$)}]{\text{pre-oracle instrument $\{\mathcal{E}_x\}$}} ( \mathsf{Q}, b, x) \xrightarrow{\mathsf{O}_{\mathrm{pub}}^{\mathrm{QPh}}(f)} (\mathsf{Q}, b, x) \xrightarrow[\text{(depends on $b, x$)}]{\text{post-oracle instrument $\{\mathcal{E}'_y\}$}} (\mathsf{Q}, b, x, y)\, ,
\end{equation}
\normalsize
where $U$ is a unitary on $\mathsf{Q}$, $\Pi_b$ projects the first qubit onto $\ket{b}$, and $\mathcal{E}$, $\mathcal{E}'$ are arbitrary quantum instruments.
 
Let $\mathsf{R}$ denote the learner’s private randomness register as in \Cref{alg:covert-public-phase-query-entanglement}, and let $\ket{\Omega}_{\mathsf{RQ}} = \frac{1}{\sqrt{2^n}} \sum_{r \in\{0,1\}^n} \ket{r}_\mathsf{R}\ket{r}_\mathsf{Q}$ be the canonical maximally entangled state. The state $\ket{\psi}_{\mathsf{RQ}} = \sum_{r \in\{0,1\}^n} \ket{r}_\mathsf{R}\ket{\psi^{(r)}}_\mathsf{Q}$
received by the eavesdropper is also maximally entangled, it can be expressed as
$\ket{\psi}_{\mathsf{RQ}} = (\mathds{1}_\mathsf{R} \otimes H^{\otimes n})\ket{\Omega}_{\mathsf{RQ}}$. Using the identity $(A\otimes I)\ket{\Omega}_{\mathsf{RQ}} = (I\otimes A^{\mathsf{T}})\ket{\Omega}_{\mathsf{RQ}}$, the joint $\mathsf{RQ}$ state immediately before the measurement is
$(H^{\otimes n}U^{\mathsf{T}}\otimes I)\ket{\Omega}_{\mathsf{RQ}}$.

If the measurement yields outcome $b$, the post-measurement state is
\begin{align}
    (I\otimes \Pi_b) (H^{\otimes n}U^{\mathsf{T}}\otimes I)\ket{\Omega}_{\mathsf{RQ}}
    &= (H^{\otimes n}U^{\mathsf{T}}\otimes I)\ket{\psi^b}_{\mathsf{RQ}} 
    = (H^{\otimes n}U^{\mathsf{T}}\otimes I) \left(\ket{b}_{\mathsf{R_1}}\ket{b}_{\mathsf{Q_1}}\otimes \ket{\Omega}_{\overline{\mathsf{R_1}},\overline{\mathsf{Q_1}}}\right)\, ,
\end{align}
where we set $\ket{\psi^b}_{\mathsf{RQ}} = \ket{b}_{\mathsf{R_1}}\ket{b}_{\mathsf{Q_1}}\otimes \ket{\Omega}_{\overline{\mathsf{R_1}},\overline{\mathsf{Q_1}}}$.
Averaging over $b$ gives the average post-measurement state,
\begin{align}
    (H^{\otimes n}U^{\mathsf{T}}\otimes I) \left(\sum_{b\in\{0, 1\}} p_b \ket{\psi^b}\bra{\psi^b}_{\mathsf{RQ}}\right) (UH^{\otimes n}\otimes I)\, .
\end{align}

Crucially, the operator $\sum_{b\in\{0, 1\}} p_b \ket{\psi^b}\bra{\psi^b}_{\mathsf{RQ}}$ has Schmidt number at most $2^{n-1}$, since every pure state in the ensemble has Schmidt rank exactly $2^{n-1}$ across the $\mathsf{R}{:}\mathsf{Q}$ cut. Now, $U$, $H^{\otimes n}$, the arbitrary quantum channels on $\mathsf{Q}$, and the public phase oracle $\mathsf{O}_{\mathrm{pub}}^{\mathrm{QPh}}(f)$ are all local operations with access to classical side information. Therefore, by monotonicity of Schmidt number under LOCC~\cite{terhal2000schmidt}, the final joint state in the each measuring branch, $\sigma_{\mathsf{RQ}}^{\mathrm{meas}}$, has Schmidt number at most $2^{n-1}$.

The bipartite target state $\ket{\Psi_f^{\mathrm{Ph}}}_{\mathsf{RQ}}$ is maximally entangled (i.e., it has Schmidt rank $2^{n}$). By~\cite{terhal2000schmidt}, any state on $\mathsf{R},\mathsf{Q}$ with Schmidt number $\le k$ has fidelity at most $k2^{-n}$ with any maximally entangled state on $\mathsf{R},\mathsf{Q}$. Applying this with $k=2^{n-1}$ yields, 
\begin{align}
    \bra{\Psi_f^{\mathrm{Ph}}}_{\mathsf{RQ}}\sigma_{\mathsf{RQ}}^{\mathrm{meas}}\ket{\Psi_f^{\mathrm{Ph}}}_{\mathsf{RQ}} &\le \frac{2^{n-1}}{2^n} = \frac{1}{2}\, .
\end{align}
Finally, since this occurs in each \emph{pre-oracle measuring} branch the overall fidelity is at most $1-\delta_{\mathrm{leak}} + \frac{\delta_{\mathrm{leak}}}{2} =1- \frac{\delta_{\mathrm{leak}}}{2}$.
\end{proof}

Now, we can combine \Cref{obs:ancilla-free-learns-nothing} and \Cref{thm:measurement-fidelity-drop} with \Cref{theorem:shadow-overlap-copy-complexity} to obtain the following covert verifiable algorithm inspired by \Cref{alg:covert-public-phase-query-entanglement} for acquiring copies of quantum phase states against i.i.d. ancilla-free adversaries:

\begin{algorithm}
    \caption{Covert Verifiable Phase States from Public Oracle against i.i.d.~Ancilla-free Adversaries}
    \begin{algorithmic}[1]
        \Require Number of qubits $n\in\mathbb{N}$; required copies $m\in\mathbb{N}$; parameters $\delta,\delta_{\mathrm{leak}},\varepsilon\in(0,1)$; access to $\mathsf{O}_{\mathrm{pub}}^{\mathrm{QPh}}(f)$ and $\mathsf{O}_{\mathrm{pri}}^{\mathrm{Mem}}(f)$.
        \Ensure ``reject'' or a candidate state
        \State Set $N=N(n,m,\varepsilon,\delta,\delta_{\mathrm{leak}})=\tilde{\mathcal{O}}\left(\frac{n m \log(1/\delta)}{\min\{\varepsilon,\varepsilon_{\mathrm{leak}}\}}\right)$ where $\varepsilon_{\mathrm{leak}}\coloneqq1-\left(1-\frac{\delta_{\mathrm{leak}}}{2}\right)^{m}$.
        \For{$1 \leq j \leq (N+1)m$} 
        \State Prepare the state $\left(\frac{1}{\sqrt{2^n}}\sum_{r\in\{0,1\}^n}\ket{r}\right)\otimes \left(\frac{1}{\sqrt{2^n}}\sum_{x\in\{0,1\}^n}\ket{x}\right)$. \Comment{Line 1 of \Cref{alg:covert-public-phase-query-entanglement}}
        \State For $1\leq i\leq n$, apply controlled-$Z$ gates between the $i$th register (as control) and the $(n+i)$th register (as target), thus preparing the state $\frac{1}{\sqrt{2^n}}\sum_{r\in\{0,1\}^n}\ket{r}\otimes \ket{\psi^{(r)}}$. \Comment{Line 2 of \Cref{alg:covert-public-phase-query-entanglement}}
        \State Send the last $n$ qubits to $\mathsf{O}_{\mathrm{pub}}^{\mathrm{QPh}}(f)$.  \Comment{Line 2 of \Cref{alg:covert-public-phase-query-entanglement}}
        \EndFor
                
        \State Let $\rho$ be the resulting $((N+1)\cdot 2nm)$-qubit state (after all oracle responses).
    \State Partition the $(N+1)$ two-register blocks into a \emph{certification set} $\mathsf{C}=\{1,\dots,N\}$ and an \emph{output set} $\mathsf{U}=\{N+1\}$.
    \State Using $\mathsf{O}_{\mathrm{pri}}^{\mathrm{Mem}}(f)$, apply the adaptive version of \Cref{theorem:shadow-overlap-copy-complexity} to the certification blocks of $\rho$ (i.e., to $\rho_\mathsf{C}:=\tr_{\mathsf{U}}[\rho]$) against $\ket{\Psi_f^{\mathrm{Ph}}}^{\otimes m}$ with parameters $(\min\{\varepsilon,(1 - \Omega(1))\varepsilon_{\mathrm{leak}}\}, \delta)$ and let the post-measurement state be $\rho'$ (including both $\mathsf{U}$ and $\mathsf{C}$ registers).
    \If{the certification test rejects}
      \State \textbf{return} ``reject''
    \Else
      \State On the output block of $\rho'$ (i.e., on $\rho'_\mathsf{U}$): for every $2n$ qubits, measure the first $n$ qubits in the computational basis to obtain string $r$, and then apply $Z^{r}$ to the remaining $n$ qubits. Discard the measured $n$ qubits. \Comment{Line 4 of \Cref{alg:covert-public-phase-query-entanglement}}
      \State \textbf{return} the resulting $(mn)$-qubit state.
    \EndIf
    \end{algorithmic}
    \label{alg:covert-verifiable-public-quantum-private-ancilla-free}
\end{algorithm}

\begin{theorem}[Covert Verifiable Phase States Against i.i.d.~Ancilla-free Adversaries---Formal version of \Cref{inf-theorem:covert-quantum-data-acquisition-v2}] \label{theorem:covert-shadow-overlap-ancilla-free} 
For $\delta, \delta_{\mathrm{leak}}, \varepsilon \in (0, 1)$, and $\varepsilon_{\mathrm{leak}} \coloneqq 1-\left(1-\frac{\delta_{\mathrm{leak}}}{2}\right)^m$, ~\Cref{alg:covert-verifiable-public-quantum-private-ancilla-free} is a $(m_{\mathrm{pri}} = \mathcal{O}\left(\frac{nm^2 \log(1/\delta)}{\min\{\varepsilon, \varepsilon_{\mathrm{leak}}\}}\right), m_{\mathrm{pub}} = \tilde{\mathcal{O}}\left( \frac{nm^2 \log(1/\delta)}{\min\{\varepsilon, \varepsilon_{\mathrm{leak}}\}} \right),\delta_c = 0,\delta_s = \delta, \varepsilon = \min\{\varepsilon, \varepsilon_{\mathrm{leak}}\})$-covert verifiable quantum algorithm for producing $m$ quantum phase states for any arbitrary concept class $\mathcal{F}$ against i.i.d. ancilla-free adversaries that extract information with probability at least $\delta_{\mathrm{leak}}$. That is: 
    \begin{itemize}
        \item \textbf{Completeness:} With no adversary, for any $f\in\mathcal{F}$, after making at most $m_{\mathrm{pri}}$ queries to $\mathsf{O}_{\mathrm{pri}}^{\mathrm{Mem}}(f)$ and at most $m_{\mathrm{pub}}$ queries to $\mathsf{O}_{\mathrm{pub}}^{\mathrm{QPh}}(f)$, \Cref{alg:covert-verifiable-public-quantum-private-ancilla-free} accepts and outputs $\ket{\psi_f^{\mathrm{Ph}}}\bra{\psi_f^{\mathrm{Ph}}}^{\otimes m}$  with success probability $1$.
        \item \textbf{Soundness:} For any $f\in\mathcal{F}$, after making at most $m_{\mathrm{pri}}$ queries to $\mathsf{O}_{\mathrm{pri}}^{\mathrm{Mem}}(f)$ and at most $m_{\mathrm{pub}}$ queries to $\mathsf{O}_{\mathrm{pub}}^{\mathrm{QPh}}(f)$, the latter of which are subject to corruption by an arbitrary i.i.d.~adversary, \Cref{alg:covert-verifiable-public-quantum-private-ancilla-free} accepts and outputs some $\sigma$ with $\bra{\psi_f^{\mathrm{Ph}}}^{\otimes m}\sigma\ket{\psi_f^{\mathrm{Ph}}}^{\otimes m}< 1-\min\{\varepsilon, \varepsilon_{\mathrm{leak}}\}$ with failure probability $\leq \delta$.
        \item \textbf{Privacy.} For $F\sim\mathcal{F}$, and any i.i.d.\,ancilla-free adversary that extracts information about the unknown function $F$ with probability at least $\delta_{\mathrm{leak}}$, \Cref{alg:covert-verifiable-public-quantum-private-ancilla-free} accepts with probability $ \leq \delta$. Here, the probability $\delta$ is over the learner’s internal randomness (and measurements), whereas $\delta_{\mathrm{leak}}$ is over the adversary’s randomness.

\end{itemize}
\end{theorem}

\begin{proof}
    \textbf{Completeness:}
If there is no adversary, each query to $\mathsf{O}_{\mathrm{pub}}^{\mathrm{QPh}}(f)$ implemented as in \Cref{alg:covert-public-phase-query-entanglement} prepares one copy of the two-register phase state $\ket{\Psi_f^{\mathrm{Ph}}}$ by \Cref{obs:target-covert-query}. Thus, after the loop, the global state of the $(N+1)$ two-register $m$-copy blocks is $\ket{\Psi_f^{\mathrm{Ph}}}\bra{\Psi_f^{\mathrm{Ph}}}^{\otimes (N+1)m}$. We then apply the adaptive variant of \Cref{theorem:shadow-overlap-copy-complexity}, using the access to $\mathsf{O}_{\mathrm{pri}}^{\mathrm{Mem}}(f)$, to the first $N$ blocks for certifying against the state $\ket{\Psi_f^{\mathrm{Ph}}}\bra{\Psi_f^{\mathrm{Ph}}}^{\otimes m}$; by completeness of \Cref{theorem:shadow-overlap-copy-complexity}, the test accepts with probability $1$. Finally, the measure-and-correct map from Step 13 in \Cref{alg:covert-public-phase-query-entanglement} sends each $\ket{\Psi_f^{\mathrm{Ph}}}$ to $\ket{\psi_f^{\mathrm{Ph}}}$, so the algorithm outputs $\ket{\psi_f^{\mathrm{Ph}}}\bra{\psi_f^{\mathrm{Ph}}}^{\otimes m}$ with unit probability.

\textbf{Soundness.}
Under the i.i.d.\ assumption, each two-register state returned by the public-oracle query has the same state $\rho_{\mathrm{out}}$. After the query loop, the pre-verification state is $\rho_{\mathrm{out}}^{\otimes (N+1)m}$. We group the first $Nm$ of these into $N$ verification blocks, each equal to $\rho_{\mathrm{out}}^{\otimes m}$, and let the last $m$-copy block be the output block. Then, we apply the adaptive variant of \Cref{theorem:shadow-overlap-copy-complexity} to the sufficient $N$ copies of $\rho_{\mathrm{out}}^{\otimes m}$, with access to $\mathsf{O}_{\mathrm{pri}}^{\mathrm{Mem}}(f)$, to test against $\ket{\Psi_f^{\mathrm{Ph}}}\bra{\Psi_f^{\mathrm{Ph}}}^{\otimes m}$. By the soundness guarantee in the theorem, we conclude that:
\begin{equation}
\Pr\left[\text{verification accepts and } \bra{\Psi_f^{\mathrm{Ph}}}^{\otimes m}\,\rho_{\mathrm{out}}^{\otimes m}\,\ket{\Psi_f^{\mathrm{Ph}}}^{\otimes m}
< 1-\min\{\varepsilon,\varepsilon_{\mathrm{leak}}\}\right] \le \delta\, .
\end{equation}
Let $\mathcal{N}$ be the CPTP “measure-and-correct’’ map from Step 13 (which in particular sends $\ket{\Psi_f^{\mathrm{Ph}}}\mapsto\ket{\psi_f^{\mathrm{Ph}}}$ on each of the $m$ copies), and let $\sigma=\mathcal{N}(\rho_{\mathrm{out}}^{\otimes m})$ be the final $(mn)$-qubit state output upon acceptance. By monotonicity of fidelity under CPTP maps,
\begin{equation}
\mathcal{F}\left(\ket{\Psi_f^{\mathrm{Ph}}}\bra{\Psi_f^{\mathrm{Ph}}}^{\otimes m},~\rho_{\mathrm{out}}^{\otimes m}\right) \le \mathcal{F}\left(\ket{\psi_f^{\mathrm{Ph}}}\bra{\psi_f^{\mathrm{Ph}}}^{\otimes m},~\sigma\right).
\end{equation}
Hence:
\begin{align}
&\Pr\left[\text{verification accepts and }
\bra{\psi_f^{\mathrm{Ph}}}^{\otimes m}\sigma\ket{\psi_f^{\mathrm{Ph}}}^{\otimes m}
< 1-\min\{\varepsilon,\varepsilon_{\mathrm{leak}}\}\right]\\
&\le
\Pr\left[\text{verification accepts and }
\bra{\Psi_f^{\mathrm{Ph}}}^{\otimes m}\rho_{\mathrm{out}}^{\otimes m}\ket{\Psi_f^{\mathrm{Ph}}}^{\otimes m}
< 1-\min\{\varepsilon,\varepsilon_{\mathrm{leak}}\}\right]\\
& \le \delta.
\end{align}

\textbf{Privacy:} As above, under the i.i.d.\ assumption, each two-register state returned by the public-oracle query is $\rho_{\mathrm{out}}$, making the total pre-verification state $\rho_{\mathrm{out}}^{\otimes (N+1)m}$. The first $N$ i.i.d.\ $m$-copy blocks are used to detect privacy violations, while the last $m$-copy block is output after post-processing. The key idea is that, by \Cref{obs:ancilla-free-learns-nothing}, in any branch of the adversary's i.i.d.\ strategy where it does not measure before the oracle action, it learns nothing about $F$. Therefore, for the adversary to gain information about $F$ with probability at least $\delta_{\mathrm{leak}}$, it must measure the received state pre-oracle with at least $\delta_{\mathrm{leak}}$ probability. \Cref{thm:measurement-fidelity-drop} implies that, in that case,
\begin{equation}
\bra{\Psi_f^{\mathrm{Ph}}}\rho_{\mathrm{out}}\ket{\Psi_f^{\mathrm{Ph}}} \le 1 - \frac{\delta_{\mathrm{leak}}}{2}.
\end{equation}

In particular, for each $m$-copy block,
\begin{equation}
\bra{\Psi_f^{\mathrm{Ph}}}^{\otimes m}\rho_{\mathrm{out}}^{\otimes m}\ket{\Psi_f^{\mathrm{Ph}}}^{\otimes m} \le \left(1 - \frac{\delta_{\mathrm{leak}}}{2}\right)^m.
\end{equation}

We set $\varepsilon_{\mathrm{leak}} \coloneqq 1 - \left(1 - \frac{\delta_{\mathrm{leak}}}{2}\right)^m$, and have the certification procedure in Line~9 test against $\min\{\varepsilon,(1 - \Omega(1))\varepsilon_{\mathrm{leak}}\}$ using the sufficient $\tilde{\mathcal{O}}\left(\frac{nm \log(1/\delta)}{\min\{\varepsilon,\varepsilon_{\mathrm{leak}}\}}\right)$ copies of $\rho_{\mathrm{out}}^{\otimes m}$, via the adaptive version of \Cref{theorem:shadow-overlap-copy-complexity}. The soundness of \Cref{theorem:shadow-overlap-copy-complexity} guarantees that against an adversary who gains information in each round with probability at least $\delta_{\mathrm{leak}}$, $\Pr\left[\text{verification accepts}\right] \leq \delta$. Note that the probability above is over the randomized measurements of the certification procedure in \Cref{theorem:shadow-overlap-copy-complexity}.
    
\textbf{Efficiency:} Using the i.i.d.\ certification protocol from \Cref{theorem:shadow-overlap-copy-complexity} to certify $\ket{\Psi_f^{\mathrm{Ph}}}$ to error $\min\{\varepsilon,(1 - \Omega(1))\varepsilon_{\mathrm{leak}}\}$ requires $N = N(n,m,\varepsilon,\varepsilon_{\mathrm{leak}},\delta)
= \tilde{\mathcal{O}}\left(\frac{n m \log(1/\delta)}{\min\{\varepsilon,\varepsilon_{\mathrm{leak}}\}}\right)$
$m$-copy, two-register blocks. Each $m$-copy block uses $m$ queries to $\mathsf{O}_{\mathrm{pub}}^{\mathrm{QPh}}(f)$, so in total we make $
\tilde{\mathcal{O}}\left(\frac{n m^{2} \log(1/\delta)}{\min\{\varepsilon,\varepsilon_{\mathrm{leak}}\}}\right)$
public oracle queries. Moreover, for each $m$-copy block we perform two classical membership checks tied to $\ket{\Psi_f^{\mathrm{Ph}}}^{\otimes m}$; each check can be implemented using $m$ queries to $\mathsf{O}_{\mathrm{pri}}^{\mathrm{Mem}}(f)$.
\end{proof}

In the above algorithm, the same state certification procedure suffices for both soundness and privacy-violation detection. This follows because \Cref{obs:ancilla-free-learns-nothing} together with \Cref{thm:measurement-fidelity-drop} links any privacy violation to a drop in fidelity for an ancilla-free adversary.

\begin{remark} Although the privacy guarantee is proved only against i.i.d.\ ancilla-free adversaries, the soundness guarantee does not rely on ancilla-freeness and holds against arbitrary i.i.d.\ adversaries. Moreover, the privacy argument extends beyond the ancilla-free model: it suffices that the adversary’s operation on the query register $\mathsf{Q}$ is entanglement breaking in some way (in our setting, we force the adversary to break entanglement by measuring to learn about $F$ due to the ancilla-free assumption). After that point, we can consider a general adversary. This is because once the fidelity with the entangled target state drops; subsequent local CPTP processing cannot increase it. In particular, as indicated in the proof of \Cref{thm:measurement-fidelity-drop}, in the considered adversary strategy, we don't require ancilla-freeness after the first measurement is made. However, our ancilla-free model is strictly weaker than bounded-storage models \cite{damgaard2008cryptography}, where the adversary’s quantum memory is capped at specified points in the protocol but unbounded computation is allowed before and after those points. Note that the standard no-go attack \Cref{obs:no-covert-qmem} relies on coherently retaining the learner’s entire $n$-qubit query across the pre- and post-oracle steps. If inter-round quantum storage is limited to $s<n$ qubits with a bounded storage assumption, that strategy fails: trying to keep all $n$ qubits forces the adversary to irreversibly discard/compress $(n-s)$ qubits, which our certification detects. Formally extending our guarantees to bounded storage would require an explicit memory–disturbance lower bound; we leave this as an interesting direction for future work.
\end{remark}

\begin{remark}
Note that for the original algorithm $A$ to be sample-efficient, we require $m = \mathcal{O}(\poly(n))$. In this setting, \Cref{alg:covert-verifiable-public-quantum-private-ancilla-free}  yields an efficient, covert and verifiable method for obtaining quantum phase states against adversaries that gain information with probability at least $\frac{1}{\ComplexityFont{poly}(n,m)}$. The learner can tune the tolerated leakage level by choosing the parameter $\delta_{\mathrm{leak}}$. More generally, if $\delta_{\mathrm{leak}}$ is noticeable\footnote{That is, the tolerated probability $\delta_{\mathrm{leak}}=\delta_{\mathrm{leak}}(n,m)$ of extracting information satisfies: $\exists c, c', n_0, m_0\in\mathbb{N}$ s.t.~$\forall n\geq n_0$ and $\forall m\geq m_0:$ $\delta_{\mathrm{leak}}\geq n^{-c}m^{-c'}$.}, then our covert verifiable algorithm will be efficient. As is to be expected, if we set $\delta_{\mathrm{leak}} \propto \frac{1}{\ComplexityFont{exp}(m,n)}$, meaning that the adversary gains information with exponentially small probability, the covert verifiable learning procedure becomes intractable, since it would require distinguishing quantum states that are inverse-exponentially (in $n, m$) close in fidelity. 
\end{remark}

Similar to \Cref{corollary:verifiable-binary-distinguishing-from-phase-states}, we can now extend \Cref{alg:covert-verifiable-public-quantum-private-ancilla-free} to obtain a protocol for solving (binary distinguishing) tasks covertly against i.i.d.\ ancilla-free adversaries. 

\begin{algorithm}
    \caption{\Cref{alg:covert-verifiable-public-quantum-private-ancilla-free} applied to distinguishing task}
    \begin{algorithmic}[1]
        \Require Number of qubits $n\in\mathbb{N}$; distinguishing algorithm $A$ with parameters $\varepsilon_A\in (0,1)$, $m \in \mathbb{N}$ and $\delta_A\in (0,1/2)$; access to oracles $\mathsf{O}_{\mathrm{pub}}^{\mathrm{QPh}}(f)$, $\mathsf{O}_{\mathrm{pri}}^{\mathrm{Mem}}(f)$; $\delta \in (0, 1), \delta_{\mathrm{leak}} \in (0, 1)$.
        \Ensure ``reject'' or candidate solution
        \State Using $\mathsf{O}_{\mathrm{pub}}^{\mathrm{QPh}}(f)$ and $\mathsf{O}_{\mathrm{pri}}^{\mathrm{Mem}}(f)$, run \Cref{alg:covert-verifiable-public-quantum-private-ancilla-free} with parameters $(\delta, \varepsilon = \varepsilon_A, \delta_{\mathrm{leak}})$ against the $m$-copy phase state
        \If{\Cref{alg:covert-verifiable-public-quantum-private-ancilla-free} rejects} 
            \State \Return ``reject''
        \Else
        \State Denote the resulting state by $\rho$
        \State \Return $A(\rho)$
        \EndIf
    \end{algorithmic}\label{alg:verifiable-public-quantum-private-classical-task-ancilla-free}
\end{algorithm}

\begin{corollary}\label{corollary:verifiable-binary-distinguishing-from-phase-states-ancilla-free}
    Consider a binary distinguishing problem $P$ for Boolean functions. Let $\ket{\psi_f^{\mathrm{Ph}}}$ be the phase state of a Boolean function $f: \{0, 1\}^n \to \{0, 1\}$. Let $A$ be a quantum algorithm for solving $P$ that, upon input of any $(mn)$-qubit state $\rho$, satisfies the following guarantee for some $\varepsilon_A \in (0,1), \delta_A\in (0,\frac{1}{2}),~m\in \mathbb{N}$: 
    \begin{align}\label{equation:assumption-algorithm-solves-problem-robustly-2}
        \text{ if } \bra{\psi_f^{\mathrm{Ph}}}^{\otimes m} \rho \ket{\psi_f^{\mathrm{Ph}}}^{\otimes m } \geq 1 - \varepsilon_A, \text{ then } \Pr[ A(\rho) \text{ is a valid solution to }P\text{ for }f] \geq 1 - \delta_A\, .
    \end{align}
    For $\delta \in (0, 1), ~\delta_{\mathrm{leak}}\in(0, 1)$ and $\varepsilon_{\mathrm{leak}}:=1-\left(1-\frac{\delta_{\mathrm{leak}}}{2}\right)^m$, \Cref{alg:verifiable-public-quantum-private-classical-task-ancilla-free} gives a quantum algorithm $A'$ that, using at most $m_{\mathrm{pub}}=\Tilde{\mathcal{O}}\left(\frac{nm^2 \log(1/\delta)}{\min\{\varepsilon_A, \varepsilon_{\mathrm{leak}}\}}\right)$ queries to a public quantum phase oracle $\mathsf{O}_{\mathrm{pub}}^{\mathrm{QPh}}(f)$ and at most $m_{\mathrm{pri}}=\Tilde{\mathcal{O}}\left(\frac{nm^2 \log(1/\delta)}{\min\{\varepsilon_A, \varepsilon_{\mathrm{leak}}\}}\right)$ queries to a private classical membership query oracle $\mathsf{O}_{\mathrm{pri}}^{\mathrm{Mem}}(f)$, achieves the following guarantees:
    \begin{itemize}
        \item \textbf{Completeness:}  With no adversary, for any $f\in\mathcal{F}$, after making at most $m_{\mathrm{pri}}$ queries to $\mathsf{O}_{\mathrm{pri}}^{\mathrm{Mem}}(f)$ and at most $m_{\mathrm{pub}}$ queries to $\mathsf{O}_{\mathrm{pub}}^{\mathrm{QPh}}(f)$, 
        \begin{equation}
            \Pr\left[A'^{\mathsf{O}_{\mathrm{pri}}^{\mathrm{Mem}}(f), \mathsf{O}_{\mathrm{pub}}^{\mathrm{QPh}}(f)}\text{ accepts and outputs a valid solution to } P \text{ for }f\right]
            \geq 1-\delta_A\, .
        \end{equation}
        \item \textbf{Soundness:} For any $f\in\mathcal{F}$, after making at most $m_{\mathrm{pri}}$ queries to $\mathsf{O}_{\mathrm{pri}}^{\mathrm{Mem}}(f)$ and at most $m_{\mathrm{pub}}$ queries to $\mathsf{O}_{\mathrm{pub}}^{\mathrm{QPh}}(f)$, the latter of which are subject to corruption by an arbitrary i.i.d.\ adversary,  
        \begin{equation}
            \Pr\left[A'^{\mathsf{O}_{\mathrm{pri}}^{\mathrm{Mem}}(f), \mathsf{O}_{\mathrm{pub}}^{\mathrm{QPh}}(f)}\text{ accepts and outputs an invalid solution to } P \text{ for }f\right]
            \leq \delta + \delta_A\, .
        \end{equation}
        \item \textbf{{Privacy}:} For $F\sim\mathcal{F}$, and any i.i.d.\,ancilla-free adversary that extracts information about the unknown function $F$ with probability at least $\delta_{\mathrm{leak}}$, \Cref{alg:verifiable-public-quantum-private-classical-task-ancilla-free} accepts with probability $ \leq \delta$. Here, the probability $\delta$ is over the learner’s internal randomness (and measurements), whereas $\delta_{\mathrm{leak}}$ is over the adversary’s randomness.
        \item \textbf{Efficiency:} $A'$ runs in time $\mathcal{O}(\poly(n, m, \log(1/\delta), 1/\varepsilon_A, 1/\varepsilon_{\mathrm{leak}}, \ComplexityFont{Complexity}(A))$ where  ~$\ComplexityFont{Complexity}(A)$ is the run time of the original algorithm $A$. In addition to the quantum operations used in $A$, $A'$ only employs uniform superposition state preparation, controlled-$Z$ gates and adaptive single-qubit gates.
    \end{itemize}
\end{corollary}

As in \Cref{theorem:covert-shadow-overlap-ancilla-free}, the soundness in \Cref{corollary:verifiable-binary-distinguishing-from-phase-states-ancilla-free} is against general i.i.d\ adversaries. Only the privacy in \Cref{corollary:verifiable-binary-distinguishing-from-phase-states-ancilla-free} requires restricting the adversary to be ancilla-free.

\begin{proof}
    \textbf{Completeness: } Due to completeness of \Cref{theorem:covert-shadow-overlap-ancilla-free}, Line 1 never rejects and the output state $\rho$ is equals \(\ket{\psi_f^{\mathrm{Ph}}}\bra{\psi_f^{\mathrm{Ph}}}^{\otimes m}\) with probability $1$ . Now by the robustness guarantee of $A$, we have 
    \begin{equation}
        \Pr[ A(\rho) \text{ is a valid solution to }P\text{ for }f] \ge 1 - \delta_A \, .
    \end{equation}
    
    \textbf{Soundness: } Here, we use the soundness guarantee of \Cref{theorem:covert-shadow-overlap-ancilla-free} and the robustness guarantee of algorithm $A$ to argue that we either ``reject'' or output the valid solution with high probability. We only err if we accept in Line 1 but $A(\rho)$ is invalid. Therefore, we want to upper bound the following probability: 

\begin{align}
    \Pr\left[\ComplexityFont{ACCEPT} \cap \ComplexityFont{INVALID}\right]
        &= \Pr\left[\ComplexityFont{ACCEPT} \cap \ComplexityFont{INVALID}  \cap \ComplexityFont{GOOD}\right] + \Pr\left[\ComplexityFont{ACCEPT} \cap \ComplexityFont{INVALID}  \cap \overline{\ComplexityFont{GOOD}}\right]\\
        &\leq \Pr\left[\ComplexityFont{INVALID}  \cap \ComplexityFont{GOOD}\right] + \Pr\left[\ComplexityFont{ACCEPT} \cap \overline{\ComplexityFont{GOOD}}\right] \\
        &\leq \Pr\left[\ComplexityFont{INVALID}  ~\bigg\vert~ \ComplexityFont{GOOD} \right] + \Pr\left[\ComplexityFont{ACCEPT} \cap \overline{\ComplexityFont{GOOD}}\right] \\
        &\leq \delta_A + \delta, 
    \end{align}
    \normalsize
    
    where $\ComplexityFont{ACCEPT}$ is the event that Line 1 accepts, $\ComplexityFont{INVALID}$ is the event that Line 6 returns an invalid outcome and $\ComplexityFont{GOOD}$ is the event that $\bra{\psi_f^{\mathrm{Ph}}}^{\otimes m} \rho \ket{\psi_f^{\mathrm{Ph}}}^{\otimes m} \ge 1-\varepsilon_A$.  Here, for the first term, we invoke the robustness guarantee of algorithm~$A$ from \Cref{equation:assumption-algorithm-solves-problem-robustly-2} and for the second term, we use the soundness guarantee of \Cref{theorem:covert-shadow-overlap-ancilla-free}. 
    
    \textbf{Privacy: }  As the adversary only interacts with the query transcript via the invocation of \Cref{theorem:covert-shadow-overlap-ancilla-free}, this is directly inherited from \Cref{theorem:covert-shadow-overlap-ancilla-free}.
        
    \textbf{Efficiency: } We run \Cref{alg:covert-verifiable-public-quantum-private-ancilla-free} once and then run $A$ on the output giving us the claimed scaling. 
\end{proof}

\Cref{corollary:verifiable-binary-distinguishing-from-phase-states-ancilla-free} shows that any quantum algorithm $A$ that uses phase states can be made covert with only polynomial overhead. For emphasis, we let the covert algorithm $A'$ inherit $A$’s internal confidence parameter $\delta_A$ for completeness and soundness. However, as in \Cref{corollary:verifiable-binary-distinguishing-from-phase-states}, we can repeat the protocol $\ell$ times based on $(\delta, \delta_A)$, and then employ a majority vote to amplify the completeness and soundness to  $\delta$-error as well. This amplified version will be used in \Cref{sec:phase-state-apps}.

We note that the discussions from \Cref{remark:robustness-assumption} (regarding obtaining robust quantum algorithms) and \Cref{rmk:extension_to_estimation} (regarding solving estimation problems covertly) also directly apply to \Cref{corollary:verifiable-binary-distinguishing-from-phase-states-ancilla-free}. 

\begin{remark}[From $\mathsf{O}^{\mathrm{QPh}}_{\mathrm{pub}}$ and phase states to $\mathsf{O}^{\mathrm{Mem}}_{\mathrm{pri}}$ and quantum examples]\label{remark:covert-quantum-examples-from-QMem}
    In \Cref{remark:phase-states-vs-quantum-examples}, we have argued that certification of phase states from $\mathsf{O}^{\mathrm{QPh}}_{\mathrm{pub}}$- and $\mathsf{O}^{\mathrm{Mem}}_{\mathrm{pri}}$-access enables certification of quantum example states from $\mathsf{O}^{\mathrm{QMem}}_{\mathrm{pub}}$- and $\mathsf{O}^{\mathrm{Mem}}_{\mathrm{pri}}$-access.
    Hence, as a consequence of the above, we can also efficiently achieve completeness and soundness in acquiring and learning from quantum example state data from $\mathsf{O}^{\mathrm{QMem}}_{\mathrm{pub}}$- and $\mathsf{O}^{\mathrm{Mem}}_{\mathrm{pri}}$-access.
    Let us argue that also our covertness guarantees carry over. 
    To do so, we first describe a suitable modification of \Cref{alg:covert-public-phase-query-randomness}.
    Namely, consider querying $\mathsf{O}^{\mathrm{QMem}}_{\mathrm{pub}}(f)$ on the input register $\mathsf{IN}$ and the auxiliary register $\mathsf{AUX}$ of the state 
    \begin{align}
        &(H^{\otimes n}\ket{r})_{\mathsf{IN}}\otimes \left((H^{\otimes m})_{\mathsf{AUX}}\mathrm{CNOT}_{\mathsf{OUT}\to\mathsf{AUX}}((H^{\otimes m}\ket{\tilde{r}})_{\mathsf{OUT}}\otimes \ket{0^m}_{\mathsf{AUX}})\right)\\
        &=\frac{1}{\sqrt{2^{n+m}}}\sum_{x\in\{0,1\}^n}\sum_{y\in\{0,1\}^m}(-1)^{r\cdot x + \tilde{r}\cdot y}\ket{x}_{\mathsf{IN}}\otimes \ket{y}_{\mathsf{OUT}}\otimes (H^{\otimes m}\ket{y})_{\mathsf{AUX}}\\
        &= \frac{1}{\sqrt{2^{n+2m}}}\sum_{x\in\{0,1\}^n}\sum_{y,z\in\{0,1\}^m}(-1)^{r\cdot x + \tilde{r}\cdot y + y\cdot z}\ket{x}_{\mathsf{IN}}\otimes \ket{y}_{\mathsf{OUT}}\otimes \ket{z}_{\mathsf{AUX}} \, ,
    \end{align}
    where we think of the $\mathsf{OUT}$-register as remaining a private register on the learner side.
    Following the same calculations as in the standard phase kickback trick for functions outputting a single bit, the oracle query transforms this state to
    \begin{align}
        &(\mathsf{O}^{\mathrm{QMem}}_{\mathrm{pub}}(f)_{\mathsf{IN},\mathsf{AUX}}\otimes \mathds{1}_{\mathsf{OUT}})\left((H^{\otimes n}\ket{r})_{\mathsf{IN}}\otimes \left((H^{\otimes m})_{\mathsf{AUX}}\mathrm{CNOT}_{\mathsf{OUT}\to\mathsf{AUX}}((H^{\otimes m}\ket{\tilde{r}})_{\mathsf{OUT}}\otimes \ket{0^m}_{\mathsf{AUX}})\right)\right)\\
        &= \frac{1}{\sqrt{2^{n+2m}}}\sum_{x\in\{0,1\}^n}\sum_{y,z\in\{0,1\}^m}(-1)^{r\cdot x + \tilde{r}\cdot y + y\cdot z}\ket{x}_{\mathsf{IN}}\otimes \ket{y}_{\mathsf{OUT}}\otimes \ket{z\oplus f(x)}_{\mathsf{AUX}}\\
        &= \frac{1}{\sqrt{2^{n+2m}}}\sum_{x\in\{0,1\}^n}\sum_{y,z\in\{0,1\}^m}(-1)^{r\cdot x + \tilde{r}\cdot y + y\cdot f(x) + y\cdot z}\ket{x}_{\mathsf{IN}}\otimes \ket{y}_{\mathsf{OUT}}\otimes \ket{z}_{\mathsf{AUX}}\\
        &= \frac{1}{\sqrt{2^{n+m}}}\sum_{x\in\{0,1\}^n}\sum_{y\in\{0,1\}^m}(-1)^{r\cdot x + \tilde{r}\cdot y + y\cdot f(x)}\ket{x}_{\mathsf{IN}}\otimes \ket{y}_{\mathsf{OUT}}\otimes (H^{\otimes m}\ket{y})_{\mathsf{AUX}} \, .\label{eq:phase-kickback-state-randomness}
    \end{align}
    After uncomputing $H^{\otimes m}_{\mathsf{AUX}}$ and $\mathrm{CNOT}_{\mathsf{OUT}\to\mathsf{AUX}}$, we obtain $\frac{1}{\sqrt{2^{n+m}}}\sum_{x\in\{0,1\}^n}\sum_{y\in\{0,1\}^m}(-1)^{r\cdot x + \tilde{r}\cdot y + y\cdot f(x)}\ket{x}_{\mathsf{IN}}\otimes \ket{y}_{\mathsf{OUT}}\otimes \ket{0^m}_{\mathsf{AUX}} = \ket{\psi_{\tilde{f}}^{(r\tilde{r})}}_{\mathsf{IN},\mathsf{OUT}}\otimes \ket{0^m}_{\mathsf{AUX}}$. So, the learner can now throw away the auxiliary $\mathsf{AUX}$-register and is left with the phase state $\ket{\psi_{\tilde{f}}^{(r\tilde{r})}}$.
    With this state, they can run an analogue of \Cref{alg:covert-public-phase-query-randomness} to perform phase state certification for the function $\tilde{f}(x,y)=y\cdot f(x)$ from \Cref{remark:phase-states-vs-quantum-examples} with private randomness $r\tilde{r}\sim\{0,1\}^{n+m}$.
    In particular, notice also that the reduced state of \Cref{eq:phase-kickback-state-randomness} on the joint $\mathsf{IN},\mathsf{AUX}$- register when averaged over the randomness $r\tilde{r}$ is the maximally mixed state on $n+m$ qubits\footnote{In fact, this is true even without averaging over the randomness $\tilde{r}$, we include it here for consistency of the presentation with the second variant involving entanglement instead of private randomness.} and in particular, is independent of $f$. Thus, we get the same privacy guarantee as in \Cref{obs:target-covert-query}. Hence, our covert verifiable guarantees for quantum data acquisition against unidirectional adversaries (compare \Cref{theorem:covert-shadow-overlap}) carry over to acquiring quantum examples from $\mathsf{O}^{\mathrm{QMem}}_{\mathrm{pub}}$- and $\mathsf{O}^{\mathrm{Mem}}_{\mathrm{pri}}$-access.
    Similarly, also our covert verifiable guarantees for solving distinguishing and estimation problems (compare \Cref{corollary:verifiable-binary-distinguishing-from-phase-states} and the discussion thereafter) carry over.
    
    Similarly to the step from \Cref{alg:covert-public-phase-query-randomness} to \Cref{alg:covert-public-phase-query-entanglement}, we can replace private randomness by entanglement by querying $\mathsf{O}^{\mathrm{QMem}}_{\mathrm{pub}}(f)$ on the input register $\mathsf{IN}$ and the auxiliary register $\mathsf{AUX}$ of the state
    \begin{align}
        &\frac{1}{\sqrt{2^{n+m}}}\sum_{r\in\{0,1\}^n}\sum_{\tilde{r}\in\{0,1\}^m}\ket{r}_\mathsf{R}\otimes \ket{\tilde{r}}_{\tilde{\mathsf{R}}}\otimes (H^{\otimes n}\ket{r})_{\mathsf{IN}}\otimes \left((H^{\otimes m})_{\mathsf{AUX}}\mathrm{CNOT}_{\mathsf{OUT}\to\mathsf{AUX}}((H^{\otimes m}\ket{0^m})_{\mathsf{OUT}}\otimes \ket{\tilde{r}}_{\mathsf{AUX}})\right) \\
        &= (H^{\otimes m})_{\mathsf{IN}}(H^{\otimes m})_{\mathsf{AUX}}\mathrm{CNOT}_{\mathsf{OUT}\to\mathsf{AUX}}(H^{\otimes m})_{\mathsf{OUT}}\left(\ket{\Omega}_{\mathsf{R},\mathsf{IN}} \otimes \ket{0^m}_{\mathsf{OUT}}\otimes \ket{\Omega}_{\tilde{\mathsf{R}},\mathsf{AUX}}\right) \, .
    \end{align}
    Now we can imitate our earlier argument: 
    By our above reasoning (and in analogy to \Cref{{obs:ancilla-free-learns-nothing}}), an ancilla-free adversary learns nothing about the unknown function unless they perform a measurement on the joint $\mathsf{IN},\mathsf{AUX}$-register before the  $\mathsf{O}^{\mathrm{QMem}}_{\mathrm{pub}}$-query. 
    However, if they perform such a measurement, then we can use the identities $(A_{\mathsf{S},\mathsf{R}}\otimes \mathds{1}_{\mathsf{IN}})(\ket{\psi}_{\mathsf{S}}\otimes\ket{\Omega}_{\mathsf{R},\mathsf{IN}}) = (A_{\mathsf{S},\mathsf{IN}}^{\top_{\mathsf{IN}}}\otimes \mathds{1}_{\mathsf{R}})(\ket{\psi}_{\mathsf{S}}\otimes\ket{\Omega}_{\mathsf{R},\mathsf{IN}})$ and similarly for $\tilde{\mathsf{R}},\mathsf{AUX}$ instead of $\mathsf{R},\mathsf{IN}$ to argue analogously to the proof of \Cref{thm:measurement-fidelity-drop} that this leads to a decrease in the Schmidt number of the average post-measurement state and hence to a drop in fidelity with the original maximally entangled state, which the learner can detect.
    With this, our covert verifiable guarantees for quantum data acquisition against i.i.d.~ancilla-free adversaries (compare \Cref{theorem:covert-shadow-overlap-ancilla-free}) carry over to acquiring quantum examples from $\mathsf{O}^{\mathrm{QMem}}_{\mathrm{pub}}$- and $\mathsf{O}^{\mathrm{Mem}}_{\mathrm{pri}}$-access.
    Similarly, also our covert verifiable guarantees for solving distinguishing and estimation problems (compare \Cref{corollary:verifiable-binary-distinguishing-from-phase-states-ancilla-free} and the discussion thereafter) carry over.
\end{remark}

\subsection{Applications} \label{sec:phase-state-apps}

Recall that classical covert learning in the \cite{canetti2021covert} model requires the learning problem to be computationally hard from access to only the private oracle and computationally easy from access to the public oracle for the setting to be meaningful. 
Similarly, for \Cref{corollary:verifiable-binary-distinguishing-from-phase-states} and \Cref{corollary:verifiable-binary-distinguishing-from-phase-states-ancilla-free} to be meaningful, we must consider problems in which the classical query complexity exceeds the quantum sample complexity. 

In fact, given that our covert verifiable algorithms have a polynomial complexity overhead over the quantum sample complexity of solving the problem from directly available quantum data, we need a problem in which the classical-vs-quantum separation is larger than this polynomial, otherwise the covert verifiable learner does not gain a meaningful advantage through querying the public quantum oracle.

In this section, we thus apply our results to two problems that famously provide an exponential quantum advantage in the black-box model: the Forrelation problem and Simon's problem.

\subsubsection{Covert Verifiable Forrelation}

Forrelation is the following task defined on two Boolean functions $f, g$ that exponential separates classical membership queries and quantum phase state queries. 

\begin{definition}[Forrelation \cite{aaronson2010bqp, aaronson2015forrelation}]
    The amount of ``forrelation'' between two Boolean functions $f,g:\{0,1\}^n\to\{0,1\}$ is defined as
    \begin{equation}
        \Phi(f,g)
        \coloneqq \frac{1}{2^{3n/2}} \sum_{x,y\in\{0,1\}^n} (-1)^{f(x)+x\cdot y + g(y)}\, .
    \end{equation}
    In the Forrelation problem, we are given oracle access to $f$ and $g$, and the task is to decide whether $(i) ~\lvert \Phi(f,g)\rvert\leq \frac{1}{100}$ or $(ii) ~\Phi(f,g)\geq \nicefrac{3}{5}$, promised that one of these two is the case.
\end{definition}

Quantum algorithms and classical lower bounds for Forrelation typically assume independent access to classical membership oracles and to quantum phase state oracles for $f$ and $g$. To fit this into the framework of this section (where we learn a property of a single function), define $h:\{0,1\}^{2n}\to\{0,1\}$ by $h(x,y)=f(x)\oplus g(y)$. Then $\ket{\psi_h^{\mathrm{Ph}}} = \ket{\psi_f^{\mathrm{Ph}}}\otimes\ket{\psi_g^{\mathrm{Ph}}}$, so a query to $\mathsf{O}_{\mathrm{pub}}^{\mathrm{QPh}}(h)$ can be simulated by one query each to $\mathsf{O}_{\mathrm{pub}}^{\mathrm{QPh}}(f)$ and $\mathsf{O}_{\mathrm{pub}}^{\mathrm{QPh}}(g)$. 
This reduction to a single function is not strictly necessary: one may covertly and verifiably obtain phase states for $f$ and $g$ separately and apply versions of  \Cref{corollary:verifiable-binary-distinguishing-from-phase-states,corollary:verifiable-binary-distinguishing-from-phase-states-ancilla-free} that work on algorithms that take phase states of multiple functions as inputs. We adopt the one-function formulation here only for simplicity of presentation and consistency with our earlier formulations.

Now let us recall the existing classical lower bounds, rephrased in our notation: 

\begin{theorem}[{\cite[Classical membership query lower bound for Forrelation]{aaronson2015forrelation}}] Any algorithm that solves Forrelation with bounded-error probability with only classical membership query access to $h$ requires $\Omega\left(\frac{2^{n/2}}{n}\right)$ queries to the oracle.
\end{theorem}

This follows from the standard formulation of the lower bounds (with access to $f$ and $g$ rather than $h$) because we can simulate $\mathsf{O}_{\mathrm{pri}}^{\mathrm{Mem}}(h)$ using $\mathsf{O}_{\mathrm{pri}}^{\mathrm{Mem}}(f)$ and $\mathsf{O}_{\mathrm{pri}}^{\mathrm{Mem}}(g)$ so the lower bounds for $f$ and $g$ port to $h$. 
On the quantum side, however, a constant number of phase state copies suffices: 

\begin{theorem}[{\cite[Quantum phase state algorithm for Forrelation]{aaronson2015forrelation}}] \label{thm:forrelation_upper_bound} A quantum algorithm can solve Forrelation with error probability $<0.25$ with $\Theta(1)$ number of copies of the phase states $\ket{\psi_h^{\mathrm{Ph}}}$.
\end{theorem}

This follows by amplifying the one-copy algorithm which gives error probability $0.4$ by repeating it $\Theta(1)$ times and taking a majority vote. Moreover, the algorithm requires doing a $\ComplexityFont{SWAP}$ test between $\ket{\psi_f^{\mathrm{Ph}}}$ and $\ket{\psi_g^{\mathrm{Ph}}}$ after applying $H^{\otimes n}$ on the latter. This can be implemented given access to $\ket{\psi_h^{\mathrm{Ph}}}$ by exploiting the tensor-product structure of $\ket{\psi_h^{\mathrm{Ph}}}$. 

Using the above theorem, we can now use \Cref{remark:robustness-assumption} to obtain a robust version of the algorithm and directly apply \Cref{corollary:verifiable-binary-distinguishing-from-phase-states} and the amplified version of \Cref{corollary:verifiable-binary-distinguishing-from-phase-states-ancilla-free} to obtain covert verifiable protocols for solving Forrelation against general unidirectional adversaries and i.i.d.\ ancilla-free adversaries: 

\begin{corollary}[Covert Verifiable Forrelation against Unidirectional Adversaries---Formal version of \Cref{inf-corollary:covert-forrelation-v1}]\label{corollary:verifiable-forrelation-V1}
Let $f,g:\{0,1\}^n\to\{0,1\}$ and define $h:\{0,1\}^{2n}\to\{0,1\}$ by $h(x,y)=f(x)\oplus g(y)$. 
Fix a confidence parameter $\delta\in(0,1)$. 
Then there exists a quantum algorithm $A'$ for deciding Forrelation on $(f,g)$ that makes at most
$m_{\mathrm{pub}}=\Tilde{\mathcal{O}}\bigl(n^5\log(1/\delta)\bigr)$ queries to $\mathsf{O}_{\mathrm{pub}}^{\mathrm{QPh}}(h)$ and $m_{\mathrm{pri}}=\Tilde{\mathcal{O}}\bigl(n^2\log(1/\delta)\bigr)$
queries to $\mathsf{O}_{\mathrm{pri}}^{\mathrm{Mem}}(h)$ 
and satisfies:
\begin{itemize}
    \item \textbf{Completeness:} With no adversary, for any $(f, g)$, after making at most $m_{\mathrm{pri}}$ queries to $\mathsf{O}_{\mathrm{pri}}^{\mathrm{Mem}}(h)$ and at most $m_{\mathrm{pub}}$ queries to $\mathsf{O}_{\mathrm{pub}}^{\mathrm{QPh}}(h)$,
    \begin{equation}
   \Pr\left[A'^{\,\mathsf{O}_{\mathrm{pri}}^{\mathrm{Mem}}(h),\,\mathsf{O}_{\mathrm{pub}}^{\mathrm{QPh}}(h)} 
    \text{ accepts and outputs the valid solution}\right]\ge1-\delta.
    \end{equation}
    \item \textbf{Soundness:} For any $(f, g)$, after making at most $m_{\mathrm{pri}}$ queries to $\mathsf{O}_{\mathrm{pri}}^{\mathrm{Mem}}(h)$ and at most $m_{\mathrm{pub}}$ queries to $\mathsf{O}_{\mathrm{pub}}^{\mathrm{QPh}}(h)$, the latter of which are subject to corruption by an arbitrary adversary,  
    \begin{equation}
    \Pr\left[A'^{\,\mathsf{O}_{\mathrm{pri}}^{\mathrm{Mem}}(h),\,\mathsf{O}_{\mathrm{pub}}^{\mathrm{QPh}}(h)}
    \text{ accepts but outputs the invalid solution} \right]\le\delta.
    \end{equation}
    \item \textbf{Privacy:}
Let $S_{(i)}=\{(f,g):|\Phi(f,g)|\le 1/100\}$ and $S_{(ii)}=\{(f,g):\Phi(f,g)\ge 3/5\}$. Suppose $(F, G)$ are sampled uniformly at random from $S_{(i)}$ with probability $1/2$ and uniformly at random from $S_{(ii)}$ with probability $1/2$. Then, after interacting with the at most $m_{\mathrm{pub}}$ public queries made by $A'$, a unidirectional adversary’s register is independent of $(F,G)$:
\begin{equation}
\rho_{\mathsf{FG}\mathsf{A}}
= \frac{1}{2}
\Bigl(
\frac{1}{|S_{(i)}|}\sum_{(f,g)\in S_{(i)}}\ket{f,g}\bra{f,g}
+ \frac{1}{|S_{(ii)}|}\sum_{(f,g)\in S_{(ii)}}\ket{f,g}\bra{f,g}
\Bigr)\otimes \rho_{\mathsf{A}}.
\end{equation}
In particular, the adversary gains no information about $(F,G)$.

\item \textbf{Efficiency.} $A'$ runs in time $\mathcal{O}\bigl(\poly(n,\log(1/\delta))\bigr)$ and, in addition to the operations used by the base Forrelation procedure, A' employs only uniform superposition preparation, non–adaptive single–qubit gates and CZ gates.
\end{itemize}
\end{corollary}

\begin{corollary}[Covert Verifiable Forrelation against i.i.d.\ Ancilla-free Adversaries---Formal version of \Cref{inf-corollary:covert-forrelation-v2}]\label{corollary:verifiable-forrelation-V2}
    Let $f,g:\{0,1\}^n\to\{0,1\}$ and define $h:\{0,1\}^{2n}\to\{0,1\}$ by $h(x,y)=f(x)\oplus g(y)$. 
Fix a confidence parameter $\delta\in(0,1)$.
For $\delta_{\mathrm{leak}}$, set $\varepsilon_{\mathrm{leak}} \coloneqq 1 - (1 - \frac{\delta_{\mathrm{leak}}}{2})^{\Theta(1)}$. Then there exists a quantum algorithm $A'$ for deciding Forrelation on $(f,g)$ that makes at most
$m_{\mathrm{pub}}=\Tilde{\mathcal{O}}\left(\frac{n\log^2(1/\delta)}{\varepsilon_{\mathrm{leak}}}\right)$ queries to $\mathsf{O}_{\mathrm{pub}}^{\mathrm{QPh}}(h)$ and $m_{\mathrm{pri}}=\Tilde{\mathcal{O}}\left(\frac{n\log^2(1/\delta)}{\varepsilon_{\mathrm{leak}}}\right)$
queries to $\mathsf{O}_{\mathrm{pri}}^{\mathrm{Mem}}(h)$ 
and satisfies:
    \begin{itemize}
        \item \textbf{Completeness:}  With no adversary, and for any $(f, g)$, after making at most $m_{\mathrm{pri}}$ queries to $\mathsf{O}_{\mathrm{pri}}^{\mathrm{Mem}}(h)$ and at most $m_{\mathrm{pub}}$ queries to $\mathsf{O}_{\mathrm{pub}}^{\mathrm{QPh}}(h)$, 
        \begin{equation}
            \Pr\left[A'^{\mathsf{O}_{\mathrm{pri}}^{\mathrm{Mem}}(h), \mathsf{O}_{\mathrm{pub}}^{\mathrm{QPh}}(h)}\text{ accepts and outputs a valid solution}\right]
            \geq 1-\delta\, .
        \end{equation}
        \item \textbf{Soundness:} For any $(f, g)$, after making at most $m_{\mathrm{pri}}$ queries to $\mathsf{O}_{\mathrm{pri}}^{\mathrm{Mem}}(h)$ and at most $m_{\mathrm{pub}}$ queries to $\mathsf{O}_{\mathrm{pub}}^{\mathrm{QPh}}(h)$, the latter of which are subject to corruption by an arbitrary i.i.d.\ adversary,  
        \begin{equation}
            \Pr\left[A'^{\mathsf{O}_{\mathrm{pri}}^{\mathrm{Mem}}(h), \mathsf{O}_{\mathrm{pub}}^{\mathrm{QPh}}(h)}\text{ accepts and outputs an invalid solution}\right]
            \leq \delta\, .
        \end{equation}
        \item \textbf{{Privacy}:} Let $S_{(i)}=\{(f,g):|\Phi(f,g)|\le 1/100\}$ and $S_{(ii)}=\{(f,g):\Phi(f,g)\ge 3/5\}$. Suppose $(F, G)$ are sampled uniformly at random from $S_{(i)}$ with probability $1/2$ and uniformly at random from $S_{(ii)}$ with probability $1/2$. Then for any i.i.d.\,ancilla-free adversary that extracts information about the unknown functions $(F, G)$ with probability at least $\delta_{\mathrm{leak}}$, A' accepts with probability $ \leq \delta$. Here, the probability $\delta$ is over the learner’s internal randomness (and measurements), whereas $\delta_{\mathrm{leak}}$ is over the adversary’s randomness.
        \item \textbf{Efficiency:} $A'$ runs in time $\mathcal{O}(\poly(n, \log(1/\delta), 1/\varepsilon_{\mathrm{leak}}))$ and, in addition to the operations used by the base Forrelation procedure, $A'$ only employs uniform superposition state preparation, controlled-$Z$ gates and adaptive single-qubit gates.
    \end{itemize}
\end{corollary}

While the covert verifiable learners in \Cref{corollary:verifiable-forrelation-V1,corollary:verifiable-forrelation-V2} make polynomially many oracle queries and thus more than needed if they were given direct, private access to a phase oracle (or phase state copies), these query complexities still vastly improve upon the exponential lower bound that applies when using only classical queries.

\subsubsection{Covert Verifiable Simon's Algorithm}\label{subsubsec:covert-simon}

Let us begin by recalling (the distinguishing version of) Simon's problem \cite{simon1997power}:
Given access to an unknown function $f:\{0,1\}^n\to\{0,1\}^n$, decide, with success probability $2/3$, whether (i) $f$ is $1$-to-$1$ (aka has period $0^n$) or (ii) $f$ is $s$-periodic for some $s\in\{0,1\}^n\setminus\{0^n\}$ (and thus in particular is $2$-to-$1$), promised that either (i) or (ii) is the case.
When given only access to $\mathsf{O}^{\mathrm{Mem}}(f)$, solving this problem requires $\Omega(2^{n/2})$ oracle queries.
In contrast, Simon gave a quantum algorithm that solves this problem using only $\mathcal{O}(n)$ queries to $\mathsf{O}^{\mathrm{QMem}}(f)$ and two additional queries to $\mathsf{O}^{\mathrm{Mem}}(f)$.

Simon's algorithm can be summarized as follows: First, use the $\mathsf{O}^{\mathrm{QMem}}(f)$-access to prepare copies of the quantum example state $\ket{\psi_f^{\mathrm{Ex}}}$. Next, for each example, apply Hadamard gates on the first $n$ qubits and measure, receiving as outcome an $n$-bit string orthogonal to the true period $s$ of the function $f$. Hence, $\mathcal{O}(n)$ quantum examples suffice to obtain, with success probability $2/3$, $n-1$ linearly independent $n$-bit strings orthogonal to $s$. We can now solve this linear system to produce a non-zero candidate period $s'$. Finally, we use two queries to $\mathsf{O}^{\mathrm{Mem}}(f)$ to check whether $f(s')=f(0^n)$. If yes, we are in case (ii); if no, we are in case (i). 
Thus, Simon's algorithm uses its $\mathsf{O}^{\mathrm{QMem}}(f)$ precisely to prepare copies of the quantum example state $\ket{\psi_f^{\mathrm{Ex}}}$.

As with Forrelation, we can now use \Cref{remark:robustness-assumption} to obtain a robust version of the algorithm and directly apply \Cref{corollary:verifiable-binary-distinguishing-from-phase-states} and the amplified version of \Cref{corollary:verifiable-binary-distinguishing-from-phase-states-ancilla-free} (carried over from $\mathsf{O}^{\mathrm{QPh}}_{\mathrm{pub}}$- and $\mathsf{O}^{\mathrm{Mem}}_{\mathrm{pri}}$-access to $\mathsf{O}^{\mathrm{QMem}}_{\mathrm{pub}}$- and $\mathsf{O}^{\mathrm{Mem}}_{\mathrm{pri}}$-access via \Cref{remark:covert-quantum-examples-from-QMem}).
This yields covert verifiable protocols for solving Simon's problem against general unidirectional adversaries and against i.i.d.\ ancilla-free adversaries: 

\begin{corollary}[Covert Verifiable Simon's Algorithm against Unidirectional Adversaries]\label{corollary:verifiable-simons-V1}
Let $f:\{0,1\}^n\to\{0,1\}^n$. 
Fix a confidence parameter $\delta\in(0,1)$. 
Then there exists a quantum algorithm $A'$ for solving Simon's problem for $f$ that makes at most
$m_{\mathrm{pub}}=\Tilde{\mathcal{O}}\bigl(n^{11}\log(1/\delta)\bigr)$ queries to $\mathsf{O}_{\mathrm{pub}}^{\mathrm{QMem}}(f)$ 
and $m_{\mathrm{pri}}=\Tilde{\mathcal{O}}\bigl(n^5\log(1/\delta)\bigr)$
queries to $\mathsf{O}_{\mathrm{pri}}^{\mathrm{Mem}}(f)$ 
and satisfies:
\begin{itemize}
    \item \textbf{Completeness:} With no adversary, for any $f$, after making at most $m_{\mathrm{pri}}$ queries to $\mathsf{O}_{\mathrm{pri}}^{\mathrm{Mem}}(\tilde{f})$ and at most $m_{\mathrm{pub}}$ queries to $\mathsf{O}_{\mathrm{pub}}^{\mathrm{QMem}}(f)$, 
    \begin{equation}
   \Pr\left[A'^{\,\mathsf{O}_{\mathrm{pri}}^{\mathrm{Mem}}(f) 
   ,\, \mathsf{O}_{\mathrm{pub}}^{\mathrm{QMem}}(f) 
   } 
    \text{ accepts and outputs the valid solution}\right]\ge1-\delta.
    \end{equation}
    \item \textbf{Soundness:} For any $f$, after making at most $m_{\mathrm{pri}}$ queries to $\mathsf{O}_{\mathrm{pri}}^{\mathrm{Mem}}(f)$ 
    and at most $m_{\mathrm{pub}}$ queries to $\mathsf{O}_{\mathrm{pub}}^{\mathrm{QMem}}(f)$, 
    the latter of which are subject to corruption by an arbitrary adversary,  
    \begin{equation}
    \Pr\left[A'^{\,\mathsf{O}_{\mathrm{pri}}^{\mathrm{Mem}}(f) 
    ,\,\mathsf{O}_{\mathrm{pub}}^{\mathrm{QMem}}(f) 
   } 
    \text{ accepts but outputs the invalid solution} \right]\le\delta.
    \end{equation}
    \item \textbf{Privacy:}
Let $S_{(i)}=\{f:f ~\text{is 1-to-1}\}$ and $S_{(ii)}=\{f:f ~\text{is }s\text{-periodic for some }s\neq 0^n\}$. Suppose $F$ is sampled uniformly at random from $S_{(i)}$ with probability $1/2$ and uniformly at random from $S_{(ii)}$ with probability $1/2$. Then, after interacting with the at most $m_{\mathrm{pub}}$ public queries made by $A'$, a unidirectional adversary’s register is independent of $F$:
\begin{equation}
\rho_{\mathsf{F}\mathsf{A}}
= \frac{1}{2}
\Bigl(
\frac{1}{|S_{(i)}|}\sum_{)\in S_{(i)}}\ket{f}\bra{f}
+ \frac{1}{|S_{(ii)}|}\sum_{f\in S_{(ii)}}\ket{f}\bra{f}
\Bigr)\otimes \rho_{\mathsf{A}}.
\end{equation}
In particular, the adversary gains no information about $F$.

\item \textbf{Efficiency.} $A'$ runs in time $\mathcal{O}\bigl(\poly(n,\log(1/\delta))\bigr)$ and, in addition to the operations used by the base Simon's algorithm, A' employs only uniform superposition preparation, non–adaptive single–qubit gates and CZ gates.
\end{itemize}
\end{corollary}

\begin{corollary}[Covert Verifiable Simon's Algorithm against i.i.d.\ Ancilla-free Adversaries]\label{corollary:verifiable-simons-V2}
    Let $f:\{0,1\}^n\to\{0,1\}^n$. 
Fix a confidence parameter $\delta\in(0,1)$.
For $\delta_{\mathrm{leak}}$, set $\varepsilon_{\mathrm{leak}} \coloneqq 1 - (1 - \frac{\delta_{\mathrm{leak}}}{2})^{\Theta(n)}$. Then there exists a quantum algorithm $A'$ for Simon's problem on $f$ that makes at most
$m_{\mathrm{pub}}=\Tilde{\mathcal{O}}\left(\frac{n^3\log^2(1/\delta)}{\varepsilon_{\mathrm{leak}}}\right)$ queries to $\mathsf{O}_{\mathrm{pub}}^{\mathrm{QMem}}(f)$ 
and $m_{\mathrm{pri}}=\Tilde{\mathcal{O}}\left(\frac{n^3\log^2(1/\delta)}{\varepsilon_{\mathrm{leak}}}\right)$
queries to $\mathsf{O}_{\mathrm{pri}}^{\mathrm{Mem}}(f)$ 
and satisfies:
    \begin{itemize}
        \item \textbf{Completeness:}  With no adversary, and for any $f$, after making at most $m_{\mathrm{pri}}$ queries to $\mathsf{O}_{\mathrm{pri}}^{\mathrm{Mem}}(f)$ 
        and at most $m_{\mathrm{pub}}$ queries to $\mathsf{O}_{\mathrm{pub}}^{\mathrm{QMem}}(f)$, 
        \begin{equation}
            \Pr\left[A'^{\mathsf{O}_{\mathrm{pri}}^{\mathrm{Mem}}(f)
            , \mathsf{O}_{\mathrm{pub}}^{\mathrm{QMem}}(f) 
            }\text{ accepts and outputs a valid solution}\right]
            \geq 1-\delta\, .
        \end{equation}
        \item \textbf{Soundness:} For any $f$, after making at most $m_{\mathrm{pri}}$ queries to $\mathsf{O}_{\mathrm{pri}}^{\mathrm{Mem}}(f)$ 
        and at most $m_{\mathrm{pub}}$ queries to $\mathsf{O}_{\mathrm{pub}}^{\mathrm{QMem}}(f)$, 
        the latter of which are subject to corruption by an arbitrary i.i.d.\ adversary,  
        \begin{equation}
            \Pr\left[A'^{\mathsf{O}_{\mathrm{pri}}^{\mathrm{Mem}}(f) 
            , \mathsf{O}_{\mathrm{pub}}^{\mathrm{QMem}}(f) 
            }\text{ accepts and outputs an invalid solution}\right]
            \leq \delta\, .
        \end{equation}
        \item \textbf{{Privacy}:} Let $S_{(i)}=\{f:f ~\text{is 1-to-1}\}$ and $S_{(ii)}=\{f:f ~\text{is }s\text{-periodic for some }s\neq 0^n\}$. Suppose $F$ is sampled uniformly at random from $S_{(i)}$ with probability $1/2$ and uniformly at random from $S_{(ii)}$ with probability $1/2$. Then for any i.i.d.\,ancilla-free adversary that extracts information about the unknown function $F$ with probability at least $\delta_{\mathrm{leak}}$, A' accepts with probability $ \leq \delta$. Here, the probability $\delta$ is over the learner’s internal randomness (and measurements), whereas $\delta_{\mathrm{leak}}$ is over the adversary’s randomness.
        \item \textbf{Efficiency:} $A'$ runs in time $\mathcal{O}(\poly(n, \log(1/\delta), 1/\varepsilon_{\mathrm{leak}}))$ and, in addition to the operations used by the base Simon's algorithm, $A'$ only employs uniform superposition state preparation, controlled-$Z$ gates and adaptive single-qubit gates.
    \end{itemize}
\end{corollary}

Again the covert verifiable learners in \Cref{corollary:verifiable-simons-V1,corollary:verifiable-simons-V2} use more oracle queries than Simon's algorithm with direct, trusted access to a quantum query oracle. However, these polynomial query complexities still significantly outperform the exponential classical query complexity lower bound.

\section*{Acknowledgments}

MCC is grateful to Peter Brown, Andrea Coladangelo, Mina Doosti, Alexander Nietner, and Ryan Sweke for insightful discussions. 
Part of this research was done while MCC was a DAAD PRIME Fellow at Caltech and at FU Berlin.
Part of this research was done while AK was a postdoc at Harvard University.

\newpage

\printbibliography
\sloppy

\newpage

\appendix

\end{document}

%% file: usepackages.tex
\usepackage[utf8]{inputenc}
\usepackage{lipsum}
\usepackage[table,xcdraw]{xcolor}
\usepackage{amsmath, amssymb,amsfonts,amsthm,mathtools,nicefrac}
\usepackage{xspace,graphicx,relsize,bm}
\usepackage{soul} 
\usepackage{enumitem}

\usepackage{parskip}  
\usepackage{comment}

\usepackage{libertinus}  
\usepackage[zerostyle=a,scaled=.97]{newtxtt}  
\usepackage[T1]{fontenc}
\usepackage[libertine]{newtxmath}  
\usepackage{dsfont}  
\usepackage[cal=stix, calscaled=.97]{mathalpha}  
\usepackage{anyfontsize}

\usepackage{url}
\usepackage[
    backend=biber,
    style=alphabetic,
    sorting=ynt,
    backref=true,
    maxbibnames=99
    ]{biblatex}
\DefineBibliographyStrings{english}{%
  backrefpage = {page},
  backrefpages = {pages},
  mathesis = {MSc thesis},
}

\usepackage[margin=1in]{geometry}
\usepackage[singlespacing]{setspace}
\definecolor{linkcol}{rgb}{0.0,0.55,0.7}
\definecolor{citecol}{rgb}{0.0, 0.6, 0.45}
\definecolor{urlcol}{rgb}{0.7, 0.0, 0.55}
\usepackage[linktocpage=true]{hyperref}
\hypersetup{
	colorlinks,
	linkcolor={linkcol},
	citecolor={citecol},
	urlcolor={urlcol}
}

\usepackage{subcaption}
\usepackage{mleftright}
\usepackage{tipa}
\usepackage{multirow}
\usepackage{authblk}  

\usepackage{subfiles} 

\usepackage{algorithm}
\usepackage{algpseudocode}

\usepackage{multirow}

\usepackage{pifont}

\usepackage{float}

\usepackage{chngpage}

\usepackage[capitalize,noabbrev]{cleveref}  

\usepackage{longtable}
\usepackage{lscape}

\usepackage{algorithm}
\usepackage{algpseudocode}
\usepackage{complexity}

%% file: commands.tex
\def\01{\{0,1\}}

\let\Pr\relax
\DeclareMathOperator*{\Pr}{\mathbb{P}}

\DeclarePairedDelimiterX\ceil[1]{\lceil}{\rceil}{
    \IfBlankTF{#1}{\thinspace\cdot\thinspace}{#1}
}
\DeclarePairedDelimiterX\floor[1]{\lfloor}{\rfloor}{
    \IfBlankTF{#1}{\thinspace\cdot\thinspace}{#1}
}
\DeclarePairedDelimiterX\norm[1]{\lVert}{\rVert}{
    \IfBlankTF{#1}{\thinspace\cdot\thinspace}{#1}
}

\DeclarePairedDelimiterX\abs[1]{\lvert}{\rvert}{
    \IfBlankTF{#1}{\thinspace\cdot\thinspace}{#1}
}
\DeclarePairedDelimiterX\ket[1]{\lvert}{\rangle}{
    \IfBlankTF{#1}{\psi}{#1}
}

\DeclarePairedDelimiterX\bra[1]{\langle}{\rvert}{
    \IfBlankTF{#1}{\psi}{#1}
}

\DeclarePairedDelimiterX\expval[1]{\langle}{\rangle}{
    \IfBlankTF{#1}{\thinspace\cdot\thinspace}{#1}
}
\DeclarePairedDelimiterX\braket[2]{\langle}{\rangle}{
    #1\,\delimsize\vert\,\mathopen{}#2
}

\DeclarePairedDelimiterX\mel[3]{\langle}{\rangle}{
    #1\delimsize\rvert\,\mathopen{}#2\,\delimsize\lvert\mathopen{}#3
}

\DeclarePairedDelimiterX\ketbra[2]{\lvert}{\rvert}{
    #1\delimsize\rangle\negthinspace\delimsize\langle\mathopen{}#2
}
\DeclarePairedDelimiterX\proj[1]{\lvert}{\rvert}{
    \IfBlankTF{#1}{
        \psi\delimsize\rangle\negthinspace\delimsize\langle\mathopen{}\psi
    }{
        #1\delimsize\rangle\negthinspace\delimsize\langle\mathopen{}#1
    }        
}

\DeclarePairedDelimiterX\Set[1]\{\}{%

#1
}

\newcommand{\tr}{\operatorname{tr}}

%% file: mytheorems.tex
\newtheoremstyle{mydefinitionsty}
{10pt}
{10pt}
{}
{}
{}
{}
{.5em}
{\textbf{\thmname{#1}~\thmnumber{#2}:  }\thmnote{(#3)}}
\theoremstyle{mydefinitionsty}
\newtheorem{definition}{Definition}
\newtheorem{remark}{Remark}

\newtheorem{observation}{Observation}

\newtheoremstyle{mythmsty}
{10pt}
{10pt}
{\itshape}
{}
{}
{}
{.5em}
{\textbf{\thmname{#1}~\thmnumber{#2}:  }\thmnote{(#3)}}
\theoremstyle{mythmsty}

\newtheorem{theorem}{Theorem}
\newtheorem{lemma}{Lemma}
\newtheorem{problem}{Problem}
\newtheorem{corollary}{Corollary}

\newtheorem{proposition}{Proposition}